\documentclass[11pt]{article}
\usepackage[T1]{fontenc}
\usepackage{amsfonts,amsmath,amsthm,amssymb,dsfont,mathtools}
\usepackage[margin=1in]{geometry}
\usepackage{fullpage}
\usepackage{enumitem}
\usepackage[linesnumbered,boxed,ruled,vlined]{algorithm2e}
\usepackage[colorlinks,citecolor=blue,linkcolor=blue,urlcolor=red,pagebackref]{hyperref}
\usepackage[capitalise]{cleveref}
\usepackage{url}
\usepackage{thmtools, thm-restate}
\usepackage{xcolor}
\usepackage{graphicx}
\usepackage{mathdots}
\usepackage[normalem]{ulem}
\usepackage{xspace}
\xspaceaddexceptions{]\}}
\usepackage{comment}
\usepackage{tabu}
\usepackage{framed}
\usepackage{float,wrapfig}
\usepackage{subfigure}
\usepackage{tikz}

\usepackage[mathlines]{lineno}
\usepackage{etoolbox}
\newcommand*\linenomathpatch[1]{%
  \cspreto{#1}{\linenomath}%
  \cspreto{#1*}{\linenomath}%
  \csappto{end#1}{\endlinenomath}%
  \csappto{end#1*}{\endlinenomath}%
}
\linenomathpatch{equation}
\linenomathpatch{gather}
\linenomathpatch{multline}
\linenomathpatch{align}
\linenomathpatch{alignat}
\linenomathpatch{flalign}

\theoremstyle{plain}
\newtheorem{theorem}{Theorem}[section]
\newtheorem{lemma}[theorem]{Lemma}

\newtheorem{corollary}[theorem]{Corollary}
\newtheorem{claim}[theorem]{Claim}
\newtheorem{observation}[theorem]{Observation}

\newtheorem{assumption}{Assumption}

\theoremstyle{definition}
\newtheorem{definition}[theorem]{Definition}
\newtheorem{remark}[theorem]{Remark}

\setlist[enumerate]{nosep, topsep=1ex}
\setlist[itemize]{nosep, topsep=1ex}
\setlist[description]{nosep}
\allowdisplaybreaks
\def\ShowAuthNotes{1}
\ifnum\ShowAuthNotes=1
\newcommand{\authnote}[2]{\ \\ \textcolor{red}{\parbox{0.9\linewidth}{[{\footnotesize {\bf #1:} { {#2}}}]}}\newline}
\else
\newcommand{\authnote}[2]{}
\fi

\newcommand{\eps}{\varepsilon}
\renewcommand{\Pr}{\operatorname*{\mathbf{Pr}}}
\newcommand{\Ex}{\operatorname*{\mathbf{E}}}
\newcommand{\Var}{\operatorname*{\mathbf{Var}}}

\newcommand{\poly}{\operatorname{\mathrm{poly}}}
\newcommand{\supp}{\operatorname{\mathrm{supp}}}
\newcommand{\polylog}{\poly\log}
\newcommand{\polyloglog}{\poly\log\log}
\newcommand{\F}{\mathbb{F}}
\newcommand{\R}{\mathbb{R}}
\newcommand{\N}{\mathbb{N}}
\newcommand{\Z}{\mathbb{Z}}
\newcommand{\C}{\mathbb{C}}
\renewcommand{\tilde}{\widetilde}
\newcommand{\dd}{\mathinner{.\,.\allowbreak}}

\newcommand{\caP}{\mathcal{P}}

\newcommand{\caX}{\mathcal{X}}

\def\tO{\widetilde{O}}
\newcommand{\linereff}[1]{\hyperref[#1]{Line~\ref*{#1}}}

\title{Shaving Logs via Large Sieve Inequality: \\Faster Algorithms for Sparse Convolution and More\thanks{Part of this work was done while the authors were visiting Simons Institute for the Theory of Computing.}}
\author{
Ce Jin\thanks{cejin@mit.edu. Supported by  NSF grants CCF-2129139
and CCF-2127597.}\\MIT \and Yinzhan Xu\thanks{xyzhan@mit.edu. Partially supported by NSF grants CCF-2330048 and BSF Grant 2020356.}\\ MIT
}
\date{\vspace{-1cm}}

\begin{document}

\maketitle
\thispagestyle{empty}
\begin{abstract}
In sparse convolution-type problems, a common technique is to hash the input integers modulo a random prime $p\in [Q/2,Q]$ for some parameter $Q$, which reduces the range of the input integers while preserving their additive structure. However, this hash family suffers from two drawbacks, which led to bottlenecks in many state-of-the-art algorithms: (1) The collision probability of two elements from $[N]$ is $O(\frac{\log N}{Q})$ rather than $O(\frac{1}{Q})$; (2) It is difficult to derandomize the choice of $p$; known derandomization techniques lead to super-logarithmic overhead [Chan, Lewenstein STOC'15]. 
	
In this paper, we partially overcome these drawbacks in certain scenarios, via novel applications of the \emph{large sieve inequality} from analytic number theory. Consequently, we obtain the following improved algorithms for various problems (in the standard word RAM model):
\begin{itemize}
    \item \textbf{Sparse Nonnegative Convolution:} We obtain an $O(t\log t)$-time Las Vegas algorithm that computes the convolution $A\star B$ of two nonnegative integer vectors $A,B$, where $t$ is the output sparsity $\|A\star B\|_0$.
    Moreover, our algorithm terminates in $O(t\log t)$ time with $1-1/\mathrm{poly}(t)$ probability. This simultaneously improves  the $O(t\log t \log \log t)$-time Las Vegas algorithm [Bringmann, Fischer, Nakos SODA'22] and the Monte Carlo $O(t\log t)$-time algorithm with failure probability $2^{-\sqrt{\log t}}$ [Bringmann, Fischer, Nakos STOC'21].
    \item \textbf{Text-to-Pattern Hamming Distances:} Given a length-$m$ pattern $P$ and a length-$n$ text $T$, we obtain a deterministic $O(n\sqrt{m\log \log m})$-time algorithm that exactly computes the Hamming distance between $P$ and every length-$m$ substring of $T$. This improves the previous $O(n\sqrt{m}(\log m\log \log m)^{1/4})$-time deterministic algorithm [Chan, Jin, Vassilevska Williams, Xu FOCS'23] and nearly matches their $O(n\sqrt{m})$-time Las Vegas algorithm. 
    \item \textbf{Sparse General Convolution:}  For sparse convolution with possibly negative input, all previous approaches required $\Omega(t\log ^2 t)$ time, where $t$ is the maximum of input and output sparsity, and an important question left open by [Bringmann, Fischer, Nakos STOC'21] is whether this can be improved. We make partial progress towards solving this question by giving a Monte Carlo $O(t\log t)$ time algorithm in the restricted case where the input vectors have length $N\le t^{1.99}$ and polynomially bounded entries. \end{itemize}
\end{abstract}

\newpage
\setcounter{page}{0}
\thispagestyle{empty}
\tableofcontents

\newpage

\section{Introduction}
\label{sec:intro}

In this paper we study various problems related to sparse pattern matching, such as Sparse Convolution, Text-to-Pattern Hamming Distances, and the Constellation problem (their definitions will be given later in this introduction). 

Many of these problems take sparse integer sets as input, and it is often useful to first reduce them to dense instances which are easier to handle. 
In our cases, a common tool to implement this sparse-to-dense reduction is the simple mod-prime hash. That is, we pick a random prime $p$ from $[Q / 2, Q]$ for some parameter $Q$, and map each input integer $x$ (coming from a large universe $[N]$\footnote{We use $[N]$ to denote $\{0, 1, \ldots, N - 1\}$. })
to $h(x) := x \bmod{p}$. One major benefit of this hash function is that it is perfectly linear, i.e., $h(x)+h(y)\equiv h(x+y)\pmod{p}$ for any $x, y$. 
This feature makes it useful in convolution and pattern matching scenarios, where we want to compress the input universe while preserving the additive structure.
However, this mod-prime hash family has two main drawbacks:
\begin{itemize}
  \item 
  For any $x \ne y \in [N]$, the probability that $h(x) = h(y)$ is $O(\frac{\log N}{Q})$.\footnote{$h(x) = h(y)$ if and only if $p$ divides $x - y$. The number of prime divisors of $x - y$ from $[Q / 2, Q]$ is at most $\log_{Q/2}|x-y| = O(\log N / \log Q)$, so a random prime from $[Q / 2, Q]$ is among them with probability $O\big (\frac{\log N / \log Q}{\#\text{primes in } [Q / 2, Q]}\big ) = O\big (\frac{\log N / \log Q}{Q / \log Q}\big ) = O\big (\frac{\log N}{Q}\big )$ by the prime number theorem. \label{footnote:prime}} In contrast, the optimal collision probability is $\frac{1}{Q}$ by a uniformly random hash function.
  This means that the hash family may have bad load-balancing that leads to logarithmic factor loss in the running times of the algorithms.
  
    An alternative is to use almost linear hash functions (such as Dietzfelbinger's hash function~\cite{Dietzfelbinger96} or the classic textbook hash function $h(x) = ((ax + b) \bmod{p}) \bmod{m}$), which have optimal (up to constant factor) collision probability. 
    However, these hash families do not \emph{perfectly} achieve the linearity condition, but rather achieve it with small errors. 
    These hard-to-predict errors often pose serious challenges in algorithm design (e.g., \cite{BringmannFN21,BringmannFN22}).
    \item This hash family is not easily applicable to deterministic algorithms, since it is hard to derandomize the choice of the prime $p$. 
The naive way that tries all primes in the given range is too slow. A smarter way of derandomization used in \cite{ChanL15,ChanHe,focs23,fischer3sum}
is to hash modulo $q=p_1\cdots p_k$, where each factor $p_i$ is small and is deterministically chosen by brute force. However, this multi-level approach leads to super-logarithmic overhead \cite{ChanL15}. 
(Some recent deterministic algorithms~\cite{BringmannFN22,nickconstellation} removed this super-logarithmic overhead by using algebraic techniques instead of hashing, but these algebraic techniques were tailored towards the specific applications, and it is not clear whether they are as widely applicable as hashing-based techniques.)
\end{itemize}

In this paper, we partially overcome these two drawbacks in certain scenarios. 
Our main tool is the large sieve inequality (and its variants)  from analytic number theory (described in more detail in \cref{subsec:analytic-number-theory}).
We use this tool to analyze the mod-prime hash function and obtain better bounds. This leads to improved hashing-based algorithms for many sparse convolution and pattern matching problems. 
To the best of our knowledge, we are the first to apply this number-theoretic tool in the context of algorithm design. %

Before discussing this number-theoretic tool, we first introduce the problems we consider in this paper, known results, and our improved results. 
We use the standard word RAM model of computation throughout this paper (see formal definitions in \cref{subsec:model}). 

\paragraph*{Sparse Convolution}
Computing the convolution $A\star B$ of two integer vectors $A,B$, defined as $(A\star B)_k = \sum_{i+j=k}A_i\cdot B_j$, is a fundamental task in computer science. It is useful in various fields of engineering: signal processing, computer vision, symbolic computation, discrete algorithm design, computational complexity theory, and so on.   
In most scenarios of discrete algorithm design, we care about \emph{nonnegative} convolution where $A,B$ are nonnegative vectors (sometimes even Boolean convolution suffices). Applications include $k$SUM \cite{ChanL15}, Subset Sum \cite{Bringmann17, KoiliarisX19, BringmannN20, icalpBringmannN21}, Knapsack \cite{bc23} and pattern matching problems \cite{fispat, Indyk98a}. %

It is well-known that convolution of two length-$N$ vectors can be computed in deterministic $O(N\log N)$ time using Fast Fourier Transform (FFT). This running time is widely conjectured to be optimal, and there is some partial evidence in favor of the conjecture \cite{Ailon13, AfshaniFKL19}.

In many algorithmic applications, the input and output vectors are often sparse (e.g., \cite{CardozeS98, ColeH02, AmirKP07, ChanL15, BringmannN20,BringmannN21, JinX23,AbboudBF23,ChanWX23,bc23,clmz24}), so it is important to have algorithms that can exploit the  sparsity of the input and output. Motivated by this, there has been a long line of work studying the Sparse Convolution problem \cite{Muthukrishnan95, ColeH02, roche2008adaptive,  MonaganP09, HoevenL12, ChanL15, ArnoldR15, Roche18, Nakos20, GiorgiGC20mult, BringmannFN21, BringmannFN22}. A major open question is whether we can solve Sparse Convolution as fast as dense convolution \cite{BringmannFN21}.

One interesting case of Sparse Convolution which is particularly useful for algorithmic applications is when the input integers are all nonnegative. This is known as the Sparse Nonnegative Convolution problem. More formally, let $A,B\in \N^{N}$ be two input vectors.
We denote the sparsity by $t = \max\{\|A\|_0, \|B\|_0, \|A\star B\|_0\}$. When $A, B$ are nonnegative vectors, $t = \|A\star B\|_0$. The Sparse Nonnegative Convolution problem asks to compute $A \star B$, and preferably the running time should mainly depend on the sparsity $t$ instead of the length $N$. 

Previously, the fastest algorithm for Sparse Nonnegative Convolution was given by Bringmann, Fischer and Nakos \cite{BringmannFN21}. Their algorithm  runs in $O(t \log t+ \polylog(N\Delta))$ time where $\Delta = \max\{\|A\|_\infty, \| B\|_\infty\}$ denotes the maximum entry size, is Monte Carlo randomized, and has error probability $2^{-\sqrt{\log t}}$. When $N\Delta$ is reasonably small, their running time is $O(t \log t)$, matching the running time of dense convolution. 

However, compared to dense convolution (which has a deterministic algorithm), one weakness of~\cite{BringmannFN21}'s result is that the success probability is only $1 - 2^{-\sqrt{\log t}}$. In algorithm design, it is often desirable to achieve $1 - 1/\poly(t)$ error probability, which is typically called ``with high probability'' in algorithmic literature. 

Resolving this issue was explicitly asked as the first open question in \cite{BringmannFN21}:
\begin{center}
\cite{BringmannFN21}: \emph{Can Sparse Nonnegative Convolution be solved by a randomized algorithm in $O(t\log t)$ time with $1-1/\poly(t)$ success probability?}
\end{center}
This issue also motivated the authors of \cite{BringmannFN21} to study Las Vegas and deterministic algorithms for Sparse Nonnegative Convolution in a follow-up work \cite{BringmannFN22}.\footnote{See the last paragraph on page 2 of \cite{BringmannFN22}.} 
However, \cite{BringmannFN22} only obtained an $O(t\log t\log \log t)$ time Las Vegas algorithm, and a much slower $O(t\log^5(N\Delta)\poly\log \log (N\Delta))$ time deterministic algorithm. \cite{BringmannFN22} raised the following open question:
\begin{center}
\cite{BringmannFN22}: \emph{Can Sparse Nonnegative  Convolution be solved by a Las Vegas randomized algorithm in $O(t\log t)$ expected time?}
\end{center}

In this paper, we resolve both open questions in the affirmative. In fact, we get the best of both worlds: We obtain a Las Vegas algorithm that achieves $O(t\log t)$ running time with high probability.

\begin{restatable}[Las Vegas Sparse Nonnegative Convolution w.h.p.]{theorem}{lasvegasmain}
\label{thm:lasvegas-main}
Given two vectors $A,B\in \N^N$, one can compute $A\star B$ by a Las Vegas algorithm that terminates in $O(t\log t)$ time with at least $1-1/t$ probability, where $t=\|A\star B\|_0$.
\end{restatable}

\cref{thm:lasvegas-main} also affirmatively resolves a secondary open question asked in \cite{BringmannFN21} on whether the $\polylog(N\Delta)$ term in their $O(t\log t+ \polylog(N\Delta))$-time Monte Carlo algorithm can be removed (see \cref{footnote:polylogndelta}).

As a straightforward application, we improve the best known Las Vegas algorithm for the Constellation problem, which was recently considered by Fischer \cite{nickconstellation} as a representative problem for a rich class of sparse pattern matching problems (with applications in computational geometry).
In this problem, we are given two integer sets $A,B\subseteq [N]$ of size $|A|,|B|\le n$, and we need to output all shifts $s$ such that $A+s\subseteq B$. The best Monte Carlo algorithm for this problem runs in $O(n \log n)$ time \cite{CardozeS98}, and the best Las Vegas algorithm runs in $O(n \log^2 N)$ time \cite{ColeH02}. Fischer~\cite{nickconstellation} gave an $O(n\polylog N)$ time deterministic algorithm, improving an $n \cdot 2^{O(\sqrt{\log n \log \log N})}$ time algorithm by Chan and Lewenstein \cite{ChanL15}. 

Our algorithm is a straightforward combination of the Monte Carlo algorithm by Cardoze and Schulman \cite{CardozeS98} and our Las Vegas Sparse Nonnegative Convolution algorithm.\footnote{In fact, our method shows that the previous $O(t \log t \log \log t)$ time Las Vegas   Sparse Nonnegative Convolution algorithm \cite{BringmannFN22} would already imply an $O(n \log n \log \log n)$ time Las Vegas algorithm for the Constellation problem, which is faster than the previously known $O(n \log^2 N)$ time Las Vegas algorithm \cite{ColeH02}.}

\begin{restatable}[Constellation]{theorem}{Constellation}
\label{thm:constellation-las-vegas}
  Given two sets $A,B\subseteq [N]$ of size $|A|,|B|\le n$, one can output all shifts $s$ such that $A+s\subseteq B$ by a Las Vegas algorithm that terminates in $O(n\log n)$ time with at least $1-1/n$ probability.
\end{restatable}
We prove \cref{thm:constellation-las-vegas} in \cref{sec:conste}.

\paragraph*{Sparse General Convolution.}
Our \cref{thm:lasvegas-main}, as well as the  previous works~\cite{ColeH02,ArnoldR15, BringmannFN21,BringmannFN22}, crucially requires that the input vectors are nonnegative. For Sparse General Convolution (with input vectors from $\Z^N$ rather than $\N^N$), various previous techniques are all stuck at $O(t\log^2 N)$  (see \cite[Section 3.3]{BringmannFN21}\footnote{When referencing the theorems and lemmas in the previous works by Bringmann, Fischer, and Nakos \cite{BringmannFN21,BringmannFN22}, we always refer to their long versions appearing on arXiv.} and \cite[Section 2.2.2]{nickphd} for related discussions). Getting faster algorithms for Sparse General Convolution is explicitly mentioned in \cite{BringmannFN21} as an open problem: 
\begin{center}
\cite{BringmannFN21}: \emph{Can Sparse General Convolution be solved in $O(t\log t + \polylog(N\Delta))$ time?}
\end{center}

In this paper, we give a partial answer to this question: For the special case of $N\le t^{1.99}$ and $\Delta \le 2^{2^{\log t/\poly\log \log t}}$, we obtain an $O(t\log t)$-time Monte Carlo randomized algorithm for general convolution. 

\begin{restatable}[Sparse General Convolution for short vectors]{theorem}{generalconvo}
\label{thm:generalconvo-main}
 Let $\eps>0$  be a fixed constant.
Given two vectors $A,B\in \Z^N$ with the promise that $N<t^{2-\eps}$, where $t=\max\{\|A\|_0, \|B\|_0, \|A\star B\|_0\}$,
one can compute $A\star B$ in $O(t\log t + t(\log\log t)^{O(1)}\log\log \Delta + \poly \log \log \Delta) $ expected time by a Monte Carlo algorithm with $0.9$ success probability, where $\Delta = \max\{\|A\|_\infty,\|B\|_\infty\}$.
\end{restatable}

The upper bound $2^{2^{\log t/\poly\log \log t}}$ on $\Delta$ is much higher than $\poly(t)$, which is the interesting case for most algorithmic applications. The constraint $N\le t^{1.99}$ is more restricted. In \cite{BringmannFN21}, they actually reduced Sparse Nonnegative  Convolution with general $N$ to Sparse Nonnegative  Convolution with a similar constraint $N \le \poly(t)$. Unfortunately, their reduction does not work for Sparse General Convolution. If one could design a similar reduction that works for Sparse General Convolution, then \cref{thm:generalconvo-main} would extend to general values of $N$. In light of such reductions, we hope \cref{thm:generalconvo-main} can eventually lead to an $O(t \log t)$ time algorithm for Sparse General Convolution.

\cref{thm:lasvegas-main} and \cref{thm:generalconvo-main} assume the word RAM model with word length $w = \Omega(\log t + \log(N\Delta))$ (see more details in \cref{subsec:model}).
An alternative choice is to consider the \emph{bit complexity} (which is the focus of papers in symbolic computation such as \cite{GiorgiGC20mult,GiorgiGCR22a}), which would be more realistic in the scenario where $\log(N\Delta)$ can be much larger than $t$. %
By taking into account the bit complexity of arithmetic operations on $\log t + \log (N\Delta)$-bit integers, we see that 
both  
\cref{thm:lasvegas-main} and \cref{thm:generalconvo-main} have bit complexity $O(t\log t \log (N\Delta)\poly \log\log (N\Delta))$.
In comparison, \cite{GiorgiGC20mult} gave a Monte Carlo algorithm for Sparse General Convolution in bit complexity $O(t\log(N\Delta) \cdot \poly(\log t,\log\log(N\Delta)))$ with unspecified number of polylogarithmic factors. 
Our algorithms 
achieve better dependency on $t$, but only apply to more restrictive scenarios than \cite{GiorgiGC20mult}.

\paragraph*{Text-to-Pattern Hamming Distances.}

In the classic \emph{Text-to-Pattern Hamming Distances} problem, we are given a pattern $P$ of length $m$ and a text $T$ of length $n$, both over a polynomial-size alphabet, and we need to compute the Hamming distance between $P$ and $T[i \mathinner{.\,.} i + m-1]$ for every shift $i$.

For a long time, the fastest algorithm for this problem had $O(n\sqrt{m\log m})$ deterministic time  by Abrahamson \cite{Abrahamson87}. 
Very recently, Chan, Jin, Vassilevska Williams and Xu \cite{focs23} gave a  Las Vegas algorithm achieving improved $O(n\sqrt{m})$ running time with high probability.
Although Text-to-Pattern Hamming Distances a priori is not a sparse problem,  the improvement of \cite{focs23} crucially used hashing-based techniques similar to those used in Sparse Convolution and sparse recovery.
Thus, it is difficult to derandomize their
 hashing-based algorithms with small overhead. 
Using a derandomization technique by Chan and Lewenstein \cite{ChanL15} (mentioned in the beginning of this introduction), they obtained a slower deterministic algorithm in $O(n\sqrt{m}(\log m\log \log m)^{1/4})$ time. 

In this work, we obtain an improved deterministic algorithm for the Text-to-Pattern Hamming Distances problem, matching the state-of-the-art randomized algorithms \cite{focs23} up to a $\sqrt{\log \log m}$ factor.

\begin{restatable}{theorem}{HammingDistances}
\label{thm:TtPHD}
       The Text-to-Pattern Hamming Distances problem can be solved by a deterministic algorithm in $O(n\sqrt{m\log \log m})$ time.
\end{restatable}

We also obtain a similar improvement for the Text-to-Pattern Dominance Matching problem (see definition in \cref{sec:texttopattern}).

The main technical component behind \cref{thm:TtPHD}, which may be of independent interest, is a variant of a lemma that computes the convolution between two nonnegative vectors that appeared in \cite{focs23}, which they call the $X+Y$ lemma.

\begin{restatable}[Deterministic $X+Y$ lemma]{theorem}{XPlusY}
\label{thm:x-plus-y-advice}
Given two nonnegative vectors $X, Y \in \N^N$ with $\|X\|_1, \|Y\|_1 \le s$ for some $s \ge \sqrt{2N}$, there is a nonuniform deterministic  algorithm  that can compute $X \star Y$ in $O(N \log(s^2 / N) + N \log \log N)$ time. 

The advice (a.k.a. non-uniformity) of the algorithm only depends on $N$ and has $O(\polylog N)$ bits. Moreover, this advice can be computed in $O(N^{1.01})$ time deterministically.
\end{restatable}

We can compare this result with \cite{focs23} who gave two algorithms for the $X+Y$ lemma. 
Their Las Vegas algorithm runs in $O(N \log(s^2 / N))$ expected time, and their deterministic algorithm runs in  $O(N \log(s^2 / N) + N \sqrt{\log s \log \log N})$ time. Their deterministic running time is about a $\sqrt{\log s} = \Omega(\sqrt{\log N})$ factor higher than ours. Although our algorithm additionally requires a small $O(N^{1.01})$ preprocessing time to compute the advice, it becomes negligible in applications where we need to apply the $X+Y$ lemma many times (such as the Text-to-Pattern Hamming Distances problem). 

\subsection{Technical overview}
\label{sec:overview}

\paragraph*{Mod-prime hash and large sieve inequality.} Suppose the input sparse vectors are from $\Z^N$. Let $p$ be a random prime from $[Q / 2, Q]$. The mod-prime hash is defined as $h(a) := a \bmod{p}$ for $a \in [N]$. In many applications of the mod-prime hash, one is interested in the probability that two hashes collide, i.e., $h(a) = h(b)$ for $a \ne b \in [N]$. Note that the number of prime factors in $[Q / 2, Q]$ of $a - b$ is upper bounded by $O(\frac{\log N}{\log Q})$, and  the total number of primes in $[Q / 2, Q]$ is $\Theta(\frac{Q}{\log Q})$, so the probability that a random prime can cause $h(a) = h(b)$, which is equivalent to $p \mid a - b$, is $O(\frac{\log N}{Q})$. 

This $O(\frac{\log N}{Q})$ probability is affected by the number of prime factors of $a - b$ in $[Q / 2, Q]$. In the worst-case, it is certainly possible that $a - b$ has $\Theta(\frac{\log N}{\log Q})$ prime factors in the range. However, on average, each number in $[N]$ only has $O(1 / \log Q)$ prime factors from $[Q / 2, Q]$.\footnote{This is because each prime $p \in [Q / 2, Q]$ is the divisor of only $O(N / Q)$ numbers in $[N]$, so on average, each number in $[N]$ has $O(\frac{1}{N} \cdot \frac{N}{Q} \cdot \frac{Q}{\log Q}) = O(1/\log Q)$ prime divisors from the range.} This suggests that, when $a - b$ is not selected in a worst-case way, it could potentially be possible to improve the $O(\frac{\log N}{Q})$ collision probability. 

In fact, when applying the mod-prime hash to a subset of numbers $A \subseteq [N]$, we are interested in the number of  pairs $(a, b)  \in A\times A$ such that $p \mid (a-b)$. If we apply mod-prime hashing naively, we get an $O(|A|^2 \cdot (\frac{\log N}{Q}) + |A|) = O(|A| \cdot (\frac{|A| \log N}{Q} + 1))$ upper bound (the $+|A|$ accounts for the fact that $a \in A$ always collides with itself). However, intuitively, it is difficult to construct $A$, so that most numbers $a-b$ where $(a,b)\in A\times A$ have the largest possible $\Theta(\frac{\log N}{\log Q})$ number of prime factors from $[Q / 2, Q]$. 

This is where the large sieve inequality plays a role in our paper. 
In number theory, the large sieve inequality has been used to prove arithmetic properties about sumsets \cite{balog1984sums,pomerance1988divisors,largesievesumsets}; this is somewhat similar to our goal of understanding the prime factors of elements in the difference set $A-A$.
A standard application of a version of the large sieve inequality \cite{wolke} formalizes the above intuition for a certain regime of parameters. More specifically, it can be shown that, as long as $10 \le Q \le N \le Q^{1.99}$, the expected number of  pairs $a, b  \in A$ such that $h(a) = h(b)$ can be bounded by $O\left(|A| \cdot \left(\frac{|A|}{Q} + \log \log N\right)\right)$. In the natural $Q = \Theta(|A|)$ setting, this bound achieves an $O(\log N / \log \log N)$ improvement over the naive $O(|A| \cdot (\frac{|A| \log N}{Q} + 1))$ bound. In fact, we can even decrease $Q$ to $|A| / \log \log N$, while maintaining the same bound,  and in this case, the $O\left(|A| \log \log N\right)$ bound matches the  bound achieved by the fully random hash function. Thus, for each $i \in [p]$ (we typically call such $i$ a bucket), if we count how many $a \in A$ are mapped to $i$, we should expect to see $O(\log \log N)$ of them, at least for most buckets. This improvement lies at the heart of our $O(t\log t)$-time Sparse Nonnegative  Convolution and Sparse General  Convolution algorithms.\footnote{This improvement alone leads to $O(t\log t)$ expected time Las Vegas algorithm. The high success probability of \cref{thm:lasvegas-main} is made possible by other techniques, which will be explained later.} 

The statements of the large sieve inequality and the formal version of the  above bound can be found in \cref{subsec:analytic-number-theory}.

\paragraph*{Sparse General Convolution.}
For the Sparse Convolution problem, we follow a series of previous works (e.g., \cite{Nakos20,BringmannFN21,BringmannFN22}) which used the iterative peeling algorithm (similar to techniques used in the sparse recovery literature) combined with linear hashing.
Roughly speaking, given input vectors $A,B$, we maintain an approximate answer $C$, and gradually decrease the number of errors, $\|A\star B - C\|_0$.  In each round, hash the non-zero entries of the difference vector $A\star B - C$ into buckets (where the number of buckets is comparable to the support size $\|A\star B - C\|_0$), and try to recover the entries in light buckets.
Here it is important that the hash function is linear, i.e., $h(x)+h(y) = h(x+y)$, so that it  respects the convolution operator $\star$.
In this way, each round becomes a dense convolution problem with length comparable to $\|A\star B - C\|_0$, and can be solved via FFT.
This is a common technique in designing algorithms for additive problems.

A common choice for the hash function is the mod-prime hash function, which satisfies perfect linearity $h(x)+h(y)\equiv h(x+y)\pmod{p}$. It suffers from an extra logarithmic factor in the collision probability, so using it in the iterative peeling algorithm only led to an $O(t\log^2 N)$ algorithm. See \cite[Section 2.2.2]{nickphd} for a detailed exposition of this algorithm; see also \cite{Nakos20}. Another choice is to use almost linear hash families, which have smaller collision probability. However, the unpredictability of the error of the hash function led to difficulties when implementing the iterative peeling algorithm, namely we do not know where to subtract the $C$ entries. See \cite[page 8]{BringmannFN21} for a discussion of this challenge. In \cite{BringmannFN21}, they overcame this challenge for  Sparse Nonnegative Convolution, but their techniques completely break down when the input vectors have possibly negative entries that might lead to cancellation, due to the errors in almost linear hashing; see discussion  in \cite[Section 3.3]{BringmannFN21}.

Now we explain how we use mod-prime hashing to obtain a faster algorithm for Sparse General Convolution in certain regimes.

Recall that in our algorithm for Sparse General Convolution, we assume the input vectors have length $N \le t^{1.99}$, that is, the size of the universe is sub-quadratically smaller than the number of elements we apply hash to. 
This perfectly meets the condition for applying the results from the large sieve inequality (i.e., $N \le Q^{1.99}$). Therefore, if we pick the prime $p$ for the mod-prime hash from the range $p = \Theta(t / \log \log N)$, we expect most buckets to have size $O(\log \log N)$. 

Then the high-level idea is to use Prony's method \cite{prony1795}, to recover all numbers in each bucket. This was suggested before in \cite{BringmannFN21}; however, without the improved bounds for mod-prime hashing from the large sieve inequality, this idea would be too slow. Prony's method is a way to recover a sparse polynomial with $s$ nonzero terms from $2s$ evaluations and its running time is \emph{roughly} $\poly(s)$ given the evaluation results. Now the idea is to view each bucket as a sparse polynomial, and use Prony's method to recover each bucket and we can set $s = O(\log \log N)$. 

We can compute these $2s$ evaluations for all buckets simultaneously, and each evaluation uses an FFT on array of size $O(p)$. Overall, these evaluations take $O(s p \log p) = O(t \log t)$ time. Running Prony's method for every bucket would roughly take time $O(t \polyloglog t)$, so the overall running time would be $O(t \log t)$. 

The above intuition ignores some details about the running time of Prony's method. As it turns out, the current best algorithm for Prony's method over finite fields is only fast enough over fields with relatively small characteristic. Therefore, we will choose some finite field with small characteristic to run Prony's method; this will make some nonzero entries in the output vector vanish, i.e., equivalent to $0$ in the finite field, but most entries will remain. Running Prony's method over such finite fields gives a large fraction of the support of the output vector, but does not give the exact values of the coefficients. Therefore, we need to add an additional step to recover the coefficients over a large prime field $\F_{p_b}$ where $p_b = \poly(N\Delta)$.\footnote{To generate such a prime $p_b= \poly(N\Delta)$, the fastest known way (allowing randomization) takes $\polylog(N\Delta)$ time, which may become a bottleneck in extreme cases where $\polylog(N\Delta)$ is much larger than $t$ (this is the same reason why the previous work of \cite{BringmannFN21} incurred extra $\polylog(N\Delta)$ time complexity).
In the main text of this paper, to simplify the presentation, we still work over $\F_{p_b}$ and incur this extra $\polylog(N\Delta)$ time. 
In \cref{sec:numerical}, we describe how to avoid spending $\polylog(N\Delta)$ time by working over $\C$ instead.\label{footnote:polylogndelta}}

\paragraph*{Las Vegas Sparse Nonnegative Convolution.}
We can adapt the above idea for Sparse General Convolution to also work for Las Vegas Sparse Nonnegative Convolution that terminates in $O(t\log t)$ time with constant probability. We briefly mention the main differences in the following. 

The first difference is that, since now we are working with Sparse Nonnegative Convolution, we can apply the length reduction steps in \cite{BringmannFN21}. As a result, this algorithm does not have the restriction that $N \le t^{2-\eps}$ any more. 

Secondly, we use a sparsity test that can verify whether a bucket indeed has at most $s$ terms. \cite{BringmannFN22} used this sparsity test in the case $s = 1$, and they suggested to use it for larger $s$.
The sparsity test uses the following result: For a nonnegative vector $V$, $\|V\|_0\ge  s$ if and only if the $s\times s$ matrix $A$ defined as $A_{i,j} = \sum_k k^{i+j} V_k$ is nonsingular \cite{karlin1968total}. As this result only works for nonnegative vectors, it could not be applied to Sparse General Convolution. Furthermore, via a folklore result, two distinct $s$-sparse vectors cannot lead to the same matrix defined above. This enables us to verify whether we have correctly recovered all terms in a bucket, and make our algorithm Las Vegas. 

One technicality in using this sparsity test is that the entries in the  matrix $A$ become too large (cannot fit into one word) due to the $k^{i+j}$ factor when $i+j\in [2s]$ is super-constant, since $k$ can be as large as $N\le t^{2-\eps}$ and $s = O(\log \log N)$.  To resolve this issue, we need to make $k$ smaller. This is made possible by first reducing $N$ down to $t\polylog(t)$ (in a similar way as \cite{BringmannFN21}, but they did this for a different purpose), and then observing that the items hashed to the $r$-th bucket ($r\in [p]$) can only belong to $\{k\in [N]: k\equiv r\pmod{p}\}$, which allows us to treat $k$ as effectively 
 coming from the range $[N/p] = [\polylog(t)]$ instead of from $[N]$.

\paragraph*{Sparse Nonnegative Convolution with high success probability.}
To improve our Las Vegas algorithm from $O(t\log t)$ time in expectation to $O(t\log t)$ time with $1-1/t$ probability, we borrow techniques from \cite{BringmannFN21} and combine with a few more ideas. 
The main techniques of this part do not rely on the large sieve inequality.

The main obstacle to achieving high success probability lies in the choice of the prime $p$ for the mod-prime hash. We already know that a random $p$ achieves the desired load-balancing properties with constant probability, and thus a sample of $O(\log t)$ primes will contain a good prime with $1-1/\poly(t)$ probability. However, in order to pick the good prime, the naive way would require running the whole procedure of FFT and bucket-recovery for all the $O(\log t)$ sampled primes, which would blow up the time complexity by a $O(\log t)$ factor.  (A similar issue also occurs in a different step in the Las Vegas algorithm, namely the step of reducing $N$ down to $t\polylog t$ by hashing modulo a prime $p \approx t\polylog t$: It is known that a random prime is good with $1-1/\polylog(t)$ probability, so the naive way of boosting to high probability incurs $O(\log t/\log \log t)$ blow-up.)
In order to avoid this blow-up, our goal is to use a cheaper way to test the load-balancing property of a mod-prime hash function.

As it turns out, in order to do this testing, it suffices to have the following subroutines given input $A,B\in \{0,1\}^{t\polylog t}$ with $|\supp(A\star B)|\le t$:
(1) a high-success-probability $O(t\log t)$-time algorithm for computing the support of Sparse Convolution $\supp(A\star B)$ and (2) an $O(t\log \log t)$-time algorithm for estimating the support size $|\supp(A\star_p B)|$ to high accuracy with $1-1/\polylog(t)$ success probability.\footnote{$\star_p$ denotes cyclic convolution of array length $p$. } 
These subroutines do not need to be Las Vegas.
See \cref{subsec:estimatesupports} for more formal statements. Intuitively, (1) is useful because we can prepare 
$\supp(A\star B)$ at the beginning (correct with high probability) and every time we can tell whether $p$ is a good prime by directly hashing $\supp(A\star B)$ modulo $p$ (instead of running FFT and bucket recovery). (2) is useful, because if the remaining support size of $A\star_p B$ is large, it means fewer hash collisions happened modulo $p$.

To implement the above two subroutines, we use a random algebrization idea (similar to Indyk's technique \cite{Indyk98a}, which was also used in \cite{BringmannFN21}): we replace the non-zero entries in $A,B$ by random elements in $\F_q$ for some small prime $q$, and then convolve them over $\F_q$, and use the support of this convolution as an estimate of $\supp(A\star B)$ (or $\supp(A\star_p B)$), which has one-sided error due to vanishing over $\F_q$. 
We end up setting prime $q$ to be $ \polylog(t)$, and thus we need a version of Sparse Convolution algorithm over $\F_q$ that runs in very fast $O(t\log q) = O(t\log \log t)$ time, matching the best known $\F_q$ convolution algorithm in the dense case (\cref{thm:packfft}).
See \cref{subsec:estimatesupports} for more formal statements. We design such an algorithm using techniques from \cite{BringmannFN21}, namely an almost 3-wise independent and almost linear hash family $h(x) = ((\sigma x+\tau)\bmod p) \bmod m$. 
As a main difference, \cite{BringmannFN21} avoided using Prony's method and replaced it with other techniques (which only apply to nonnegative convolution instead of our $\F_q$ case), while we still use Prony's method, but with a more time-efficient implementation powered by a lot of bit-packing and table look-ups in word RAM.
Another challenge happens in the bucket-recovery step: Since the field $\F_q$ is small, we do not have primitive roots of high enough order to let us recover the exponent of a term (which can be as high as $t\polylog t$). We work around this by effectively seeing the exponent as from the range $[\polylog(t)]$ (similar to how we reduced $k$ in the Las Vegas algorithm, but it is more complicated here because of the hash family $h(x)$), but this requires some technicalities dealing with the errors of almost-linear hash functions.

\paragraph*{Text-to-Pattern Hamming Distances.}
Now we very briefly describe the ideas behind the new deterministic $X+Y$ lemma (\cref{thm:x-plus-y-advice}), which implies
our improved deterministic Text-to-Pattern Hamming Distances algorithm by the same reduction as \cite{focs23}.

\cite{focs23} proved their slower deterministic $X+Y$ lemma  using hashing based techniques,
which required them to deterministically find a modulus $q$ so that the mod-$q$ hash function achieves good load-balancing property (we know that a random $q$ will have good load-balancing). To achieve this, they used the derandomized hashing idea from \cite{ChanL15} of constructing $q$ as the product of super-constant many small prime factors, which led to large overhead. 

In our work, we achieve this desired derandomization with smaller overhead than \cite{ChanL15,focs23} in the regime of interest. 
At a high level, we show that in our scenario there exists a small subset $\caP$ of candidate primes in some interval $[T/2,T]$, such that sampling a random prime from $\caP$ achieves almost as good load-balancing as sampling a prime from the full interval $[T/2,T]$.
Hence, for fast derandomization, we only need to deterministically check all candidates in the small set $\caP$, instead of all primes in the interval $[T/2,T]$.

The existence of such a small set $\caP$ is highly nontrivial (the set $\caP$ only depends on the input size, not the input instance itself).
The load-balancing quality of set $\caP$ for this purpose is analyzed using the large sieve inequality: Set $\caP$ is good if the Farey fractions with denominator from $\caP$ are well-spaced. Then, the existence of good $\caP$ follows from the probabilistic method; this already yields a (non-uniform) deterministic algorithm with small advice. 
To further remove the advice, we use the method of conditional expectation to deterministically compute the advice as a preprocessing stage. 
For technical reasons, we need to use a composite modulus (product of constant many small primes) instead of mod prime to further speed up the running time of this preprocessing stage.

\subsection{Further related works}
To the best of our knowledge, our work is the first (direct) algorithmic application of the large sieve inequality. 
As for indirect algorithmic applications in the literature, the large sieve inequality and its variants have been used to show a wide range of results in analytic number theory, such as the Bombieri--Vinogradov theorem and Linnik's theorem, which in turn have been used in algorithm design (mostly for the purpose of constructing finite fields of certain desired order). For example, \cite{arnold2016sparse,JinVW21,DudekGGS22,BhargavaGG0U22} used the Bombieri--Vinogradov theorem and \cite{BlaserHLV09, DeKSS13, BlaserJ14} used Linnik's theorem.

Other results from analytic number theory were also used in algorithmic applications. For instance, \cite{LiNakos20, Xu11, bourgain2011explicit} explored the relationship between explicit construction of RIP matrices and the
exponential sum of characters.

\subsection{Open questions}
An interesting open question is whether we can relax the $N\le t^{2-\eps}$ assumption in our $O(t\log t)$ algorithm 
 for Sparse General Convolution (\cref{thm:generalconvo-main}). It would  also be interesting to see whether our derandomization techniques can lead to better deterministic algorithms for Sparse Nonnegative  Convolution \cite{BringmannFN22} and Constellation \cite{nickconstellation}; currently these deterministic algorithms rely on algebraic techniques instead of hashing-based techniques, and have many logarithmic factors in the running time.
 Finally, it would be interesting to explore other applications of tools from analytic number theory in algorithm design.

\paragraph*{Paper organization.}
In \cref{sec:prelim} we introduce useful notations and tools, in particular the large sieve inequality.

In \cref{sec:generalconvo} we describe a Monte Carlo $O(t\log t + t(\log \log t)^{O(1)} \log \log \Delta + \polylog(N\Delta))$-time Sparse General Convolution algorithm for short vectors (\cref{thm:generalconvo-main}).
In \cref{sec:lasvegasnonnega}, we describe a Las Vegas algorithm for Sparse Nonnegative Convolution in expected $O(t\log t + \polylog(N\Delta))$ time.
In \cref{sec:numerical} we remove the $\polylog(N\Delta)$ dependency from these algorithms.
In \cref{sec:highprob} we describe how to boost the Las Vegas algorithm to succeed with high probability, fully proving \cref{thm:lasvegas-main}.

In \cref{sec:determi} and \cref{sec:texttopattern}, we give an improved  deterministic algorithm for the Text-to-Pattern Hamming Distances problem. These two sections can be read independently from \cref{sec:generalconvo,sec:lasvegasnonnega,sec:highprob}.

\section{Preliminaries}

\label{sec:prelim}
\subsection{Notations}
Let $\Z,\N, \R, \C$ denote the integers, the nonnegative integers, the reals, and the complex numbers, respectively.
Let $[N]=\{0,1,\dots,N-1\}$.
We use $\tilde O (f)$ to denote $f\cdot (\log f)^{O(1)}$.

We sometimes use $a \equiv_m b$ to denote $a\equiv b\pmod{m}$ ($a$ and $b$ are congruent modulo $m$).

In an Abelian group $G$, the \emph{sumset} of two sets $A,B\subseteq G$ is defined as $A+B = \{a+b:a\in A,b\in B\}$.

For a vector (an array) $A$ of length $N$ and an index $i\in [N]$, we use $A_i$ or $A[i]$ to refer to the $i$-th coordinate of $A$.  
The \emph{support} of $A$ is denoted as $\supp(A)=\{i\in [N] : A[i]\neq 0\}$. Denote $\|A\|_0 = |\supp(A)|$, $\|A\|_\infty = \max_{i\in [N]} |A[i]|$, and $\|A\|_1 = \sum_{i\in [N]}|A[i]|$.
We say a vector $A$ is $s$\emph{-sparse} if $\|A\|_0\le s$.

For two vectors $A,B\in \R^N$ we write $A\le B$ if $A[i] \le B[i]$ for all $i\in [N]$. We say $A$ is a \emph{nonnegative vector} if $A[i]\ge 0$ for all $i\in [N]$.

The \emph{convolution} of two length-$N$ vectors $A,B$ is defined as the vector $A\star B$ of length $2N-1$ with
\[ (A\star B)[k] = \sum_{\substack{i,j\in [N]\\i+j=k}}A[i]\cdot B[j].\]
The \emph{cyclic convolution} $A\star_m B$ is the length-$m$ vector with 
\[ (A\star_m B)[k] = \sum_{\substack{i,j\in [N]\\i+j\equiv_m k}}A[i]\cdot B[j].\]

We also view vectors as (univariate) polynomials: A length-$N$ vector $A$ corresponds to the polynomial $a(x)=\sum_{i=0}^{N-1}A[i] x^i$. In this way, the convolution $A\star B$ corresponds to polynomial multiplication $a(x)b(x)$, and the cyclic convolution $A \star_m B$ corresponds to $a(x)b(x) \bmod (x^m -1)$.
We say a polynomial is $s$\emph{-sparse} if it has at most $s$ non-zero terms.

We will use the following (non-standard but convenient) notations borrowed from~\cite{BringmannFN22}: 
For a vector $A$, we define vector $\partial A$ with \[(\partial A)[i] = i\cdot A[i],\]
and more generally define \[(\partial^d A)[i] = i^d\cdot A[i].\]
Given a hash function $h\colon [N] \to [m]$, define the length-$m$ vector $h(A)$ by \[h(A)[j] = \sum_{i\in [N]: h(i)=j}A[i].\]

\subsection{Machine model}
\label{subsec:model}
We work in the standard word RAM model where each word is a $w$-bit binary integer, and standard bitwise operations and arithmetic operations (including integer multiplication and division with remainders) on words take unit time, and accessing an element by its memory address also takes unit time. 
The word length $w$ is chosen so that 
each input integer fits into one machine word, and each input integer can be indexed by an address that fits into one machine word.

For the sparse convolution problem of two length-$N$ input integer vectors $A,B$ with coefficient bound $\|A\|_\infty,\|B\|_\infty \le \Delta$ and sparsity bound $\max \{\|A\|_0,\|B\|_0,\|A\star B\|_0\} \le t$, we assume the word length is $w = \Omega(\log t + \log(N\Delta))$. Note that the total input and output size is $O(t\log(N\Delta))$ bits and thus $O(t)$ words. 

This assumption on word length was also used in the previous works \cite{BringmannFN21,BringmannFN22,nickphd} that we build on. The most interesting regime for our results is where $\log(N\Delta)\le O(\log t)$, which is satisfied in almost all combinatorial applications of sparse convolution algorithms mentioned in \cref{sec:intro}.

In the case of large word length $w = \omega(\log t)$, we additionally assume that the word RAM model supports usual elementary functions such as $\sin,\cos,\ln,\exp$ on $w$-bit precision numbers in unit time.
In particular, we need this assumption to justify that two arrays of $t$ words can be convolved in the word RAM model in time complexity $O(t\log t)$ without any extra $w$-dependent factors  (which is regarded as a standard fact in the algorithmic literature, including previous works \cite{BringmannFN21,BringmannFN22}), because the FFT algorithm needs the $t$-th root of unity, $\exp(2\pi \sqrt{-1}/t)$, with $O(w)$-bit precision.
Without this assumption, we need to pay an additive $O(\log w)= O(\log\log (N\Delta))$ time complexity in word RAM to compute the root of unity. 
(We are not aware of any alternatives to FFT that run in $O(t\log t)$ time without this assumption.
For example, the Number-Theoretic Transform requires generating a suitable $O(w)$-bit prime, which takes time complexity growing with $w$.)

\subsection{Large sieve inequality and its consequences}
\label{subsec:analytic-number-theory}
We need several versions of the large sieve inequality from analytic number theory (see e.g.\ the survey of Montgomery \cite{Montgomery}). 

For $x\in \R$, let $e(x) = \exp(2\pi \sqrt{-1} x)$, which only depends on $(x\bmod 1) \in \R/\Z$. Let $\|x\|$ denote the distance from $x$ to the nearest integer.  A basic version of the large sieve inequality is as follows.

\begin{theorem}[see e.g., {\cite[Theorem 7.7]{iwanieckowalski}}]
   \label{thm:largesieve-basic}
Let $\caX\subset \R/\Z$ be a set of $\delta$-spaced numbers $(0<\delta \le 1/2)$, i.e.,  
$\|\alpha -\beta\|\ge \delta$ for all distinct $\alpha,\beta\in \caX$. 
Let $a_0,\dots,a_{N-1}$ be complex numbers. Then,
\[ \sum_{\alpha \in \caX} \Big \lvert \sum_{n=0}^{N-1} a_n \,  e(n\alpha ) \Big \rvert^2 \le (\delta^{-1} +N)\sum_{n=0}^{N-1} |a_n|^2.\]
\end{theorem}
The best possible factor  on the right hand side is $(\delta^{-1}+N-1)$ \cite{MR337775}.
For our application it is sufficient to use estimates with $C(\delta^{-1}+N)$ for a constant $C$; such estimates were first shown in~\cite{MR197427}.

Specializing \cref{thm:largesieve-basic} to the Farey fractions $\mathcal{X} = \{h/q : 1\le h\le q\le Q, \gcd(h,q)=1\}$,
which have spacing $\|h/q-h'/q'\|=\big \|\frac{hq'-h'q}{qq'}\big \| \ge (qq')^{-1} \ge Q^{-2}$, yields the following well-known estimate (see e.g., \cite[Theorem 7.11]{iwanieckowalski}):
\begin{equation}
   \label{eqn:largesieve-useless}
 \sum_{1\le q\le Q}\sum_{\substack{1\le h\le q,\\  \gcd(h,q)=1}} \Big \lvert \sum_{n=0}^{N-1} a_n \,  e(nh/q) \Big \rvert^2 \le (Q^2 +N)\sum_{n=0}^{N-1} |a_n|^2.
\end{equation}

However, \cref{eqn:largesieve-useless} is not sufficient for our applications in \cref{sec:generalconvo,sec:lasvegasnonnega,sec:highprob}. Instead, we need a more refined version of \cref{eqn:largesieve-useless} which restricts to prime denominators, proved by Wolke \cite{wolke0,wolke}. 
\begin{theorem}[\cite{wolke}]
\label{thm:largesievewolke}
Let $Q\ge 10, 0<\eps < 1, N\le Q^{2-\eps}$. 
Let $a_0,\dots,a_{N-1}$ be complex numbers.
Then,
\begin{equation}
   \label{eqn:largesieve-prime}
 \sum_{\text{prime $p\le Q$}} \sum_{h=1}^{p-1} 
\Big \lvert \sum_{n=0}^{N-1} a_n \,  e(nh/p) \Big \rvert^2 \le \frac{C}{\eps} \frac{Q^2 \ln \ln Q}{\ln Q}\sum_{n=0}^{N-1} |a_n|^2, 
\end{equation}
for some absolute constant $C$.
\end{theorem}
By restricting to prime denominators $p\le Q$, \cref{eqn:largesieve-prime} improves over \cref{eqn:largesieve-useless} by an order of $\frac{\ln \ln Q}{\ln Q}$, which is crucial for our applications in \cref{sec:generalconvo,sec:lasvegasnonnega,sec:highprob}.
 It is conjectured by Elliott \cite{elliott} that this factor in \cref{eqn:largesieve-prime} can be further improved to $\frac{1}{\ln Q}$, matching the density of primes $p\le Q$. This conjecture remains open; see recent partial results by Iwaniec \cite{iwaniec}.
Balog, Rivat, and S\'{a}rk\"{o}zy
\cite{largesievesumsets} obtained a better estimate than \cref{eqn:largesieve-prime} when $N\le Q^{1+o(1)}$.

Our derandomization result in \cref{sec:determi,sec:texttopattern} is based on the following immediate corollary of \cref{thm:largesieve-basic}. 
See also \cite[Theorem 2.1]{montgomerytextbook} for a stronger statement. 
\begin{corollary}
\label{lem:largesieve_v2}
Let $\mathcal{X}\subset \R/\Z$ be a finite set of real numbers.
Suppose there are $0 < \delta \le \frac{1}{2}$ and $K\ge 1$ such that for all $\alpha \in \mathcal{X}$ it holds that $|\{\beta \in \mathcal{X}: \|\alpha - \beta\| < \delta\}| \le K$. Then,
\[ \sum_{\alpha \in \caX} \Big \lvert \sum_{n=0}^{N-1} a_n \,  e(n \alpha ) \Big \rvert^2 \le K(\delta^{-1} +N)\sum_{n=0}^{N-1} |a_n|^2.\]
\end{corollary}
\cref{lem:largesieve_v2} is easily derived from \cref{thm:largesieve-basic} by greedily decomposing $\caX$ into $\delta$-spaced sets $\caX_1,\caX_2,\dots,\caX_K$ and applying  \cref{thm:largesieve-basic} separately.

Using \cref{thm:largesievewolke}, the following lemma follows from standard techniques (such as \cite{largesievesumsets}).
\begin{lemma}
   \label{lem:largesievenumbercollision}
Let $0<\eps<1$ be a constant. Let $10\le Q\le N\le Q^{2-\eps}$.
Let $A \subseteq [N]$.
Then, over a uniform random prime $p\in [Q/2,Q]$,
        \[\sum_{a,b\in A} \Pr_{p}[p \mid a-b] \le   
        |A| \cdot O\left ( \frac{|A|}{Q} + \log \log N \right ).
        \]
\end{lemma}

\begin{proof}
    Apply \cref{thm:largesievewolke} with $a_n =  \begin{cases}1 & n \in A \\ 0 & n\notin A\end{cases}$.
    Then $\sum_{n=0}^{N-1} |a_n|^2 = |A|$. 
    Denote $S(\alpha) =
    \sum_{n=0}^{N-1} a_n \,  e(n\alpha)=
    \sum_{a\in A} e(a\alpha )$, and note that $S(0)=|A|$.

    Let $\caP$ denote the set of primes in $[Q/2,Q]$, with size $|\caP| = \Omega(Q/\log Q)$ by the prime number theorem.

    For any integer $c$, if $p$ does not divide $c$, then $\{e(\frac{hc}{p})\}_{h=0}^{p-1}$ equals  $\{e(\frac{i}{p})\}_{i=0}^{p-1}$ as $hc$ ranges over all values modulo $p$. Hence $\sum_{h=0}^{p-1} e(\frac{hc}{p}) = \sum_{i=0}^{p-1} \exp(2\pi\sqrt{-1} i/p) = 0$. If $p$ divides $c$, then $\sum_{h=0}^{p-1} e(\frac{hc}{p}) = p \cdot e(0) = p$. In either case, we have $\mathbf{1}[p\mid c] = \frac{1}{p} \sum_{h=0}^{p-1} e(\frac{hc}{p})$. Therefore, 
    \begin{equation}
    \label{eq:lem:largesievenumbercollision:eq1}
    \sum_{a,b\in A} \Pr_{p\in \caP}[p \mid a-b] = \sum_{a,b\in A}\Ex_{p\in \caP} \left [\frac{1}{p}\sum_{h=0}^{p-1}e\left (\frac{h(a-b)}{p}\right )\right ].
    \end{equation}
    For fixed $h, p$, we have 
    \begin{align*}
    \sum_{a, b \in A} e\left( \frac{h(a-b)}{p}\right) &= \sum_{a, b \in A} e\left( \frac{ha}{p}\right) e\left(-\frac{hb}{p}\right)\\
    &= \sum_{a, b \in A} e\left( \frac{ha}{p}\right) \overline{e\left(\frac{hb}{p}\right)}\\
    &= \left(\sum_{a\in A} e\left( \frac{ha}{p}\right)\right)\overline{\left(\sum_{a\in A} e\left( \frac{ha}{p}\right)\right)}\\
    &= S(h/p) \overline{S(h/p)}\\
    &= |S(h/p)|^2. 
    \end{align*}
    Hence, \cref{eq:lem:largesievenumbercollision:eq1} can be further simplified in the following way: 
\begin{align*}
\sum_{a,b\in A}\Ex_{p\in \caP} \left [\frac{1}{p}\sum_{h=0}^{p-1}e\left (\frac{h(a-b)}{p}\right )\right ]
    & =   \Ex_{p\in \caP}\left [  \frac{1}{p}\sum_{h=0}^{p-1} \left \lvert S(h/p) \right  \rvert^2\right ]\\
    & \le \frac{|A|^2}{Q/2} + \frac{1}{|\caP| Q/2}\sum_{p\in \caP}\sum_{h=1}^{p-1} \left \lvert S\left (h/p\right )\right \rvert ^2\\
    & \le \frac{|A|^2}{Q/2} + \frac{1}{|\caP| Q/2}\cdot O \Big (\frac{Q^2 \ln\ln Q}{\ln Q} |A|\Big ) \tag{by \cref{thm:largesievewolke}}\\
     & = |A| \cdot O\left ( \frac{|A|}{Q} + \ln \ln Q \right ) \tag{\mbox{$|\caP| = \Omega(Q/\log Q)$}}. 
\end{align*}
\end{proof}

\begin{corollary}[Few items land in heavy buckets]
   \label{cor:largesievehash}
   Let $m\le N\le m^{2-\eps'}$, where $0<\eps'<1$ is some fixed constant, and $m$ is large enough.
  Let $A \subseteq [N]$ with $|A|\le m$. Then, for a uniformly random prime $p\in [m/\log \log m, 2m/\log\log m]$,
  if we define buckets $B_k = \{a\in A: a\equiv_p k \}$ for all $k\in [p]$, then 
  \[ \Pr_{p}\Big [ \big |\big \{a\in A :  |B_{a\bmod p}| \ge c\cdot \log \log m\big \}\big | < 0.1|A|\Big ] \ge 0.9,\]
  where $c>0$ is some constant.
\end{corollary}
\begin{proof}
We apply \cref{lem:largesievenumbercollision} with $Q := 2m/\log \log m$, $N:=N$ and $\eps := \eps'/2$. Assume $m$ is large enough so that $m^{2-2\eps} \le (2m/\log \log m)^{2-\eps}$, which implies $N\le m^{2-\eps'}\le Q^{2-\eps}$ as required in \cref{lem:largesievenumbercollision}.
Then \cref{lem:largesievenumbercollision} states that
        \[\sum_{a,b\in A} \Pr_{p}[p \mid a-b] \le   
        |A| \cdot O\left ( \frac{|A|}{2m/\log \log m} + \log \log N \right ) \le O(|A|\cdot \log \log m).
        \]
        Define a random variable $X_{a, b}$ that equals $1$ if $p \mid a - b$ and $0$ otherwise. The above equation implies that 
        \[
        \sum_{a, b \in A} \Ex[X_{a, b}] = O(|A| \cdot \log \log m). 
        \]
    By linearity of expectation, $\Ex\left[\sum_{a, b \in A} X_{a, b}\right] = O(|A| \cdot \log \log m)$. By Markov's inequality, we know that with at least $0.9$ probability over the choice of $p$, it holds that
        \[ \sum_{a,b\in A} X_{a, b} \le c'\cdot |A|\cdot \log \log m\]
        for some constant $c'$. Since the left hand side equals $\sum_{a\in A}|B_{a\bmod p}|$, we know the number of $a\in A$ such that $|B_{a\bmod p}| \ge 10c'\log \log m$ is at most $\frac{c'\cdot |A|\cdot \log \log m}{10c'\log \log m} \le 0.1|A|$.
  Hence setting $c=10c'$ satisfies the desired statement.
\end{proof}

\subsection{Basic algebraic tools}
We will perform arithmetic operations over finite fields $\F_q = \F_{p^\kappa}$, where an $\F_q$ element fits in a constant number of machine words. We perform such operations efficiently with the help of precomputed look-up tables. The following lemma will be applied in our sparse convolution algorithms with the parameter choice $q \le \tilde O(t^6)$, so that the precomputation time is sublinear in $t$ and hence not a bottleneck.
\begin{lemma}
\label{lem:tablefinitefield}
Let $q=p^\kappa$ where $p$ is prime and $\kappa\ge 2$ is an integer.  On a machine with word length $\Omega(\log q)$, after deterministic preprocessing in $\tilde O(q^{0.04} + \kappa^4p^{1/2})$ time, we can perform additions, subtractions and multiplications over $\F_q$ in $O(1)$ time, and multiplicative inversions in $O(\log \kappa)$ time.
\end{lemma}
\begin{proof}[Proof Sketch]
We need a degree-$\kappa$ monic irreducible polynomial over $\F_{p}$ that represents $\F_q$, which can be found by a deterministic algorithm in $\tilde O(\kappa^4 p^{1/2} ) $ time~\cite{Shoup88}.

By our machine assumption, additions, subtractions, and multiplications over $\F_p$ can be done in $O(1)$ time each. If $\kappa \le 100$, then we can naively perform arithmetic operations on $\F_q$-elements (viewed as $\F_p$-polynomials of degree $<\kappa$) in $\kappa^{O(1)} = O(1)$ time.

Otherwise, $\kappa > 100$, and we partition each $\F_q$-element (which we view as an $\F_p$-polynomial of degree $<\kappa$) into $100$ chunks each containing $\le \lceil \kappa/100\rceil $ terms. Each chunk has $p^{\lceil \kappa/100\rceil }\le p^{2\kappa/100} = q^{0.02}$ possibilities.
We prepare an addition table and a multiplication table for the chunks in $(q^{0.02})^2 \cdot \kappa^{O(1)} \le  \tilde O(q^{0.04})$ time.  Then, addition and multiplication over $\F_q$ take $O(1)$ table look-ups.

It is known that computing multiplicative inverses in $\F_{p^{\kappa}}$ amounts to $O(\log \kappa)$ many multiplications of degree-$\kappa$ $\F_p$-polynomials \cite[Theorem 11.10 (i)]{von2013modern} and can hence be done in $O(\log \kappa)$ time via table look-ups (the algorithm \cite[Theorem 11.10 (i)]{von2013modern} uses recursion; in our word RAM adaptation here, once the problem size drops below a small constant fraction of $\kappa$, we can directly look up the answers from another precomputed table instead).
\end{proof}

We need a division-free algorithm for computing the determinant:
\begin{theorem}[\cite{Kaltofen92}]
   \label{thm:ringdeter}
   There is a deterministic algorithm that 
   computes the determinant of an $n\times n$ matrix with entries from an arbitrary commutative ring in $O(n^3\sqrt{n}\cdot \polylog(n))$ ring additions, subtractions, and multiplications. 
\end{theorem}

The following is a word RAM adaptation of Harvey, van der Hoeven, and Lecerf's algorithm \cite{HarveyHL17} for multiplying polynomials over small prime fields. 
\begin{theorem}[{\cite[Appendix B]{focs23}}]
   \label{thm:packfft}
   For any prime $p$, computing the product of two degree-$n$ polynomials from $\F_p[x]$ can be solved in $O(n \log p)$ deterministic time on a machine with word length $\Omega(\log n + \log p)$ bits.
\end{theorem}

The following is an immediate corollary of \cref{thm:packfft} by flattening a degree-$n$ polynomial in $\F_q[x]$ to a degree-$2\kappa n$ polynomial in $\F_p[x]$.

\begin{corollary}
     \label{cor:packfft-finite-field} 
   For any $q = p^\kappa$ for prime $p$, computing the product of two degree-$n$ polynomials from $\F_q[x] $ can be solved in $O(n \log q)$ deterministic time on a machine with word length $\Omega(\log n + \log q)$ bits, assuming the preprocessing in \cref{lem:tablefinitefield} has been performed. 
\end{corollary}
\begin{proof}
    Given two polynomials $f, g \in \F_q[x]$, where $f = \sum_{i=0}^{n-1} \sum_{j=0}^{\kappa - 1} a_{i,j} x^i y^j$ and $g = \sum_{i=0}^{n-1} \sum_{j=0}^{\kappa - 1} b_{i,j} x^i y^j$, we create two degree-$2\kappa n$ polynomials in $\F_p[y]$: $F = \sum_{i=0}^{n-1} \sum_{j=0}^{\kappa - 1} a_{i,j} y^{2i\kappa + j}$ and $G = \sum_{i=0}^{n-1} \sum_{j=0}^{\kappa - 1} b_{i,j} y^{2i\kappa + j}$. Then we use \cref{thm:packfft} to compute the product $F \cdot G$ in $O(\kappa n \log p) = O(n \log q)$ time. Note that $F \cdot G$ equals 
    \[
    \sum_{k \ge 0} \left(\sum_{i + j = k} \left(\sum_{t=0}^{\kappa - 1} a_{i, t} y^t\right) \left(\sum_{t=0}^{\kappa - 1} b_{j, t} y^t\right)\right) \cdot y^{2\kappa k}. 
    \]
    The $k$-th term in the summation (ignoring the $y^{2\kappa k}$ term) modulo the monic polynomial that represents $\F_q$ is exactly the $k$-th term in $f \cdot g$. 
\end{proof}

\subsection{Prony's method}
We need an old algorithm for interpolating sparse polynomials, called Prony's method \cite{prony1795}, which was rediscovered later by Ben-Or and Tiwari \cite{Ben-OrT88}. See \cite{Roche18} for a recent survey, and expositions in \cite[Section 2.7.3]{arnold2016sparse} and \cite[Section 2.3.1]{nickphd}. We will need two instantiations of Prony's method. Throughout, we use the fact that the word length of word RAM is $\Omega(\log t)$ bits, where $t$ is the input/output sparsity.  Note that in our actual applications we will only apply Prony's method with very small sparsity parameter $ \polylog t $ or even $O(\log \log t)$, so the word size will be comparable to the size of the Prony instances.

In the following we first discuss the first instantiation of Prony's method, which will be mainly used in \cref{sec:generalconvo}:
We work with a finite field $\F_q=\F_{p^\kappa}$
(in our applications we can assume additions and multiplications over $\F_q$ take $O(1)$ time, and multiplicative inverses take $O(\log \kappa)$ time, using \cref{lem:tablefinitefield}).
We need the following algorithm for factorizing degree-$s$ polynomials over $\F_q$ into distinct linear factors, which is very fast when the characteristic $p$ is much smaller than the field size $q$. 
In this application, the polynomials have very small degree $s$ (think of $s = O(\log \log t)$), so we are less concerned about the exponent of $s^{O(1)}$.
However, it will be important  that it only has $\log p$ dependency instead of $\log q$, since in our application $q$ can be as large as $t^{O(1)}$.
\begin{lemma}[\cite{KaltofenS97} or {\cite[Section 4.3]{hoeven2022univariate}}]
    \label{lem:factorize}
Let $q=p^\kappa$.
On a machine with word length $\Omega(\log q)$, after deterministic preprocessing in $\tilde O(q^{0.04} + \kappa^4p^{1/2})$ time: 
given a degree-$s$ polynomial $f\in \F_q[z]$ with $s$ distinct roots in $\F_q$,  we can factorize $f$ into linear factors by a Las Vegas algorithm using expected $s^{O(1)}\cdot \log (p\kappa)$ time. 
\end{lemma}
\begin{proof}[Proof Sketch]
    The desired task can be solved by the equal degree factorization algorithm described in \cite[Section 4.3]{hoeven2022univariate}, but with slight adaptation to our word RAM setting in order to achieve the claimed time complexity.  We briefly describe the needed modification as follows.

\cite[Proposition 3]{hoeven2022univariate} stated that equal degree factorization on a degree-$s$ input polynomial in $\F_q[z]$ with $s$ distinct roots in $\F_q$ can be performed in expected time complexity
\begin{align*}
 & O\big (\big (\mathsf{C}_{\F_q}(s)\log \kappa + \mathsf{M}_{\F_q}(s)\log(s\kappa p) + s\mathsf{D}_{\F_q} + s\mathsf{\Phi}_{\F_q}\log\kappa \big ) \log s \big )\\ & \le  s^{O(1)} \cdot O\big (\big (\mathsf{C}_{\F_q}(s) + \mathsf{M}_{\F_q}(s) +  \mathsf{\Phi}_{\F_q}\big )\log(p \kappa) + \mathsf{D}_{\F_q}  \big ),
\end{align*}
where 
\begin{itemize}
    \item $\mathsf{C}_{\F_q}(s)$ denotes the time complexity for modular composition:  compute $(f(z)\circ g(z)) \bmod h(z)$ given $f,g \in  \F_q[z]$ of degrees less than $s$ and monic $h \in \F_q[z]$ of degree $s$,
     \item $\mathsf{M}_{\F_q}(s)$ denotes the time complexity for multiplying two degree-$s$ polynomials in $\F_q[z]$,
     \item  $\mathsf{D}_{\F_q}$ denotes the time complexity for inverting one element in $\F_q$, and
 \item $\mathsf{\Phi}_{\F_q}$ denotes the time complexity for evaluating the iterated Frobenius map, i.e., $a^{p^j}$ for given $a\in  \F_q$ and $0\le j\le \kappa$. 
\end{itemize}
In our word RAM model with word length $\Omega(\log q)$,
we assume efficient arithmetic operations over $\F_q$ after the precomputation of \cref{lem:tablefinitefield}.
Then, $\mathsf{C}_{\F_q}(s), \mathsf{M}_{\F_q}(s) \le s^{O(1)}$ (by brute force) and $\mathsf{D}_{\F_q} = O(\log \kappa)$ (by \cref{lem:tablefinitefield}). Thus, to establish the claimed $s^{O(1)}\cdot \log (p\kappa)$ overall time bound, it remains to show $\mathsf{\Phi}_{\F_q} = O(1)$, i.e., 
iterated Frobenius map can be evaluated in constant time in the word RAM after pre-computation.

To evaluate $a^{p^j}$ for any given $a\in  \F_q = \F_{p^\kappa}$ and $0\le j\le \kappa$,
we write $a$ (which we view as an $\F_p$-polynomial of degree $<\kappa$) into the sum of $100$ chunks each containing $\le \lceil \kappa/100\rceil $ terms, $a = \sum_{i=1}^{100} a_i$.
Then, $a^{p^j} = \sum_{i=1}^{100} (a_i)^{p^j}$ due to characteristic $p$. Thus, after precomputing all possible values for $(a_i)^{p^j}$, we can compute $a^{p^j}$ using constant many table look-ups and additions, in constant total time. Since each chunk $a_i$ has $p^{\lceil \kappa/100\rceil }\le p^{2\kappa/100} = q^{0.02}$ possibilities, the total time for precomputation is only $100\cdot q^{0.02}\cdot \kappa^{O(1)} \cdot \polylog(q) \le O(q^{0.02+o(1)})$, which is dominated by the precomputation time required by \cref{lem:tablefinitefield}.
\end{proof}

A sequence $a = (a_i)_{i\in \N} \in F^{\N}$ is called \emph{linearly recurrent} (over a field $F$) if there exists $n\in \N$ and $f_0,\dots,f_n\in F$ with $f_n\neq 0$ such that $\sum_{0\le j\le n} f_j a_{i+j} = 0$ for all $i\ge 0$, and the polynomial $f = \sum_{0\le j\le n}f_jx^j \in F[x]$ of degree $n$ is called a \emph{characteristic polynomial} of $a$. The \emph{minimal polynomial} of $a$ is the monic polynomial of least degree that is a characteristic polynomial of $a$, and the degree of this minimal polynomial is called the \emph{recursion order} of $a$. %

\begin{lemma}[Prony's method over $\F_q$]
    \label{lem:pronyfq}
Let $q = p^\kappa$ and  let $\omega_{pm} \in \F_q^*$ have multiplicative order at least $d$.
Assume efficient arithmetic operations over $\F_q$ after the precomputation of \cref{lem:tablefinitefield}.
The following can be done after $O(d)$ time preprocessing:

Let $f\in \F_q[x]$ be an $s$-sparse polynomial with degree smaller than $d$.
Given the evaluation of $f(\omega_{pm}^j)$ for all $0\le j< 2s$, we can recover all terms of $f$ in $ \poly(s) \log (p\kappa)$ expected time by a Las Vegas algorithm. 

Note that this implies a zero-error algorithm (i.e., either outputs the correct answer or $\bot$)\footnote{Throughout the paper we use $\bot$ to indicate failure.} that terminates in worst case $ \poly(s) \log (p\kappa)$ time (even under ill-formed input, i.e., the input evaluations do not come from an $s$-sparse polynomial), and the probability of outputting $\bot$ under valid input is at most $2^{-\poly(s)}$.
\end{lemma}
\begin{proof}[Proof Sketch]
   Here we verify that Prony's method can be easily implemented under the claimed time budget. 
   We refer to the exposition in \cite[Section 2.3.1]{nickphd}. 
   
   We first let $\Lambda(z)$ denote the minimal polynomial of the linear recurrence $f(\omega_{pm}^0),\dots, f(\omega_{pm}^{2s-1})$, which is defined as the lowest-degree monic polynomial $\Lambda(z) =\sum_{\ell=0}^r \lambda_\ell z^\ell$ such that $\sum_{\ell=0}^r \lambda_\ell f(\omega_{pm}^{i+\ell})=0$ for all $0\le i\le 2s-1-r$.
   It is known that $\Lambda(z)=\prod_{j=1}^{s'}(z- \omega_{pm}^{e_j})$ if $f(x)$ has exactly $s'$ terms and $f(x) = \sum_{j=1}^{s'} a_j x^{e_j}$ (where $a_j\neq 0$, $s'\le s$, and $e_j < d$).

We compute the coefficients of $\Lambda(z)$  by solving linear systems  
     using $\poly(s)$ $\F_q$-operations in $\poly(s)\log (\kappa)$ time.
  Then, we can factor $\Lambda(z)$ to obtain all roots $\omega_{pm}^{e_j}$ using \cref{lem:factorize} in $\poly(s)\cdot \log (p\kappa)$ time. 
We can precompute $\omega_{pm}^{i}$ for all $0\le i< d$ in $O(d)$ time, so that we recover the support $\{e_j\}$ from the given $\{\omega_{pm}^{e_j}\}$ using table look-ups.

Finally, the equations $f(\omega_{pm}^i) = \sum_{j=1}^{s'} a_j (\omega_{pm}^{e_j})^i$ for all $0\le i< s'$ form a full rank linear system (a transposed Vandermonde system) with unknown coefficients $a_1,\dots,a_{s'}$, which can be solved in $\poly(s')\log(\kappa)\le \poly(s)\log(\kappa)$ time. 

The algorithm described above is Las Vegas, due to the factoring step. We transform it into an algorithm with worst-case time upper bound by aborting after two times the expected running time. Then we reduce the failure probability to $2^{-\poly(s)}$ by running $\poly(s)$ times independently. 
\end{proof}

Now we state our second instantiation of Prony's method (over small prime fields only), which will be used in \cref{sec:highprob}. In this application we have more stringent requirement on the dependency on the sparsity $s$. It is more complicated and more substantially relies on bit-packing and look-up tables in word RAM. 
\begin{restatable}[Prony's method over $\F_p$ with bit packing]{lemma}{pronybitpack}
    \label{lem:pronybitpack}
Let $t$ be a parameter. Let $\eps>0$ be any small fixed constant. Let $p\le \polylog(t)$ be a prime,  
and $\omega\in \F_p^*$ have multiplicative order $p-1$.
On a machine with word length $\Omega(\log t)$ bits, the following holds after $t^{0.2+o(1)}$ time preprocessing:

Let $f\in \F_p[x]$ be an $s$-sparse polynomial with degree smaller than $(p-1)/2$ and $p\ge s^2$.
Given the evaluation of $f(\omega^j)$ for all $0\le j< 2s$, we can recover all terms of $f$ in worst case 
    $O\big (s\log s +  ( \frac{s}{\log t}+ 1 ) \cdot (\log t)^{0.1}\cdot (\log p)^{O(\sqrt{\log s})}\big )$
 time by a zero-error algorithm (i.e., either outputs the correct answer or returns $\bot$), and the probability of outputting $\bot$ is at most $2^{-(\log t)^{0.1}}$ under valid input.
\end{restatable}

In our actual application we will choose $s =\polylog(t)$ and $p = \polylog(t)$, so the second term in the running time becomes negligible. The proof of \cref{lem:pronybitpack} is given in \cref{sec:pronybitpack}.

\subsection{Sparsity test}
We need to test the sparsity of a polynomial  given its evaluations on several points.

In \cite{BringmannFN22}, the authors suggested that the following lemma could be potentially useful in improving their $O(t\log t\log \log t)$-time Las Vegas algorithm.

\begin{lemma}[\cite{karlin1968total}, see also \cite{moment}]
    \label{lem:momentmatrix}
    Let $V$ be a nonnegative vector. Then $\|V\|_0\ge  s$ if and only if the following $s\times s$ matrix is non-singular:
    $A_{i,j} = \|\partial^{i+j}V\|_1$
    where $0\le i,j<s$.
\end{lemma}

The following lemma is folklore: 

\begin{lemma}
\label{lem:unique-moment}
    Let $U, V$ be nonnegative vectors indexed on $\N$ where $\|U\|_0, \|V\|_0 < s$. If $\|\partial^{i}U\|_1 = \|\partial^{i}V\|_1$ for every $0 \le i < 2(s-1) $, then $U = V$.
\end{lemma}
\begin{proof}
    Let $\{x\in \N: U[x]\neq 0 \text{ or } V[x]\neq 0\} = \{x_1,\dots,x_{m}\}$ with $m\le \|U\|_0+\|V\|_0 \le 2(s-1)$. By nonnegativity of $U$ and $V$, $\|\partial^{i}U\|_1 = \sum_{j=1}^{m} x_j^i U[x_j]$, $\|\partial^{i}V\|_1 = \sum_{j=1}^{m} x_j^i V[x_j]$, and hence $\sum_{j=1}^{m} x_j^i(U[x_j]-V[x_j]) = \|\partial^{i}U\|_1- \|\partial^{i}V\|_1 = 0$ for all $0\le i<m\le 2(s-1)$, expressed as the following linear system:
    \[
    \begin{bmatrix}
    1 & 1 & \cdots & 1 \\
    x_1 & x_2 & \cdots & x_{m}\\
    & & \vdots & \\
    x_1^{m - 1} & x_2^{m - 1} & \cdots & x_{m}^{m - 1}\\ 
    \end{bmatrix}
    \begin{bmatrix}
        U[x_1]-V[x_1] \\
        U[x_2]-V[x_2] \\
        \vdots \\
        U[x_{m}]-V[x_{m}]
    \end{bmatrix}
    =\begin{bmatrix}
        0\\0\\ \vdots \\ 0
    \end{bmatrix}.
    \]
    The coefficient matrix is a transposed Vandermonde matrix and has nonzero determinant as $x_1, \ldots, x_{m}$ are distinct. Hence the system is uniquely solved to get $U[x_j]-V[x_j]=0$ for all $1\le j\le m$, that is, $U = V$. 
\end{proof}

\subsection{Known lemmas for sparse convolution}

\paragraph*{Sumset size estimation.}
We use the following lemma from \cite{nickphd} for estimating the sumset size up to a constant factor. 

\begin{lemma}[Estimating Sumset Size {\cite[Lemma 4.26]{nickphd}}]
   \label{lem:estimatesumset}
 Let $\delta > 0$. There is an algorithm that, given
two sets $X,Y\subseteq [N]$ computes a constant-factor approximation 
$t^*$ of $t = |X+Y|$ in
time $O(t\log (\delta^{-1}) + t\log \log t)$ with error probability $\delta$. Additionally, $t^* \le O(t)$ always holds. 
\end{lemma}

\paragraph*{Efficient error correction.}
The following lemma says that, once we have a good approximation $C$ of $A\star B$, then we can quickly correct the remaining errors and compute $A\star B$.

\begin{restatable}[see {\cite[Section 8]{BringmannFN21}}]{lemma}{modprimerecoverfinitefield}
Let $t$ be a parameter, and $\F_p$ be a prime field for $p \le t$. 
Let $\delta \le 1/\log t$.
Given vectors $A,B\in \F_p^{t^{2-\eps}}$ and $C \in \F_p^{2t^{2-\eps}-1}$ for $\eps > 0$ such that $\|A\star B - C\|_0 \le t/\log^2 t$, there is a Monte Carlo algorithm that computes $A \star B$ in time
\[ O\left((\|A\|_0+\|B\|_0+\|C\|_0 + t)\cdot \log(1/\delta)\right) \]
with error probability $\delta$.

The same statement holds if $\F_p$ is replaced by $\Z$.
   \label{lem:modprimerecoverfinitefield}
\end{restatable}
A version of \cref{lem:modprimerecoverfinitefield} over nonnegative integers was implicit in the proof from \cite[Section 8]{BringmannFN21}. In \cref{sec:modprimerecoverfinitefield} we prove \cref{lem:modprimerecoverfinitefield} by adapting that proof to finite fields and integers.

We also need a Las Vegas version from \cite{BringmannFN22}.
\begin{lemma}[{\cite[Lemma 20]{BringmannFN22}}]
    \label{lem:naiverecover-full}
   Given vectors $A, B\in \Z^N$ and $C \in \Z^{2N-1}$ such that $A \star B - C$ is nonnegative, and any parameter $s$, there is a Las Vegas algorithm that runs in (worst case) time $O(\|A\|_0 + \|B\|_0 + \|C\|_0 + s \log s)$ and computes a vector $R$ such that for every $z \in [2N-1]$:
\begin{itemize}
    \item $R[z] \leq (A \star B - C)[z]$ (always), and
    \item $R[z] < (A \star B - C)[z]$ with probability at most $c \log N \cdot \frac{\|A \star B - C\|_0}{s}$ for some constant $c$.
\end{itemize}
\end{lemma}

This immediately implies the following high probability version of the algorithm. 
\begin{corollary}
    \label{cor:naiverecover-full-whp}
   Given vectors $A, B\in \Z^N$ and $C \in \Z^{2N-1}$ such that $A \star B - C$ is nonnegative and any parameter $m$, there is a Las Vegas algorithm running in (worst case) time $O((\|A\|_0 + \|B\|_0 + \|C\|_0) \log N + m \log^2 N (\log m + \log \log N) )$ that computes a vector $R$ such that
\begin{itemize}
    \item $R \le A \star B - C$ (always), and
    \item when $m \ge \|A \star B - C\|_0$, $R = A \star B - C$ with $1- 1/\poly(N)$ probability. 
\end{itemize}

\end{corollary}
\begin{proof}
    We apply \cref{lem:naiverecover-full} $c \log N$ times with $s = 2c m \log N$, and take the entrywise maximum of the resulting arrays. 

    If $m \ge \|A \star B - C\|_0$, then for every $z \in \supp(A \star B - C)$, each invocation correctly recovers $(A \star B - C)[z]$ with probability $1 - c\log N\cdot \frac{\|A \star B - C\|_0}{s}\ge 1/2$.  Therefore, after $O(\log N)$ invocations, all terms are recovered with high probability. 

    The total running time is $O(\log N \cdot ((\|A\|_0 + \|B\|_0 + \|C\|_0) + s \log s)) = O((\|A\|_0 + \|B\|_0 + \|C\|_0) \log N + m \log^2 N (\log m + \log \log N ))$. 
\end{proof}

\paragraph*{Sparse verification.}
We need the following lemma from \cite{BringmannFN21} for efficiently checking whether $A\star B = C$ given sparse vectors $A,B,C$.
\begin{lemma}[Sparse Verification, slight modification of  {\cite[Lemma 9.2]{BringmannFN21}}]
\label{lem:sparse-verification}
    Given three vectors $A, B\in \Z^N$ and $C \in \Z^{2N-1}$ with sparsity at most $t$ where $N \le \poly(t)$, and with $\|A\|_\infty, \|B\|_\infty \le \Delta$, there is an $O(t + \polylog(t \log \Delta))$ time randomized algorithm checking whether $A \star B = C$ with $1/ \poly(t)$ error probability. In addition, if $A \star B = C$, the algorithm is always correct. 
\end{lemma}
\begin{proof}
    We follow the proof of \cite[Lemma 9.2]{BringmannFN21} with a slight modification. First, we verify that every entry of $C$ has absolute value at most $\Delta^2 t$, since otherwise $C$ cannot be equal to $A\star B$ and we can immediately reject. 
    Pick a random prime $p \in [P,2P]$ where $P:=t(N +  \log \Delta)$, in time $O(\polylog P ) = O(\polylog(t \log \Delta))$.\footnote{This is the only difference from \cite{BringmannFN21}. In \cite[Lemma 9.2]{BringmannFN21}, the prime $p$ was from a different range, and the analysis did not rely on $p$ being random.} Reduce the entries of the input vectors modulo $p$, and view the input vectors as polynomials $A(x),B(x),C(x)\in \F_p[x]$ of degree lower than $2N$. Then, pick a random $\omega \in \F_p$ and evaluate $A(\omega),B(\omega),C(\omega)$ in $O(t\log_t N)=O(t)$ time (using bulk exponentiation \cite[Lemma 4.2]{BringmannFN21}), and reject if and only if $A(\omega) B(\omega) \neq C(\omega)$.

 When $A \star B - C = 0$, the algorithm always accepts, so we focus on the case where $A \star B - C \ne 0$. Fix any entry $(A \star B - C)[i]\neq 0$, which has absolute value bounded by $2 \Delta^2 t$. 
By the same argument as in \cref{footnote:prime}, the probability that $p$ divides $(A \star B - C)[i]$ is $O(\frac{\log(2\Delta^2 t)}{P \log P}) = O(1/t)$. Thus, $A(x)B(x)-C(x)$ is a nonzero polynomial in $\F_p[x]$ with $1-O(1/t)$ probability. Then, by Schwartz--Zippel lemma, $A(\omega)B(\omega)=C(\omega)$ happens with probability at most $\frac{\deg(A(x)B(x)-C(x))}{p} \le \frac{2N}{t(N+\log \Delta)} = O(1/t)$. 
 The error probability can be reduced to $1/\poly(t)$ by repeating a constant number of times. 
\end{proof}

We also need a version of sparse verification over small prime fields.
\begin{lemma}
\label{lem:sparseverifysmallp}
    Given three vectors $A, B \in \F_p^N$ and $C \in \F_p^{2N-1}$ with sparsity at most $t$ where $N\le t^{2}$ and prime $p\le t$,  there is an $O(t)$ time randomized algorithm checking whether $A \star B = C$ with $1/ \poly(t)$ error probability.  Furthermore, if $A \star B = C$, the algorithm is always correct. 
\end{lemma}
\begin{proof}
    Let $\kappa$ be the smallest positive integer such that $q=p^\kappa \ge t^{3}$.
We work over the finite field $\F_q \supset \F_p$. 
By \cref{lem:tablefinitefield}, after preprocessing in $\tilde O(q^{0.04} + \kappa^4 p^{1/2} ) \le \tilde O((t^3p)^{0.04}+ (\log t)^4p^{1/2}) = o(t)$ time,  we can assume additions and multiplications over $\F_q$ take constant time each.

The rest of the proof is the same as the previous lemma.
View vectors $A,B,C$ as polynomials $A(x),B(x),C(x)\in \F_p[x] \subset \F_q[x]$, pick a random element $\omega \in \F_q$, and verify whether $A(\omega)B(\omega) = C(\omega)$ in $O(t)$ time (using bulk exponentiation \cite[Lemma 4.2]{BringmannFN21}). By Schwartz--Zippel lemma, the error probability is at most $2N/q = O(t^2/t^3) = O(1/t)$. The error probability can be reduced to $1/\poly(t)$ by repeating a constant number of times. 
\end{proof}

\section{Sparse General Convolution for Short Vectors}
\label{sec:generalconvo}
In this section we show a Monte Carlo algorithm for Sparse General Convolution (with possibly negative coefficients) when the input vectors have short length:

\generalconvo*

Throughout \cref{sec:generalconvo,sec:lasvegasnonnega,sec:highprob}, we make the following simplifying assumption:
\begin{assumption}
    \label{assumption1}
    Assume a large prime $p_b =\poly(N\Delta)$ is given.
\end{assumption}
\cref{assumption1} allows us to work with a large enough prime field $\F_{p_b}$ for recovering the coefficients of $A\star B$.
 We could fulfill \cref{assumption1} by additionally running a Las Vegas $\polylog(N\Delta)$-time algorithm that finds such a large prime $p_b$ (by rejection sampling with the AKS primality test \cite{agrawal2004primes}). This would add a $\polylog(N\Delta)$ term to the claimed time complexities of the algorithms described in \cref{sec:generalconvo,sec:lasvegasnonnega,sec:highprob} (namely \cref{thm:lasvegas-main} and \cref{thm:generalconvo-main}), which is non-negligible only when $\polylog(N\Delta)$ is very large compared to $t$.
 Later in \cref{sec:numerical}, we will describe how to avoid paying this extra $\polylog(N\Delta)$ term, by modifying the algorithms described in \cref{sec:generalconvo,sec:lasvegasnonnega,sec:highprob} to work over $\C$ instead of over the prime field $\F_{p_b}$ assumed in \cref{assumption1}.

\subsection{The algorithm}
\label{subsec:generalalgo}

\newcommand{\RecoveryStep}{\textsc{RecoveryStep}\xspace}
\newcommand{\RecoverBucket}{\textsc{RecoverBucket}\xspace}

For now, we assume that we know an upper bound $t$ such that $\|A\|_0+\|B\|_0+\|A\star B\|_0 \le t$ and $N< t^{2-\eps}$. Later we will relax this assumption.

Our algorithm proceeds in multiple rounds. In each round we invoke \RecoveryStep with some sparsity parameter $m$ to be determined later (see pseudocode in \cref{alg:recovermontecarlo}).

\paragraph*{Preparation.}
Let $\hat P_0 = \Theta(\log^2 t\cdot \log (t\Delta))$, and we need the following technical definition of $P_0$ to get around some corner case later:
\[ P_0 := \begin{cases} \hat P_0 & (\hat P_0 \notin [t,2t^5]) \\ 2t^5 & (\hat P_0 \in [t,2t^5]). \end{cases}\]
Note $P_0 \ge \Omega(\log^2 t\cdot \log (t\Delta))$ and $\log P_0 \le O(\log \log t + \log \log \Delta)$.
Every time in the beginning of \RecoveryStep,  we sample a uniformly random prime $p_0 \in [P_0,2P_0]$ by a Las Vegas algorithm in $\polylog(P_0) = O(\poly \log \log t + \poly \log \log \Delta)$ expected time.
(This is the only place in the algorithm with the $\poly \log \log \Delta$ dependency.) 
We say index $i\in [2N-1]$ is \emph{bad}, if $(A\star B-C)[i]\neq 0$ but $(A\star B - C)[i] \equiv_{p_0} 0$.
Since $|(A\star B-C)[i]|\le \poly(t\Delta)$, we have $\Pr_{p_0}[i\text{ is bad}] \le O(\frac{\log (t\Delta)/\log P_0}{P_0/\log P_0})\le O(\frac{1}{\log^2 t})$ for every $i\in [2N-1]$ by the prime number theorem (similarly to \cref{footnote:prime}).
By linearity of expectation and Markov's inequality, with at least $0.9$ probability, there are at most $O(\|A\star B- C\|_0/\log^2 t)$ bad indices.

We choose the smallest integer $\kappa \ge 1$ so that $q := p_0^\kappa \ge 2t^5$, and we will work over the finite field $\F_q$ (as we will see in the next paragraph, the choice of $2t^5$ here is made to simplify the finding of a large order element in $\F_q$).
If $\kappa=1$, then we can perform additions and multiplications over $\F_q = \F_{p_0}$ in $O(1)$ time because 
elements in $\F_{p_0}$ can fit into a constant number of machine words. Otherwise, $\kappa \ge 2$, which implies $P_0 \le p_0 < 2t^5$ and hence $P_0<t$ by our definition of $P_0$, and we can use the precomputation of \cref{lem:tablefinitefield}   in $\tilde O(q^{0.04} + \kappa^4 p_0^{1/2}) \le \tilde O( (4P_0t^5)^{0.04} + (2P_0)^{1/2} \polylog(t)) \le \tilde O( (t^6)^{0.04} + t^{1/2})= o(t)$ time, so that  additions and multiplications over $\F_q$ also take $O(1)$ time.\footnote{Note that if we did not define $P_0$ as $2t^5$ when $\hat{P}_0 \in [t, 2t^5]$, this case would cost $\tilde O(t^{2.5})$ time, which we cannot afford.}

We now find an $\omega\in \F_q^*$ with multiplicative order at least $2N$. For every $\sigma < 2N$, there are at most $\sigma$ elements  $x\in \F_q^*$ such that $x^\sigma=1$,  so the number of elements in $\F_q^*$ with multiplicative order smaller than $2N$ is at most $1+2+\dots +(2N-1)<2N^2<2t^{2(2-\eps)} < q/t$, so a randomly chosen $\omega\in \F_q^*$ has multiplicative order at least $2N$ with $\ge 1 - 1/t$ success probability. We can verify whether the multiplicative order of an element is at least $2N$ by baby-step giant-step  (see e.g., \cite{nechaev1994complexity}) in $O(\sqrt{N} \log(N)) = o(t)$ time. Therefore, the expected time to find $\omega\in \F_q^*$ is $o(t)$. 
We also choose an $\omega_b\in \F_{p_b}$ with multiplicative order at least $2N$. Similarly, we can do this in $o(t)$ expected time.

We will use the results in $\F_q$ to approximately recover the support, and use results in $\F_{p_b}$ to recover the actual coefficients on the support.

\begin{algorithm}
\DontPrintSemicolon
\caption{$\RecoveryStep(A,B,C,m)$}
\label{alg:recovermontecarlo}
\textbf{Input:} Sparse vectors $A,B\in \Z^N$, $C \in \Z^{2N-1}$, and a parameter $m$ where $m \le N \le m^{2-\eps'}$ \tcp{$\eps'$ is the constant from \cref{cor:largesievehash}}
\textbf{Output:} A sparse vector $R \in \Z^{2N-1}$, for details see \cref{lem:recovermontecarlo-step}.\;
Sample a random prime $p\in [m/\log \log m, 2m/\log \log m]$ and let $h(x)= x\bmod p$\;
Sample $p_0$, construct the finite field $\F_q$, and find $\omega$ and $\omega_b$ as described in paragraph ``Preparation''\;
Let parameter $s = \lceil c\cdot\log \log m\rceil $. \tcp{$c$ is the constant from \cref{cor:largesievehash}}
\For{$i \gets 0,1,2,\dots, 2s$\label{Line:recovermontecarlo:forloop1}}{
    Define vector $\tilde A^{(i)}\in \F_q^N$ by $\tilde A^{(i)}[k]:= A[k]\cdot \omega^{ki}\in \F_q$ (and similarly define vectors $\tilde B^{(i)}\in \F_q^N,\tilde C^{(i)}\in \F_q^{2N-1}$)\;
    Define vector $ A^{(i)}\in \F_{p_b}^N$ by $ A^{(i)}[k]:= A[k]\cdot \omega_{b}^{ki}\in \F_{p_b}$ (and similarly define vectors $ B^{(i)}\in \F_{p_b}^N, C^{(i)}\in \F_{p_b}^{2N-1}$) \tcp{Recall that we assume a large prime $p_b$ is given as stated in \cref{assumption1}.}
    Compute vector $\tilde Z^{(i)}\gets h(\tilde A^{(i)}) \star_p h(\tilde B^{(i)}) - h(\tilde C^{(i)})$ via FFT \label{line:recovermontecarlo:FFT1}\;
    Compute vector $Z^{(i)}\gets h(A^{(i)}) \star_p h(B^{(i)}) - h(C^{(i)})$ via FFT \label{line:recovermontecarlo:FFT2}\;
}
Initialize $R \gets (0,\dots,0)\in \Z^{2N-1}$\;
\For{$k\in [p]$}{
    $f(x)\gets \RecoverBucket\big (\{\tilde Z^{(i)}[k]\}_{i=0}^{2s},\{Z^{(i)}[k]\}_{i=0}^{2s}, k\big )$, where $f(x) = \sum_{\ell=0}^{2N-2}f_\ell x^\ell\in \Z[x]$ is a polynomial in sparse representation \label{line:recovermontecarlo:recover-bucket}\;
    $R \gets R + (f_0,\dots,f_{2N-2})$\;
}
\For{$i \in \supp(R)$}{
    \lIf{$|R[i] + C[i]| > \|A\|_\infty \|B\|_\infty t$}{$R[i] \gets -C[i]$ \label{line:recovermontecarlo:adjust-bound}}
}
\Return{} $R$\;

\end{algorithm}

\begin{algorithm}
\DontPrintSemicolon
\caption{$\RecoverBucket(\tilde v_0,\tilde v_1,\dots,\tilde v_{2s}, v_0, v_1,\dots, v_{2s}, k)$}
\label{alg:recoverstepmontecarlo}
\textbf{Input:} Values $\tilde v_0,\dots,\tilde v_{2s}\in \F_q$, $ v_0,\dots, v_{2s}\in \F_{p_b}$ and $k\in [p]$\;
\textbf{Output:} a polynomial $f(x) \in \Z[x]$ represented as a list of non-zero coefficients. For details see \cref{lem:recoverbucketmontecarlo}\;
Use Prony's method over $\F_q$ (\cref{lem:pronyfq}) with degree bound $2N/p$, sparsity bound $s$ and $\omega_{pm} = \omega^p$ to obtain $g(y)=\sum_{j=1}^{s'}b_j y^{e_j}\in \F_q[y]$ such that $g((\omega^{p})^{i}) = \tilde v_i / \omega^{k i} $ for all $0\le i\le 2s$.\;
\lIf{\cref{lem:pronyfq} failed to find such a $g(y)\in \F_q[y]$}{\Return{$0$}}
Determine $a_1,\dots,a_{s'}\in \F_{p_b}$ such that polynomial $f(x) =\sum_{j=1}^{s'} a_j x^{p \cdot e_j + k}$ satisfies $v_i = f(\omega_b^{i})$ for all $0\le i < s'$, by solving a linear system over $\F_{p_b}$. \label{line:general-conv-linear-system}\;
\For{$j \in [s']$}{
    \lIf{$a_j \in [\lfloor p_b / 2\rfloor]$}{$a'_j \gets a_j \in \Z$}
    \lElse{$a'_j \gets a_j - p_b \in \Z$}
}
\Return{$f(x) =\sum_{j=1}^{s'} a_j' x^{p\cdot e_j + k}$}\;

\end{algorithm}

\paragraph*{Recovering algorithm.}
At the beginning of \cref{alg:recovermontecarlo}, based on the   sparsity parameter $m$,  we sample a random prime $p\in [m/\log \log m, 2m/\log \log m]$ and hash the indices into $p$ buckets according to their remainders modulo $p$.  For $k\in [p]$, we say bucket $k$ is \emph{bad}, if there is at least one bad index $i$ such that $i\equiv_p k $. The number of terms mapped to bucket $k$ is the number of indices $i$ where $(A \star B - C)[i] \ne 0$ and $i\equiv_p k$. We say bucket $k$ is \emph{heavy} if the bucket contains more than $s$ elements ($s =  \lceil c\cdot\log \log m\rceil $ where $c$ is the constant from \cref{cor:largesievehash}); otherwise, the bucket is \emph{light}.

\begin{lemma}
\label{lem:recoverbucketmontecarlo}
    If a bucket $k\in [p]$ is light and  does not contain bad indices, \cref{alg:recoverstepmontecarlo} (given inputs computed in \cref{alg:recovermontecarlo}) returns $\sum_{i \in \mathcal{B}} (A \star B - C)[i] x^i$, where $\mathcal{B}$ is the set of indices in the bucket; otherwise, \cref{alg:recoverstepmontecarlo} returns an $s$-sparse polynomial. The running time of each call of \cref{alg:recoverstepmontecarlo} is $(\log \log t)^{O(1)} \log \log \Delta$, after an $O(t^{2-\eps} (\log \log m) / m + m + \log\log \Delta)$ initial preprocessing. 
\end{lemma}
\begin{proof}
    For a light bucket $k$ that does not contain bad indices, suppose the nonzero terms in it have indices $k+p E_1, \ldots, k+p E_{s'}$ for some $s', E_1, \ldots, E_{s'}$. Then it is straightforward to verify that $\tilde{v}_i = \tilde Z^{(i)}[k] = \sum_{j=1}^{s'} (A \star B - C)[k+p E_j] \cdot \omega^{(k+p E_j) i}$. Thus, $\tilde{v}_i / \omega^{k i} = \sum_{j=1}^{s'} (A \star B - C)[k+p E_j] \cdot (\omega^p)^{i E_j}$ (note that we can precompute $\omega^{-ki}$ for all $k \in [p], i \in [2s+1]$ in $O(ps) = O(m)$ time, after computing $\omega^{-1}$ via exponentiation in $O(\log q)= O(\log P_0 + \log t) = O(\log t +\log\log \Delta)$ time). Since the bucket contains no bad indices, $(A \star B - C)[k+p E_j] \ne 0$ in $\F_q$, so Prony's method can successfully recover $e_1, \ldots, e_{s'}$ where $e_j = E_j$ for every $j = 1, \ldots, s'$. 

    Then in 
    \linereff{line:general-conv-linear-system}, we are solving a linear system where the coefficient matrix is   $$
    \begin{bmatrix}
        1 & 1 & \cdots & 1 \\
        \omega_b^{p \cdot e_1 + k} & \omega_b^{p \cdot e_2 + k} & \cdots & \omega_b^{p \cdot e_{s'} + k}\\
        & & \vdots & \\
        (\omega_b^{p \cdot e_1 + k})^{s'-1} & (\omega_b^{p \cdot e_2 + k})^{s'-1} & \cdots & (\omega_b^{p \cdot e_{s'} + k})^{s'-1}\\
    \end{bmatrix},
    $$
    which is a transposed Vandermonde matrix. As $\omega_b$ has multiplicative order at least $2N$, 
    $\omega_b^{p \cdot e_1 + k}$, $ \omega_b^{p \cdot e_2 + k}$, $\ldots, \omega_b^{p \cdot e_{s'} + k}$ are distinct, and thus the Vandermonde matrix is invertible. Thus, the linear system has a unique solution, so we must have $a_j \equiv_{p_b} (A \star B - C)[k+p e_j]$. Finally, as $\|A \star B - C\|_\infty \le \poly(t \Delta) < p_b / 2$ by picking $p_b$ large enough, the choices of $a_j'$ ensure that $a_j' = (A \star B - C)[k+p e_j]$. 

    If the bucket is not light, or if it contains a bad index, the algorithm might return $0$, or some erroneous polynomial, but the sparsity of the polynomial is at most $s$. 

    The bottlenecks of the algorithm are 
    \begin{itemize}
        \item Prony’s method: it runs in $\poly(s) \log(p_0 \kappa) = (\log \log t)^{O(1)} \log \log \Delta$ time (after $O(N / p) = O(t^{2-\eps} \log \log m / m)$ preprocessing).
        \item Linear system solving: it runs in $\poly(s) = (\log \log t)^{O(1)}$ time. 
    \end{itemize}
\end{proof}

\begin{lemma}
   \label{lem:recovermontecarlo-step} 
Given parameter $m$ and vectors $A,B\in \Z^N$ and $C \in \Z^{2N-1}$ with coefficient bound $\|A\|_\infty,\|B\|_\infty,\|C\|_\infty\le \poly(t\Delta)$,
 \cref{alg:recovermontecarlo} returns a random vector $R\in \Z^{2N-1}$ in 
 \begin{align*}
     O(t^{2-\eps} (\log \log m) / m & +  m\log t 
      + (\|A\|_0+\|B\|_0+\|C\|_0)\cdot \log \log m \\
      & + m (\log \log t)^{O(1)} \log \log \Delta +  \poly \log \log \Delta )
 \end{align*}
 time such that:
 \begin{itemize}
        \item If $\|A\star B - C\|_0 \le m$, then 
 \[ \Ex[\|A\star B - C - R\|_0] \le 0.4 \|A\star B - C \|_0.\]
        \item $\|R+C\|_\infty \le \poly(t \Delta)$. 
 \end{itemize} 
\end{lemma}
\begin{proof}

By \cref{lem:recoverbucketmontecarlo}, if a bucket is light and does not contain bad indices, then \RecoverBucket{} can successfully return all terms in the bucket. Thus, the following are the only possibilities that $(A \star B - C - R)[i] \ne 0$ for some term $i$:
\begin{itemize}
    \item $i$ is in a heavy bucket. By \cref{cor:largesievehash},
   the number of terms in heavy buckets is at most $0.1 \|A \star B - C\|_0$, with probability $0.9$.  This means that the expected number of terms in heavy buckets is $\le 0.1 \|A \star B - C\|_0 \cdot 0.9 + \|A \star B - C\|_0 \cdot 0.1 \le 0.19 \|A \star B - C\|_0$. 
    \item $i$ is in a light bucket that contains a bad index. As shown earlier, the probability that each $i \in \supp(A \star B - C)$ is bad is $O(\frac{1}{\log^2 t})$. Thus, the expected number of bad indices is $O(\frac{\|A \star B - C\|_0}{\log^2 t})$. Each bad index can cause at most $s$ terms $i$ to be in a light bucket containing a bad index. Thus, the expected number of indices $i$ in this case is $O(\frac{\|A \star B - C\|_0 s}{\log^2 t})  \le o(\|A \star B - C\|_0)$. 
    \item $i$ is among the terms that are erroneously returned by \RecoverBucket{}. By the previous two cases, the expected number of buckets that can cause \RecoverBucket{} to return erroneously is at most \[\frac{0.19 \|A \star B - C\|_0}{s} + O(\frac{\|A \star B - C\|_0}{\log^2 t}) \le \frac{0.2 \|A \star B - C\|_0}{s}.\] Each such bucket can cause $s$ terms to be erroneous, so the expected number of indices $i$ in this case is $\le 0.2\|A \star B - C\|_0$. 
\end{itemize}
Summing over all the cases, we get the claimed bound $\Ex[\|A\star B - C - R\|_0] \le 0.4 \|A \star B - C\|_0$. 

\linereff{line:recovermontecarlo:adjust-bound} ensures that $\|R+C\|_\infty \le \poly(t \Delta)$. Note that it does not increase $\|A\star B - C - R\|_0$ because when $|R[i] + C[i]| > \|A\|_\infty \|B\|_\infty t$, $(A\star B - C - R)[i]$ cannot be $0$, so the sparsity $\|A\star B - C - R\|_0$ cannot increase when we reset the value of $R[i]$.

Finally, we analyze the running time of the algorithm. The bottlenecks of the algorithm are the following:
\begin{itemize}
    \item The for loop starting at \linereff{Line:recovermontecarlo:forloop1}: 
    Preparing $A^{(i)}, B^{(i)}, C^{(i)}, \tilde A^{(i)}, \tilde B^{(i)}, \tilde C^{(i)}$ takes $O(\|A\|_0 + \|B\|_0 + \|C\|_0)$ time for each iteration of the for loop (using bulk exponentiation \cite[Lemma 4.2]{BringmannFN21}). 
    The FFT on \linereff{line:recovermontecarlo:FFT1} is for multiplying two degree-$O(m / \log \log m)$ polynomials in $\F_q[x]$; separately consider two cases:
    \begin{itemize}
        \item 
            If $\kappa=1$, then $q$ is the prime $p_0$. In this case, we can simply use FFT over integers (i.e., the standard FFT over complex numbers and then round to integers) then modulo $p_0$ to compute it in $O(m \log m / \log \log m)$ time.\footnote{In this case, $p_0$ may have dependency on $\Delta$ and cannot be bounded by $\poly(m)$, so 
            \cref{thm:packfft} and \cref{cor:packfft-finite-field} would be too expensive here.}
            \item
            Otherwise, $\kappa \ge 2$, so $q \le \poly(t)$. Hence, we can use \cref{cor:packfft-finite-field} to compute it in time $O(m / \log \log m \cdot \log q ) = O(m \log t / \log \log m)$.
    \end{itemize}
    The FFT on \linereff{line:recovermontecarlo:FFT2} can be computed by FFT over integers then modulo $p_b$ in $O(m \log m / \log \log m)$ time. Thus, the total running time over all iterations is $O(m \log t + (\|A\|_0 + \|B\|_0 + \|C\|_0) \cdot \log \log m)$. 
    \item Calling \RecoverBucket{} in \linereff{line:recovermontecarlo:recover-bucket}: By \cref{lem:recoverbucketmontecarlo}, after a preprocessing with time $O(t^{2-\eps} (\log \log m) / m+ m + \log \log \Delta)$, each call to \RecoverBucket{} takes $(\log \log t)^{O(1)} \log \log \Delta$ time. Summing over all $k \in [p]$ gives the $m (\log \log t)^{O(1)} \log \log \Delta$ bound. 
\end{itemize}
Note that even if $\|A \star B - C\|_0 > m$, we can still halt the algorithm after its running time exceeds the claimed time bound, so the time bound always holds. 
\end{proof}

\begin{lemma}
   \label{lem:recovermontecarlo-known-t} 
Let $\eps>0$ be a fixed constant.
Given a parameter $t$ and two vectors $A,B\in \Z^N$ with the promise that $N<t^{2-\eps}$, where $t\ge \|A\|_0  +\|B\|_0 + \|A\star B\|_0$,
one can compute $A\star B$ in $O(t \log t + t(\log\log t)^{O(1)}\log\log \Delta +   \poly \log \log \Delta) $ expected time by a Monte Carlo algorithm with $0.9$ success probability, where $\Delta = \max\{\|A\|_\infty,\|B\|_\infty\}$.
\end{lemma}
\begin{proof}
We repeatedly apply \cref{lem:recovermontecarlo-step}. Initially, we set $m_0 = 100 t$ and $C_0 = 0$. For the $i$-th time, we call \cref{lem:recovermontecarlo-step} with $m = m_{i-1}$ and $C = C_{i-1}$. Afterwards, we set $m_i \gets 0.8 \cdot m_{i-1}$ and $C_i \gets C_{i-1} + R$, where $R$ is returned by \cref{lem:recovermontecarlo-step}. After we reach $m_i \le 2t /\log^2 t$, we use \cref{lem:modprimerecoverfinitefield} to compute $A \star B$ with $\delta = 1/t$. 

    First, suppose we always have $m_i \ge \|A \star B - C_i\|_0$ for every $i$. Then the running time of all calls to \cref{lem:recovermontecarlo-step} can be bounded as (up to constant factors)
    \begin{align*}
    & \sum_m \left(t^{2-\eps} (\log \log m) / m+  m\log t + (\|A\|_0+\|B\|_0+\|C\|_0)\cdot \log \log m \right.\\ 
    & \quad\quad\quad\quad \left.+ m (\log \log t)^{O(1)} \log \log \Delta +   \poly \log \log \Delta  \right)\\
    \le &   \sum_m \left(t^{1-\eps} \polylog(t) + m \log t + (\|A\|_0 + \|B\|_0 + \|A \star B\|_0 + m)  \cdot \log \log t \right.\\ 
    & \quad\quad\quad\quad \left. + m (\log \log t)^{O(1)} \log \log \Delta  +  \poly \log \log \Delta \right)\\
    \le& t \log t + t (\log \log t)^{O(1)} \log \log \Delta  +  \poly \log \log \Delta . 
    \end{align*}

Also, as $\|A \star B - C_i\|_0 \le m_i \le 2t/ \log^2 t$ holds when calling \cref{lem:modprimerecoverfinitefield}, the conditions of \cref{lem:modprimerecoverfinitefield} are satisfied. 
Thus, we will get the correct $A \star B$ with error probability $1/t$. The running time of \cref{lem:modprimerecoverfinitefield} is $O(\|A\|_0 + \|B\|_0 + \|C\|_0 + t) \cdot \log(1/ \delta) \le O(\|A\|_0 + \|B\|_0 + \|A \star B\|_0 + 2t/ \log^2 t + t) \cdot \log(1/ \delta) = O(t \log t)$. 

It remains to bound the probability that we always have $m_i \ge \|A \star B - C_i\|_0$ for every $i$. Let $E_i$ be the event that for every $0\le j < i$, $m_j \ge \|A \star B - C_j\|_0$. We will show 
\begin{equation}
\label{eq:recovermontecarlo-known-t:expectation}
\Ex\biggl[\|A \star B - C_i\|_0 \mid E_i\biggr] \le 0.4^i \cdot t
\end{equation}
by induction. For $i = 0$, $\|A \star B - C_0\|_0 = \|A \star B\|_0  \le t$, so the base case holds. Now consider $i > 0$. By induction, we have $\Ex\biggl[\|A \star B - C_{i-1}\|_0 \mid E_{i-1}\biggr] \le 0.4^{i-1} \cdot t$ and note that $E_i = E_{i-1} \wedge (m_{i-1} \ge \|A \star B - C_{i-1}\|_0)$. Thus, we have $\Ex\biggl[\|A \star B - C_{i-1}\|_0 \mid E_{i}\biggr] \le \Ex\biggl[\|A \star B - C_{i-1}\|_0 \mid E_{i-1}\biggr] \le 0.4^{i-1} \cdot t$. Since $E_{i}$ implies $m_{i-1} \ge \|A \star B - C_{i-1}\|_0$,  \cref{lem:recovermontecarlo-step} implies \[\Ex\biggl[\|A \star B - C_i\|_0 \mid E_{i}\biggr] \le 0.4 \Ex\biggl[\|A \star B - C_{i-1}\|_0 \mid E_{i}\biggr] \le 0.4^i \cdot t.\]

Finally, the probability that $m_i \ge \|A \star B - C_i\|_0$ for every $i$ is exactly $\Pr[E_\ell]$, where $\ell$ is the total number of times we call \cref{lem:recovermontecarlo-step}, which can be bounded as follows:
\begin{align*}
    \Pr[E_\ell] & = 1 - \Pr[\neg E_\ell]\\ 
    &= 1 - \sum_{i=1}^{\ell} \Pr\biggl[\|A \star B-C_{i-1}\|_0 > m_{i-1} \mid E_{i-1}\biggr] \Pr[E_{i-1}]\\
    &\ge  1 - \sum_{i=1}^{\ell} \Pr\biggl[\|A \star B-C_{i-1}\|_0 > m_{i-1} \mid E_{i-1}\biggr] \\
    &\ge 1 - \sum_{i=1}^{\ell} \Ex\Bigl[\|A \star B-C_{i-1}\|_0  \mid E_{i-1}\Bigr]/m_{i-1} \tag{Markov's inequality}\\
    &\ge 1 - \sum_{i=1}^\ell \frac{0.4^{i-1} \cdot t}{0.8^{i-1} \cdot 100 t} \tag{by \cref{eq:recovermontecarlo-known-t:expectation}}\\
    & \ge 0.99. 
\end{align*}
\end{proof}

Finally, it remains to remove the assumption that an upper bound of $t$ is given. We follow the strategy of \cite[Section 9]{BringmannFN21}, by guessing $\tilde{t} = 2^\nu \cdot \max\{\|A\|_0, \|B\|_0\}$ for $\nu = 0, 1, \ldots, \lceil \log t\rceil$ as an upper bound for $t$ in \cref{lem:recovermontecarlo-known-t}. Note that we can halt \cref{lem:recovermontecarlo-known-t} after $O(\tilde{t} \log \tilde{t} + \tilde{t}(\log\log \tilde{t})^{O(1)}\log\log \Delta)$ steps, and each call to \cref{lem:recovermontecarlo-known-t} would still have constant success probability, provided that $\tilde{t} \ge t$. In order to boost the success probability, we run \cref{lem:sparse-verification} after each call to \cref{lem:recovermontecarlo-known-t}, and only stop if the check in  \cref{lem:sparse-verification} passes. 

Let $\nu^*$ be the smallest $\nu$ such that $2^\nu \cdot \max\{\|A\|_0, \|B\|_0\} \ge t$. Because of \cref{lem:sparse-verification}, when $\nu < \nu^*$, the algorithm does not stop with a wrong answer w.h.p., and the total time spent for these iterations is $O(t \log t + t(\log\log t)^{O(1)}\log\log \Delta)$. For each iteration $\nu \ge \nu^*$, the algorithm terminates with constant probability, say $3/4$. Thus, the probability that we still have not succeeded right before calling \cref{lem:sparse-verification} with parameter $\tilde{t} = \max\{\|A\|_0,\|B\|_0\} \cdot 2^{\nu}$ is $\le 1/4^{\nu - \nu^*}$. Therefore, the expected running time is (up to constant factors)
\begin{align*}
    & \sum_{\nu \ge \nu^*} \frac{1}{4^{\nu - \nu^*}} \cdot \big(2^{\nu - \nu^*} \cdot t \log (2^{\nu - \nu^*} \cdot t) \\
    & + 2^{\nu - \nu^*} \cdot t (\log\log (2^{\nu - \nu^*} \cdot t))^{O(1)}\log\log \Delta + \poly \log \log \Delta\big)\\
    \le \ & t \log t + t(\log\log t)^{O(1)}\log\log \Delta + \poly \log \log \Delta. 
\end{align*}

The error probability of the algorithm is bounded by constant.

\begin{remark}
    \label{remark:tlogtmontecarlo}
    By combining the algorithm in this section with  a length reduction step from \cite[Section 10]{BringmannFN21}, we get an alternative Monte Carlo algorithm for Sparse \emph{Nonnegative} Convolution in $O(t\log t)$ time with constant success probability (when $\Delta$ is not too big), partly recovering the main result of \cite{BringmannFN21}.
        
    The length reduction in \cite{BringmannFN21} does not work for general convolution with possibly negative entries, because they use an almost-affine hash function in this step, which can cause a zero  entry to be split into multiple nonzero entries. See \cite[Section 3.3]{BringmannFN21} for more details. 
\end{remark}

\section{Las Vegas Sparse Nonnegative Convolution}
\label{sec:lasvegasnonnega}

Now we turn to sparse convolution with nonnegative input numbers. 
Let $A,B \in \N^N$ be the sparse input vectors with nonnegative integer entries with $\|A\|_\infty, \|B\|_\infty \le \Delta$.
In this section, we give a Las Vegas algorithm for Sparse Nonnegative Convolution with expected time $O(t\log t)$, where $t =\|A\star B\|_0$, under \cref{assumption1} (that a large prime $p_b = \poly(N\Delta)$ is given).
We will improve it to terminate in $O(t\log t)$ time with $1-1/\poly(t)$ probability in \cref{sec:highprob}.

The meat of the proof is in \cref{subsec:lasvegas-smalluniverse}, where we use the sparsity test of \cref{lem:momentmatrix} to design a Las Vegas algorithm under two additional assumptions:
\begin{itemize}
    \item $N\le  t\polylog t$.
    \item We are given an auxiliary small input set $S_0$ purported to satisfy $\supp(A\star B)\subseteq S_0$.
\end{itemize}
 Then, in  \cref{subsec:lasvegas:universereduction}, we use more-or-less standard techniques to lift these two assumptions:
 \begin{itemize}
\item Using the length-reduction techniques from \cite{BringmannFN21,BringmannFN22} based on hashing, we can reduce the general case to the $N\le \poly(t)$ case, and then to the $N\le t\polylog t$ case.
\item 
 The desired auxiliary set $S_0$ can be supplied by running a Monte Carlo Sparse Convolution algorithm (such as \cite{BringmannFN21} or \cref{remark:tlogtmontecarlo}) to compute $\supp(A \star B)$. This set $S_0$ may be incorrect with small non-zero probability, but the Las Vegas algorithm is guaranteed to detect such an error when it happens.
 \end{itemize}

\subsection{Algorithm for short vectors}
\label{subsec:lasvegas-smalluniverse}
In this subsection, we assume $A,B\in \N^{N}$ for some $N\le t \polylog t$.   This assumption will be useful for controlling the time complexity in two different places of the proof:
\begin{itemize}
    \item  The iterative peeling procedure decreases the sparsity of $A\star B-C$ at exponential speed, so it only takes $O(\log \log t)$ iterations before the sparsity $m$ drops to $m\le t/\polylog t$, at which point we can finish the computation using  \cref{cor:naiverecover-full-whp}.
    \item Our implementation of the sparsity tests needs to perform arithmetic operations on integers of magnitude roughly $(N/p)^{\log \log t}\cdot \poly(t\Delta)$, where prime $p = \Theta(m/\log\log t)$ is chosen according to the current sparsity $m\in [t/\polylog t, t]$. Our assumption $N\le t\polylog t$ implies that these integers have magnitude bounded by $(\polylog t)^{\log \log t}\poly(t\Delta) \le \poly(t\Delta)$ and hence fit into constant many machine words, allowing constant time arithmetic operations.
\end{itemize}

We choose an element $\omega_b\in \F_{p_b}$ with multiplicative order at least $2N$. 
As described in \cref{subsec:generalalgo}, this can be done using baby-step giant-step in $O(\sqrt{N} \log N) = O(t)$ time.

\newcommand{\RecoveryStepLasVegas}{\textsc{RecoveryStepLasVegas}}
\newcommand{\RecoverBucketLasVegas}{\textsc{RecoverBucketLasVegas}}

Our main subroutine, \RecoveryStepLasVegas, is described in \cref{alg:recoverlasvegas}, and its properties are stated in \cref{lem:recoverlasvegas}. The parameter $m$ in this subroutine will always be upper bounded by our main parameter $t$ for the Sparse Convolution problem.
\begin{algorithm}
\DontPrintSemicolon
\caption{\RecoveryStepLasVegas$(A,B,C,S,m)$}
\label{alg:recoverlasvegas}
\textbf{Input:} Sparse vectors $A,B\in \N^N$, $C \in \N^{2N-1}$ such that $A\star B \ge C$, and a set $S\subseteq [2N-1]$ and a parameter $m$\;
\textbf{Output:} A sparse vector $R \in \N^{2N-1}$ and a set $S'\subseteq [2N-1]$, for details see \cref{lem:recoverlasvegas}\;
Let parameter $s = \lceil c\cdot\log \log m\rceil $. \tcp{$c$ is the constant from \cref{cor:largesievehash}}
Sample a random prime $p\in [m/\log \log m, 2m/\log \log m]$ and let $h(x)= x\bmod p$. Repeat until the number of elements of $S$ hashed by $h$ into heavy buckets (i.e., of size $\ge s$) is at most $0.1m$.\;
\For{$i \gets 0,1,2,\dots, 2s$\label{line:recoverlasvegas:forloop0}}{
    Compute vector $Y^{(i)}\in \N^{p}$ according to \cref{eqn:recurrencefory}, where $Y^{(i)}_k$ is defined as $\lVert \partial^i V_k\rVert_1$ for $V_k \in \N^{\lfloor (2N-2-k)/p\rfloor+1}$ such that $V_k[j] = (A \star B - C)[jp + k]$. \label{line:recoverlasvegas:definey}\;
    Define vector $ A^{(i)}\in \F_{p_b}^N$ by $ A^{(i)}[k]:= A[k]\cdot \omega_{b}^{ki}\in \F_{p_b}$ (and similarly define vectors $ B^{(i)}\in \F_{p_b}^N, C^{(i)}\in \F_{p_b}^{2N-1}$)\;
    Compute vector $Z^{(i)}\gets h(A^{(i)}) \star_p h(B^{(i)}) - h(C^{(i)})$\;
}
Initialize $R \gets (0,\dots,0)\in \N^{2N-1}$\;
Initialize $S' \gets S$\;
\For{$k\in [p]$ \label{line:recoverlasvegas:forloop1}}{
    \If{$Y^{(0)}_k, Y^{(1)}_k,\dots,Y^{(2s)}_k$ pass the $<s$-sparsity test \label{line:recoverlasvegas:sparsity-test}(\cref{lem:momentmatrix})}{
        $f(x)\gets \RecoverBucketLasVegas\big (\{Z^{(i)}[k]\}_{i=0}^{2s},\{e\in S: e\equiv_p k\}\big )$, where $f(x) = \sum_{\ell=0}^{2N-2}f_\ell x^\ell\in \Z[x]$ is a polynomial in sparse representation\;
        \If{$f$ is nonnegative and is consistent with $Y^{(0)}_k, Y^{(1)}_k,\dots,Y^{(2s)}_k$ \label{line:recoverlasvegas:verify-moments}}{
            $R \gets R + (f_0,\dots,f_{2N-2})$\;
            $S' \gets S'\setminus \{e \in S': f_e\neq 0\}$\;
        }
    }
}
\Return{$R, S'$}\;

\end{algorithm}

\begin{algorithm}
\DontPrintSemicolon
\caption{\RecoverBucketLasVegas$(v_0,v_1,\dots,v_{2s},S_k)$}
\label{alg:recoversteplasvegas}
\textbf{Input:} Values $v_0,\dots,v_{2s}\in \F_{p_b}$ and a set $S_k = \{e_1,e_2,\dots,e_{s'}\} \subseteq [2N-1]$\;
\textbf{Output:} a polynomial $f(x) \in \Z[x]$ represented as a list of non-zero coefficients.\;
Determine $a_1,\dots,a_{s'}\in \Z$ such that polynomial $f(x) =\sum_{j=1}^{s'} a_j x^{e_j}$ satisfies $ v_i = f(\omega_{b}^i)$ for all $0\le i\le 2s$, by solving a linear system over $\F_{p_b}$ and then converting the answers from $\F_{p_b}$ to $\Z$.\;
\lIf{successfully found $a_1,\dots,a_{s'}$}{\Return{$f(x) =\sum_{j=1}^{s'} a_j x^{e_j}$}}
\lElse{\Return{$f(x)=0$}}

\end{algorithm}

\begin{lemma}
   \label{lem:recoverlasvegas} 
Given parameter $m$ and vectors $A,B\in \N^N$, $C \in \N^{2N-1}$ where $N\le m \log^{7} m$ and $A\star B  \ge C$, and an auxiliary set $S \subseteq [2N-1]$ of size $|S|\le m$, 
 \cref{alg:recoverlasvegas} returns a random sparse vector $R\in \N^{2N-1}$ and $S'\subseteq [2N-1]$ in 
 $O(m\log m + (\|A\|_0+\|B\|_0+\|C\|_0)\cdot \log \log m)$ time with $1-1/\poly(m)$ probability, such that:
 \begin{itemize}
    \item $R\le A\star B -C$ always holds.
    \item If $\|A\star B - C\|_0 \le m$ and $\supp(A\star B - C)\subseteq S$, then we always have
 \[ \|A\star B - C - R\|_0 \le 0.1 m \quad \text{ and } \quad  \supp(A \star B - C - R)\subseteq S'.\]
 \end{itemize} 
\end{lemma}
\begin{proof}
We first define $X^{(i)}$, which will be used to compute $Y^{(i)}$ defined at \linereff{line:recoverlasvegas:definey}: for $i\in \{0,1,\dots,2s\}$, $X^{(i)}\in \N^{p}$ is defined by
\begin{equation}
    \label{eqn:defnX}
X^{(i)}:= \Big (\sum_{j=0}^i \tbinom{i}{j}h(\partial^j A)\star_p h(\partial^{i-j}B) \Big )- h(\partial^i C), 
\end{equation}
so for $k\in [p]$,
\begin{align*}
X^{(i)}_k &= \Big (\sum_{j=0}^i \tbinom{i}{j} \sum_{k' \in [p]} \sum_{\ell_1 \equiv_p k' } (\ell_1^j A_{\ell_1}) \sum_{\ell_2 \equiv_p k-k' } (\ell_2^{i-j} B_{\ell_2}) \Big)-h(\partial^i C)_k\\
&= \Big (\sum_{k' \in [p]} \sum_{\ell_1 \equiv_p k' }  \sum_{\ell_2 \equiv_p k-k' } ((\ell_1+\ell_2)^{i} A_{\ell_1}B_{\ell_2}) \Big)-h(\partial^i C)_k\\
&= \Big (\sum_{\ell \equiv_p k } \ell^i (A \star B)_\ell \Big )- h(\partial^i C)_k\\
&=h(\partial^i (A\star B - C))_k.
\end{align*}
Therefore, by running the sparsity test of \cref{lem:momentmatrix} using  $X^{(0)}_k, X^{(1)}_k,\dots,X^{(2s)}_k$, we would be able to verify whether the bucket $k$ has less than $s$ elements, but the time complexity would be too high, due to high-precision integers: The integers $X^{(i)}_k$ can have magnitude roughly as large as $\Delta m \cdot N^{2s}$ where $s=\Theta(\log \log m)$, so computing them directly using FFT would lead to a $\log \log m$ factor blow up.
That's why we introduce $Y^{(i)}$ as an alternative: 
By definition of $Y^{(i)}$
at \linereff{line:recoverlasvegas:definey}, 
\[ Y^{(i)}_k = \sum_{d = 0}^{\lfloor (2N-k-1)/p\rfloor }d^i\cdot (A\star B  -C)[k+dp],  \]
which is analogous to $X_k^{(i)}$ 
(in the sense that running the sparsity test of \cref{lem:momentmatrix} on $Y^{(0)}_k, Y^{(1)}_k,\dots,Y^{(2s)}_k$ would also decide whether the bucket $k$ has less than $s$ elements),
but has smaller magnitude and hence fits into constant many machine words of length $\Omega(\log m + \log \Delta)$ bits:
\begin{align}
\label{eqn:ykupbound}
0\le Y_k^{(i)} & \le  (2N/p)^i \sum_{d=0}^{\lfloor (2N-k-1)/p\rfloor }(A\star B  -C)[k+dp]  \\  & \le (\polylog m)^{2s} \cdot \poly(m\Delta) \nonumber \\
& \le \poly(m\Delta),\nonumber
\end{align}
where we used the assumption that $N\le m\log^7 m$ and that $p =\Theta(m/\log \log m), s = \Theta(\log \log m)$.
In the following we describe how to compute the integers $Y_k^{(i)}$ efficiently.

Note that $Y^{(i)}_k$ is determined by $\left\{X^{(j)}_k\right\}_{0 \le j \le i}$ using the following recurrence:
\begin{equation}
    \label{eqn:recurrencefory}
 Y_k^{(i)} =  p^{-i}\cdot \Big (X_k^{(i)} - \sum_{j=0}^{i-1}\tbinom{i}{j} k^{i-j}p^{j} Y_k^{(j)}  \Big ). 
\end{equation}
We compute all $X_k^{(i)}$ mod $M$ using FFTs over integers, where $M = \poly(m\Delta)$ is a power of $2$ that is larger than the upper bound in \cref{eqn:ykupbound}.
More specifically, we compute the DFTs of $h(\partial^i A), h(\partial^i B)$ for all $0\le i \le 2s$ in $O(p\log p)\cdot 2s \le O(m\log m)$ total time. 
Then,  using the fact that the DFT of the  convolution of two vectors is the entry-wise multiplication of their DFTs,
by \cref{eqn:defnX} we can compute the DFT of $X^{(i)}$  for all $0\le i\le 2s$, in $O(s^2 \cdot p) \le O(m\log \log m)$ time. Finally, we perform inverse-DFT to recover all $X^{(i)}$ in $O(p\log p)\cdot 2s \le O(m\log m)$ time.

Then, we use \cref{eqn:recurrencefory} to compute all $Y_k^{(i)}$ modulo $M$ in $p\cdot \poly(s) \le m\poly(\log \log m)$ time. Note that we need to compute $p^{-1}$ in the ring $\Z_M$ (recall that we picked $M$ to be a power of $2$, so $p^{-1}$ exists). We can do so in $O(p)$ time by finding the smallest $j M + 1$ for $j < p$ that is divisible by $p$, and $p^{-1}$ would be $\frac{jM + 1}{p}$. Since $M = \poly(m\Delta)$ is chosen to be large enough, we actually recover all $Y_k^{(i)}$ over $\Z$.

When invoking \cref{lem:momentmatrix} in \linereff{line:recoverlasvegas:sparsity-test} for the sparsity test, we need to compute the determinant of an $s\times s$ integer matrix where each entry fits into constant many machine words.
This can be done using \cref{thm:ringdeter} in $\poly(s) = \poly(\log \log m)$ time.
Then, checking whether $f$ is consistent with $Y_k^{(i)}$ (meaning that $Y^{(i)}_k = \sum_{d = 0}^{\lfloor (2N-2-k)/p\rfloor }d^i\cdot f_{k+dp}$ where $f(x) = \sum_{\ell=0}^{2N-2}f_\ell x^\ell$) at \linereff{line:recoverlasvegas:verify-moments} takes $\poly(s) = \poly(\log \log m)$ time. 

The bottlenecks of the algorithm are the following:
\begin{itemize}
    \item The cost for finding a good $p$. By \cref{cor:largesievehash}, the number of terms in heavy buckets is at most $0.1 |S|\le 0.1m$, with probability $\ge 0.9$. This means that we only need to repeat $O(\log m)$ times to achieve high probability. Therefore, the running time is $O(m \log m)$ w.h.p. 
    \item The for loop starting at \linereff{line:recoverlasvegas:forloop0}. The FFTs there take $O(s \cdot p \log p) = O(m \log m)$ time, and preparing the arrays $\partial^i A, \partial^i B, \partial^i C, A^{(i)}, B^{(i)}, C^{(i)}$ takes $O((\|A\|_0 + \|B\|_0 + \|C\|_0) \cdot \log \log m)$ time. The claimed time bound thus follows. 
\end{itemize}

Finally, we verify the guarantees of the lemma. 
\begin{itemize}
    \item $R \le A \star B - C$ always holds. If the condition on \linereff{line:recoverlasvegas:sparsity-test} is satisfied, then we know that there are less than $s$ terms in bucket $k$. By \cref{lem:unique-moment}, the only $s$-sparse nonnegative polynomial $f$ that is consistent with $X_k^{(i)}$ (i.e., $\|\partial^i f\|_1 = X_k^{(i)}$) is the one containing the correct terms in the bucket $k$. Consequently, the only $s$-sparse polynomial $f$ whose nonzero terms all have indices consistent with $Y_k^{(i)}$ is the one containing the correct terms in the bucket $k$. Thus, if the algorithm passes the condition on \linereff{line:recoverlasvegas:verify-moments}, we know $f_\ell = (A \star B - C)_\ell$ for all $\ell \equiv_p k $, and by \RecoverBucketLasVegas{}, $f_\ell = 0$ for all $\ell \not\equiv_p k $. Therefore, $R \le A \star B - C$ always holds. 
\item If $\supp(A\star B - C)\subseteq S$, then $ \supp(A \star B - C - R)\subseteq S'$. This is implied by the same analysis as above.  
\item If $\|A\star B - C\|_0 \le m$ and $\supp(A\star B - C)\subseteq S$, then \cref{alg:recoverlasvegas} recovers all items in light buckets. As $p$ is a good prime, $\|A \star B - C - R\|_0 \le 0.1 m$.\qedhere 
\end{itemize}

\end{proof}

\begin{lemma}
    \label{lem:lasvegas-tiny-universe}
Given vectors $A,B\in \N^N$, a parameter $t$ where $N\le t \log^{5} t$, and an auxiliary set $S_0\subseteq [2N-1]$ of size $|S_0|\le t$, there is a randomized algorithm in (worst case) $O(t\log t)$ time that outputs either a sparse vector or $\bot$, with the following guarantee: 
\begin{itemize}
    \item Either $output = A\star B$ or $output = \bot$ (always).
        \item If $\supp(A\star B)\subseteq S_0$, then $\Pr[output = \bot] \le 1/t$.
\end{itemize}
\end{lemma}

Note that the failure probability can be reduced to $1/\poly(t)$ by a constant number of repetitions. 
\begin{proof}
We repeatedly apply \cref{lem:recoverlasvegas}, starting with parameters $m_0 = t, C_0 = 0, S_0$ and continuing with $m_{i+1} = 0.1 \cdot m_{i}, C_{i+1} = C_{i}+R, S_{i+1} = S'$, until we reach $m_i \le t / \log^2 t$.
If $\supp(A\star B)\subseteq S_0$, then by induction we can show $\|A \star B - C_i\|_0 \le m_i \le t / \log^2 t$.
Then, we run \cref{cor:naiverecover-full-whp} using $C_i$ to compute $C = A \star B$. By the guarantees of \cref{lem:recoverlasvegas} and \cref{cor:naiverecover-full-whp}, we always have $C \le A \star B$, we can check whether $C = A \star B$ by verifying whether $\lVert C \rVert_1 = \lVert A \rVert_1 \lVert B \rVert_1$. If $C \ne A \star B$, then some call  to \cref{lem:recoverlasvegas} or \cref{cor:naiverecover-full-whp} failed, which happens with probability at most  $1/\poly(t)$. 

 The total running time for calling \cref{lem:recoverlasvegas} is
\[\sum_{m_i} O(m_i \log t + t \log \log m) = O(t \log t + t \polyloglog t), \]
as $m_i$ is geometrically decreasing and the number of them is $O(\log \log t)$.
The running time of \cref{cor:naiverecover-full-whp} is 
\[
O(t \log N + (t / \log^2 t) \log^2 N (\log(t / \log^2 t) + \log \log N)) = O(t \log t). 
\]

Note that the guarantees of \cref{lem:recoverlasvegas} and \cref{cor:naiverecover-full-whp} ensure that we never return an incorrect vector, even if $\supp(A\star B) \not \subseteq S_0$.
\end{proof}

\subsection{Length reduction}
\label{subsec:lasvegas:universereduction}
To relax the $N\le t\polylog t$ assumption in the previous subsection, we first need to reduce the length of the input vectors from $\poly(t)$ to $t\polylog t$ by hashing modulo a random prime $p\le t\polylog t$.
This length reduction is similar to \cite[Section 7]{BringmannFN21}, and  can be additionally made Las Vegas using the techniques from \cite[Section 4.5]{BringmannFN22}. Using this reduction, we prove the following lemma:
\begin{lemma}
\label{lem:lasvegas-small-universe}
Given vectors $A,B\in \N^N$ where $N\le \poly(\|A\|_0 \|B\|_0)$, and a parameter $t \ge \|A \star B\|_0$, 
there is a Las Vegas algorithm  that returns $C = A \star B$ in expected
 $O(t\log t)$ time. Even when $t < \|A \star B\|_0$, the algorithm is Las Vegas. 
\end{lemma}
\begin{proof}
Let $p$ be a random prime from $[0.5 t \log^5 t, t \log^5 t]$ and let $h(x) = x \bmod p$, which hashes the terms in $\supp(A\star B)$ into $p$ buckets. 
Compute $\supp(A\star_p B)$ by an $O(t\log t)$-time Monte Carlo algorithm with constant success probability (such as the algorithm in \cref{thm:generalconvo-main}).
Then, using the computed $S_0 = \supp(A\star_p B)$, we use \cref{lem:lasvegas-tiny-universe} to compute $V^{(0)} = h(A) \star_p h(B) = h(A \star B)$, and similarly compute $V^{(1)} = h(\partial A) \star_p h(B) + h(A) \star_p h(\partial B) = h(\partial (A \star B))$ and $V^{(2)} = h(\partial^2 A) \star_p h(B) + 2 h(\partial A) \star_p h(\partial B) + h(A) \star_p h(\partial^2 B)= h(\partial^2 (A \star B))$, in $O(t \log t)$ time (with failure probability $1 / \poly(t)$). 
 Here we remark that, since \cref{lem:lasvegas-tiny-universe} only takes non-cyclic vectors instead of cyclic vectors (with indices modulo $p$) as input, we need to first view these cyclic input vectors as non-cyclic ones, and use $\{i \in [2p-1]: i\bmod p \in S_0\}$ as the auxiliary set supplied to 
 \cref{lem:lasvegas-tiny-universe}; the output vectors of \cref{lem:lasvegas-tiny-universe} are then hashed back to cyclic vectors under $h(x)=x\bmod p$.

Then for every $k \in [p]$, we use the values $V^{(0)}[k],V^{(1)}[k],V^{(2)}[k]$ to perform a sparsity test (\cref{lem:momentmatrix})  to check whether bucket $k$ contains a unique term in $O(1)$ time. If the check passes, we can use $V^{(1)}[k] / V^{(0)}[k]$ to recover the index of the term, and use $V^{(0)}[k]$ to recover the value of the term. 

As $p$ is a random prime from $[0.5t \log^5 t , t \log^5 t]$, for each fixed $i \in \supp(A \star B)$, the probability that $i$ collides with some other term in a bucket is $O(t \cdot \frac{\log N / \log(t\log^5 t) }{t \log^5 t/ \log(t\log^5 t)}) = O(1/\log^4 t)$ by the prime number theorem (similarly to \cref{footnote:prime}). Therefore, the expected number of indices that collide is $O(t / \log^4 t)$, so by Markov's inequality, with constant probability, we have recovered all but $O(t / \log^4 t)$ indices. 
Then we use \cref{cor:naiverecover-full-whp} to correct the remaining $O(t / \log^4 t)$ errors and compute $A \star B$ in $O(t\log t)$ time (with failure probability $1 / \poly(t)$). 

Similarly as before, we always have $C \le A \star B$, so by verifying whether $\|A\|_1\|B\|_1 =  \|C\|_1$, we can check whether $A \star B = C$. If $\|A\|_1\|B\|_1 \ne  \|C\|_1$, we return $\bot$.
\end{proof}

\cref{lem:lasvegas-small-universe} requires knowing a good upper bound on $\|A\star B\|_0$. 
To relax this requirement, we could again use the exponential guessing strategy of \cite[Section 9]{BringmannFN21} (which we used in the end of \cref{sec:generalconvo} for our sparse general convolution algorithm). 
But here we can use the simpler method by Fischer \cite{nickphd} based on nonnegativity of $A,B$: we apply \cref{lem:estimatesumset} to obtain a constant approximation of $t=\|A\star B\|_0 = |\supp(A)+\supp(B)|$ in $O(t\log t)$ time (with $1-1/\poly(t)$ success probability), and then run \cref{lem:lasvegas-small-universe} using this upper bound. Then we immediately get the following corollary:
\begin{corollary}
\label{cor:firstuniversereduction}
Given vectors $A,B\in \N^N$ where $N\le \poly(\|A\|_0 \|B\|_0)$, there is a Las Vegas algorithm that returns $C = A \star B$ in expected $O(t\log t)$ time, where $t=\|A\star B\|_0$. 
\end{corollary}

Finally we relax the $N\le \poly(t)$ requirement by another length reduction step in {\cite[Section 4.5]{BringmannFN22}} based on almost linear hash functions $h\colon [N] \to [\poly(t)]$ (\cite[Lemma 13]{BringmannFN22}). This is summarized in the following lemma:

\begin{lemma}[Length reduction from $N$ to $\poly(t)$, {\cite[Section 4.5]{BringmannFN22}}]
\label{lem:reducetopolyt}
    Suppose there is a Las Vegas algorithm for computing $A\star B$, where $A,B\in \N^{N}$ and $N\le (\|A\|_0  \|B\|_0)^{3}$, which terminates in $T(\|A\star B\|_0)$ time with $\ge 1-\delta$ probability. 
    
    Then, there is a Las Vegas algorithm for computing $A\star B$, where $A,B\in \N^{N}$, which terminates in $O(T(\|A\star B\|_0))$ time
    with $\ge 1-\delta - 1/\|A\star B\|_0^{5}$ probability.
\end{lemma}

By combining \cref{lem:reducetopolyt} and 
\cref{cor:firstuniversereduction}, we immediately obtain a Las Vegas algorithm for computing $A\star B$ in expected $O(t\log t)$ time, where $t = \|A\star B\|_0$.

\section{Achieving High Success Probability}
\label{sec:highprob}

The Las Vegas algorithm for Sparse Nonnegative Convolution described in \cref{sec:lasvegasnonnega} terminates in $O(t\log t)$ time only with $0.9$ probability (by Markov's inequality).
In this section, we will boost this 
probability to $1-1/t$, thereby proving our main \cref{thm:lasvegas-main} under \cref{assumption1} (which will be removed in \cref{sec:numerical}).

Throughout this section, we assume the main parameter $t$ is large enough in order to simplify some calculations.

\subsection{Estimating supports via sparse convolution modulo small prime}
\label{subsec:estimatesupports}

Our algorithm with high success probability crucially relies on the following lemma for computing sparse convolution $A\star B$ given $A,B\in \F_q^N$ where both $N$ and prime $q$ are small.  
Here we measure sparsity $t$ by the size of the sumset $\supp(A)+\supp(B)$  (also called ``structural support'' in \cite{ArnoldR15,Roche18}), which may be larger than the actual support $\supp(A\star B)$ due to vanishing over $\F_q$.
In the dense regime $t=\Theta(N)$, the time complexity matches the $O(N \log q)$ bound in \cref{thm:packfft}.

\begin{restatable}[Computing $A\star B$ modulo small prime for short vectors]{lemma}{lemmodq}
\label{lem:modq}
Let $t$ be a parameter, $N\le t\log^{5} t$, and $q\in [0.5\log^{25} t, \log^{25} t]$ be a prime.
Given sparse vectors $A,B\in \F_q^{N}$,
there is a randomized algorithm in (worst case) $O(t \log \log t)$ time that outputs either a vector in $\F_q^{2N-1}$ or  $\bot$,  with the following guarantee:
\begin{itemize}
	\item If $|\supp(A) + \supp(B)|\le t$, then $\Pr[\text{output} = \bot]\le 1/\log t$. 
		\item $\Pr\big [\text{output} \neq \bot \text{ and } \text{output}\neq A\star B\big ] \le 1/t$.
\end{itemize}
\end{restatable}

Note that in \cref{lem:modq} we can assume  without loss of generality that the input vectors have sparsity $|\supp(A)|,|\supp(B)|\le t$; otherwise, it is always valid to output $\bot$. (This  assumption also applies to the upcoming \cref{cor:est} and \cref{cor:sumsetwhp}.)

\cref{lem:modq} is slightly weaker than a Las Vegas algorithm: it reports failure with moderately low probability, and outputs a wrong answer with inverse-polynomial probability.
Nevertheless,  \cref{lem:modq} is used as a subroutine in our main Las Vegas algorithm (\cref{thm:lasvegas-main}).

\cref{lem:modq} is
proved by a vast modification and refinement of \cite{BringmannFN21}; we defer its proof to \cref{sec:proofmodq}.
In the following, we show two applications of \cref{lem:modq}, which are used later in \cref{subsec:highprob} to obtain the desired high-probability Las Vegas algorithm for Sparse Nonnegative Convolution.

As a first application of \cref{lem:modq}, we give a very fast algorithm that accurately estimates the support size of (cyclic) sparse nonnegative convolution:

\begin{lemma}[Estimating $\|A\star_p B\|_0$ for small $p$]
Let $t$ be a parameter, and prime $p\le t\log^{5} t$.
Given nonnegative sparse vectors $A,B\in \N^{p}$, 
there is a randomized algorithm in (worst case) $O(t \log \log t)$ time that outputs either an integer or $\bot$, with the following guarantee:
\begin{itemize}
	\item 
 If $\|A\star_p B\|_0 \le t$, then
 $\Pr[\text{output} = \bot]\le 1/\log t$.
		\item $\Pr\Big [\text{output} \neq \bot \text{ and } \big \lvert \text{output}-\|A\star_p B\|_0 \big \rvert\ge \sqrt{\|A\star_p B\|_0}t^{0.4}\Big ] \le 1/\sqrt{t}$.
\end{itemize}
\label{cor:est}
\end{lemma}

A natural idea to prove \cref{cor:est} is as follows: replace the positive entries of $A,B$ by random $\F_q$-elements where $q=\polylog(t)$, and invoke \cref{lem:modq} to compute their convolution. Then we expect roughly $1-1/q$ fraction of $\supp(A\star_p B)$ to survive in the convolution result.
We slightly modify this idea in the following proof to make the analysis easier.
\begin{proof}[Proof of \cref{cor:est}]
 Assume $A,B\in \{0,1\}^p$ without loss of generality, as replacing the positive entries by ones does not change $\supp(A\star_p B)$.
Pick an arbitrary prime $q\in [0.5\log^{25} (8t), \log^{25} (8t)]$. 
Define random vectors
$\hat A,\hat B \in \F_q^{3p}$ as follows:
\begin{itemize}
    \item For $0\le i< p$, let $\hat A[i] = \begin{cases}\text{a random element $\in \F_q $} & A[i]=1, \\ 0 & A[i]=0. \end{cases}$
    \item For $0\le i< p$, let $\hat B[i] = \begin{cases}\text{a random element $\in \F_q $} & B[i]=1, \\ 0 & B[i]=0. \end{cases}$
    \item For $0\le i< p$, let $\hat A[i+2p] = A[i] \in \{0,1\}$.
    \item For $0\le i< p$, let $\hat B[i+2p] = B[i] \in \{0,1\}$.
    \item All other unspecified coordinates in $\hat A$ and $\hat B$ are zeros.
\end{itemize}
Let $\hat C = \hat A \star \hat B \in \F_q^{6p-1}$.
 We can verify that $|\supp(\hat A) +\supp(\hat B)| \le 3\|A\star B\|_0 \le 6\|A\star_p B\|_0$.
Moreover, the middle part of $\hat C$ encodes the information of $\supp(A\star_p B)$ in the following sense: For every $k\in [p]$, we can verify that
\begin{align}
\label{eqn:randomoutcome}
    \hat C[k+2p] + \hat C[k+3p] = \begin{cases}
       0 & (A\star_p B)[k]  = 0,\\
       \sum_{i+j\equiv_p k,\, A[i]=B[j]=1} (\hat A[i]+ \hat B[j]) & (A\star_p B)[k]  > 0.
    \end{cases}
\end{align}
Note that we included both $\hat C[k+2p]$ and $\hat C[k+3p]$ to take care of the wrap-around of $i+j$ modulo $p$.
As we can see, in the $(A\star_p B)[k]  > 0$ case, the result of $\hat C[k+2p]+\hat C[k+3p]$ is the sum of (at least one) independent uniformly random $\F_q$-elements, and hence vanishes with exactly $1/q$ probability.
Hence, the sum of indicators 
\begin{equation}
\label{eqn:defnS}
S := \sum_{k=0}^{p-1}\mathbf{1}\big [\hat C[k+2p] + \hat C[k+3p]\neq 0\big ]
\end{equation}
has expectation $\Ex[S] = (1-1/q)\|A\star_p B\|_0$.  Moreover, we can analyze the concentration of $S$ using the following claim. \begin{claim}
\label{claim:pairwise}
Assume $A,B\in\{0,1\}^{p}$ are not both all-$1$ vectors.
For every distinct $k,k'\in \supp(A\star_p B)$, the random variables $\hat C[k+2p] + \hat C[k+3p]$ and $\hat C[k'+2p] + \hat C[k'+3p]$ are pairwise independent.
\end{claim}
Note that if the assumption in \cref{claim:pairwise} does not hold,  then $A\star_p B$ is the all-$p$ vector, and we can immediately finish by outputting $\|A\star_pB\|_0 = p$.
\begin{proof}[Proof of \cref{claim:pairwise}]
Denote $c_k = \hat C[k+2p] + \hat C[k+3p]$ for short.
By \cref{eqn:randomoutcome}, both $c_k$ and $c_{k'}$ are sums of independent uniformly random $\F_q$-elements, where the summands for $c_k$ and for $c_{k'}$ may be overlapping. 
In more detail, let $I_k:= \{i\in [p]: A[i] = B[(k-i)\bmod p] = 1\}$ and $J_k:= \{j\in [p]: B[j] = A[(k-j)\bmod p] = 1\}$.
Then, by definition, $c_k$ equals the sum of random $\F_q$-elements, $\sum_{i\in I_k} \hat A[i] + \sum_{j\in J_k} \hat B[j]$, and similarly for $c_{k'}$.
In order to show $c_k$ and $c_{k'}$ are pairwise independent, it suffices to show that these two sets of summands are not identical, that is, we want to show either $I_k\neq I_{k'}$ or $J_{k} \neq J_{k'}$.  

To show this, simply note that $|I_k|=|J_k|$ and $\sum_{i\in I_k} i + \sum_{j\in J_k}j \equiv_p k\cdot |I_k|$, and similar equations hold for $k'$.  Suppose to the contrary that $(I_k,J_k)=(I_{k'},J_{k'})$. Then $|I_k|=|I_{k'}|$, and the previous property implies $k\cdot |I_k| \equiv_p k'\cdot |I_{k'}|$.
Note that $|I_k|,|I_{k'}|\ge 1$ since $k,k'\in \supp(A\star_p B)$. If $|I_k|<p$, then by division modulo prime $p$ we obtain $k\equiv_p k'$, contradicting the assumption that $k\neq k'$. If $|I_k|=p$, then it must hold that both $A,B\in\{0,1\}^{p}$ are all-$1$ vectors, which violates the initial assumption. 
\end{proof}
From \cref{claim:pairwise}, we know the random variable
$S$ defined in \cref{eqn:defnS} has variance $\Var[S] =\|A\star_p B\|_0 \cdot (1/q)(1-1/q)$. By Chebyshev's inequality, 
\begin{align*}
& \Pr\left [\, |S - \Ex[S]| \ge (1-1/q) \sqrt{\|A\star_p B\|_0}t^{0.3}\right ] \\
\le & \frac{\Var[S]}{\big ((1-1/q)\sqrt{\|A\star_p B\|_0}t^{0.3}\big )^2} = \frac{ 1}{ (q-1)t^{0.6}} .
\end{align*}
Define estimator $\tilde S := qS/(q-1)$ which has expectation $\Ex[\tilde S]= q\Ex[S]/(q-1)= \|A\star_p B\|_0$.  By scaling both sides of $|S - \Ex[S]| \ge (1-1/q) \sqrt{\|A\star_p B\|_0}t^{0.3}$ by $q/(q-1)$,  the previous inequality implies
\begin{equation}
\label{eqn:estimatorgood}
\Pr\big [ \big \lvert \tilde S - \|A\star_p B\|_0\big \rvert \ge \sqrt{\|A\star_p B\|_0}t^{0.3} \big ] \le \frac{1}{(q-1) t^{0.6}} < 1/t^{0.6}. 
\end{equation}

Now we continue to describe our algorithm: We apply the randomized algorithm from \cref{lem:modq}  to $\hat A,\hat B \in \F_q^{3p}$ with parameter $8t$, which terminates in $O(t\log \log t)$ time. If \cref{lem:modq} returns $\bot$, then we return $\bot$. Otherwise, let $\tilde C\in \F_q^{6p-1}$ be returned by \cref{lem:modq}, and we return $\frac{q}{q-1}\cdot \sum_{k=0}^{p-1}\mathbf{1}\big [ \tilde C[k+2p]+ \tilde C[k+3p]\neq 0\big ]$, rounded to the nearest integer.

Now we analyze the correctness of this algorithm. First, if $\|A\star_p B\|_0\le t$, then  $|\supp(\hat A)+\supp(\hat B)|\le 8\|A\star_p B\|_0\le 8t$, and by \cref{lem:modq} we know the probability of returning $\bot$ is at most $1/\log(8t) \le 1/\log(t)$, as desired. Second, the probability that \cref{lem:modq} returns $\tilde C$ but $\tilde C \neq  \hat A\star \hat B = \hat C$ is at most $1/(8t)$. In the successful case where $\tilde C = \hat C$, the integer we return is $\tilde S$ (up to additive difference 1 due to rounding), and by \cref{eqn:estimatorgood} we know the probability that it differs from $\|A\star_p B\|_0$ by at least $\sqrt{\|A\star_p B\|_0}t^{0.4} > \sqrt{\|A\star_p B\|_0}t^{0.3} + 1$ is no more than $1/t^{0.6}$.  By a union bound, the overall error probability is at most $1/(8t) + 1/t^{0.6} \le 1/\sqrt{t}$ as desired. This finishes the proof of correctness of \cref{cor:est}.
\end{proof}

By a similar proof to the above, we can also use \cref{lem:modq} to compute $\supp(A\star_p B)$ given $A,B\in \N^p$ with high probability of correctness, as stated in the following lemma:

\begin{lemma}[Computing $\supp(A\star_p B)$ for small $p$ w.h.p.]
Let $t$ be a parameter, and prime $p\le t\log^{5} t$.
Given sparse vectors $A,B\in \N^{p}$, there is a randomized algorithm in (worst case) $O(t \log t)$ time that outputs either a subset of $[p]$ or $\bot$, with the following guarantee:
\begin{itemize}
	\item 
 If $\|A\star_p B\|_0 \le t$, then
 $\Pr[\text{output} = \bot]\le 1/t^2$.
		\item $\Pr\big [\text{output} \neq \bot \text{ and }  \text{output}\neq \supp(A\star_p B) \big ] \le 1/\sqrt{t}$.
\end{itemize}
\label{cor:sumsetwhp}
\end{lemma}
\begin{proof}
Again, we assume without loss of generality that $A,B\in \{0,1\}^p$.
Our algorithm has $\ell := \lceil 10\log t/\log \log t\rceil$ independent rounds, each with worst case time complexity $O(t\log \log t)$. The total time complexity is thus $O(t\log t)$. 

Pick an arbitrary prime $q\in [0.5\log^{25} (8t), \log^{25} (8t)]$. In the $r$-th round ($r\in [\ell]$), we construct random vectors $\hat A^{(r)},\hat B^{(r)}\in \F_q^{3p}$, and apply \cref{lem:modq} to $\hat A^{(r)},\hat B^{(r)}$ with parameter $8t$, in the same way as in the proof of \cref{cor:est}. 
We denote $\hat C^{(r)} = \hat A^{(r)} \star \hat B^{(r)}\in \F_{q}^{6p-1}$, and we denote the output of \cref{lem:modq} by $\tilde C^{(r)}$, which is either $\bot$ or a vector in $\F_q^{6p-1}$. Recall from \cref{eqn:randomoutcome} that for each $k\in [p]$, $\hat C^{(r)}[k+2p] + \hat C^{(r)}[k+3p]$ is a random $\F_q$-element if $(A\star_p B)[k]\neq 0$, and zero otherwise.

After all $\ell$ rounds are finished: 
\begin{itemize}
    \item If more than $\ell/3$ rounds returned $\tilde C^{(r)} = \bot$, then we return $\bot$.
    \item Otherwise, we return the set 
\begin{equation}
\label{eqn:returnset}
    \{k\in [p]: \text{exists } r\in [\ell] \text{ s.t. } \tilde C^{(r)}\neq \bot \text{ and } \tilde C^{(r)}[k+2p] + \tilde C^{(r)}[k+3p]\neq 0 \}.
\end{equation}
\end{itemize}

To analyze the correctness of this algorithm,   we need the following three claims.
Recall that we assume $t$ is large enough in order to simplify calculations.
\begin{claim}
\label{claim:goodmodq}
   With at least $1-1/t^{18}$ probability, for every $k\in \supp(A\star_p B)$, the number of $r\in [\ell]$ such that $\hat C^{(r)}[k+2p] + \hat C^{(r)}[k+3p]\neq 0$  is at least $\ell/2$.
\end{claim}
\begin{proof}
   Fix  $k\in \supp(A\star_p B)$. Since the $\ell$ rounds are independent, and in each round
   $\hat C^{(r)}[k+2p] + \hat C^{(r)}[k+3p]$
   is a random $\F_q$-element, the probability that more than $\ell/2$ of them  are zeros is at most 
   \[ \binom{\ell}{\ell/2}\cdot (1/q)^{\ell/2} < (4/q)^{\ell/2} < \left (\frac{4}{0.5\log^{25} (8t)}\right )^{5\log t/\log \log t} \le \frac{1}{t^{20}}.\]
   We apply a union bound over at most $\|A\star_p B\|_0 \le p \le t\log^5 t$ many $k$'s, and the total failure probability is at most $1/t^{18}$.
\end{proof}
\begin{claim}
\label{claim:fewfail}
    Suppose $\|A\star_p B\|_0\le t$. Then, with at least $1-1/t^2$ probability,  the number of rounds $r$ with $\tilde C^{(r)}=\bot $ is at most $\ell/3$.
\end{claim}
\begin{proof}
   If  $\|A\star_p B\|_0\le t$, then, similarly to the argument in the proof of \cref{cor:est}, we know each round returns $\bot$ with at most $1/\log(8t)$ probability as guaranteed by \cref{lem:modq}. Then the probability that more than $\ell/3$ rounds returned $\bot$ is at most $\binom{\ell}{\ell/3} (1/\log(8t))^{\ell/3} \le (8/\log(8t))^{\ell/3} \le 1/t^2$.
\end{proof}
\begin{claim}
\label{claim:nowrong}
    With at least $1-1/t^{0.9}$ probability, for all $r\in [\ell]$, either $\tilde C^{(r)} = \bot $ or $\tilde C^{(r)} = \hat C^{(r)}$.
\end{claim}
\begin{proof}
  By the second bullet point of \cref{lem:modq}, for each round $r$, the probability that $\tilde C^{(r)} \neq \bot $ and $\tilde C^{(r)} \neq \hat C^{(r)}$ is at most $1/(8t)$. The claim then follows from a union bound over $\ell =\lceil 10\log t/\log \log t\rceil $ rounds.
\end{proof}

Now we analyze the correctness of our algorithm. 
Since we return $\bot$ only if more than $\ell/3$ rounds returned $\tilde C^{(r)} = \bot$,  the first bullet point in the lemma statement immediately follows from \cref{claim:fewfail}. 

Now we prove the second bullet point. We assume the events in \cref{claim:goodmodq} and \cref{claim:nowrong} both happen (which holds with at least $1-1/t^{0.9}-1/t^{18} \ge 1-1/\sqrt{t}$ probability by a union bound).
We also assume $\le \ell/3$ of the rounds returned $\bot$ (otherwise, our algorithm would return $\bot$ in the end,  and there is nothing to prove). In this case, for every $k\in \supp(A\star_p B)$, by \cref{claim:goodmodq} there are at least $\ell/2$ rounds $r$ with $\hat C^{(r)}[k+2p]+ \hat C^{(r)}[k+3p]\neq 0$, among which there are only $<\ell/3$ rounds with $\tilde C^{(r)} = \bot$, and the remaining $ \ell/2-\ell/3\ge 1$ non-$\bot$ rounds must satisfy $\tilde C^{(r)}[k+2p]+\tilde C^{(r)}[k+3p] = \hat C^{(r)}[k+2p]+\hat C^{(r)}[k+3p]\neq 0$ by \cref{claim:nowrong}, and thus $k$ is included in our returned set \cref{eqn:returnset}. For every $k\in [p] \setminus \supp(A\star_p B)$, we know $\hat C^{(r)}[k+2p]+ \hat C^{(r)}[k+3p]=0$ always holds, and then by \cref{claim:nowrong} we know 
$\tilde C^{(r)}[k+2p]+ \tilde C^{(r)}[k+3p]=0$ holds whenever $\tilde C^{(r)}\neq \bot$, and hence $k$ is not included in our returned set \cref{eqn:returnset}. This finishes the proof of the second bullet point.
\end{proof}

\subsection{Sparse nonnegative convolution with high success probability}
\label{subsec:highprob}
Now we will use \cref{cor:est} and \cref{cor:sumsetwhp} to boost \cref{lem:lasvegas-small-universe} to succeed with high probability. 

Recall that in the proof of \cref{lem:lasvegas-small-universe} we sample a random $p \le t\polylog(t)$ to reduce the length of the input vectors to $t\polylog(t)$ while incurring $\le t/\polylog(t)$ collisions in expectation.
In the following, we use \cref{cor:est} to find a prime $p$ that incurs $\le t/\polylog(t)$ collisions with high probability.

\begin{lemma}[Finding a good $p$ w.h.p.]
\label{lem:findpwhp}
Let $t$ be a parameter, $A,B \in \N^N$ be input vectors with $\|A\|_0,\|B\|_0 \le t$, and $N \le \poly(t)$. There is a randomized algorithm that outputs a prime $p \le t\log^5 t$ in (worst case) $O(t\log t)$ time, such that:
 If $\|A\star B\|_0\le t$, then with  at least $1-1/t^{0.3}$ probability,
 \begin{equation}
 \label{eqn:fewcol}
     |\{ k\in \supp(A\star B): \text{exists } k'\in \supp(A\star B) \setminus \{k\} \text{ s.t. }  k'\equiv_p k \}| \le O\big (\tfrac{t}{\log^3 t}\big ).
 \end{equation}
\end{lemma}
\begin{proof}
Denote the left hand side of \cref{eqn:fewcol} by $f_p$. We have the following bound on $f_p$ over a random prime $p$. 
\begin{claim}
Suppose $\|A\star B\|_0\le t$. Then, for a uniform random prime  $p \in [0.5 t\log^5 t, t\log^5 t]$,  we have  $\Pr_p[f_p \le \frac{c t}{\log^3 t}] \ge 1-\frac{1}{\log t}$ for some constant $c>0$.
\label{claim:primegood}
\end{claim}
\begin{proof}
For each $0\le k<k'\le 2N - 2 \le \poly(t)$ we have $\Pr_p[p\mid k'-k] \le \frac{\log 2N / \log (0.5t\log^5 t)}{\Omega(t\log^5 t/\log (t\log^5 t))}   = O(\frac{1}{t\log^4 t})$ by the prime number theorem (similarly to \cref{footnote:prime}).
Then, for each $k\in \supp(A\star B)$, by a union bound, the probability that there exists $k'\in \supp(A\star B)\setminus \{k\}$ with $k'\equiv_p k$ is at most $\|A\star B\|_0 \cdot O(\frac{1}{t\log^4 t}) \le O(\frac{1}{\log^4 t})$.
By linearity of expectation,
we know $\Ex_p[f_p] \le O(\frac{t}{\log^4 t})$.
By Markov's inequality, we know $f_p \le c t/\log^3 t$ holds with at least $1- 1/\log t$ probability where $c$ is some constant. 
\end{proof}

We can relate $f_p$ with $\|A\star_p B\|_0 $ via the following property 
\begin{equation}
\label{eqn:relation}
\frac{f_p}{2} \le \|A \star B\|_0 -   \|A \star_p B\|_0 \le f_p.
\end{equation}
To show \cref{eqn:relation}, consider bucketing $\supp(A\star B)$ modulo $p$, and observe that each non-empty bucket of size $x$ ($x\ge 1$) contributes $x-1$ to the quantity $\|A \star B\|_0 -   \|A \star_p B\|_0$ and contributes $x\cdot \mathbf{1}[x\ge 2]$  to $f_p$, and use the inequality $\frac{1}{2}x\cdot \mathbf{1}[x\ge 2]\le x-1 \le x\cdot \mathbf{1}[x\ge 2]$ to get \cref{eqn:relation}.

Informally speaking, our plan is as follows: \cref{claim:primegood} means  that a random sample of $O(\log t/\log \log t)$ primes $p$ will contain a good prime with $1-1/\poly(t)$ probability, and our goal is to efficiently find a good prime among these sampled primes. To do this, we will use \cref{cor:est} to estimate the size of $\|A\star_p B\|_0$ for each sampled prime $p$, and keep the prime $p$ that (approximately) maximizes $\|A\star_p B\|_0$ (which means small $f_p$ by \cref{eqn:relation}).

In more detail, our algorithm works as follows:
\vspace{0.2cm}

For $r = 1,2,\dots, \ell:= \lceil 5\log t/\log \log t\rceil$: 
\begin{itemize}
    \item  Sample a uniform random prime $p_r \in [0.5 t\log^5 t, t\log^5 t]$. Let $\tilde A:= h(A),\tilde B=h(B) \in \N^{p_r}$ where $h(x)=x\bmod p_r$.
    \item  Repeatedly run the algorithm from \cref{cor:est} on $\tilde A,\tilde B \in \N^{p_r}$ with parameter $t$, until the output is not $\bot$. Denote the output integer by 
    $\tilde g_{p_r}$.
\end{itemize}

Finally, return the prime $p_{r^*}$ with maximum $\tilde g_{p_{r^*}}$.

If the algorithm runs for longer than $O(t\log t)$ time, then we abort and output an arbitrary prime.
\\

We first analyze the running time of the algorithm described above.
\begin{claim}
\label{claim:canfinish}
The algorithm described above (without aborting) finishes in time $O(t\log t)$ with at least $1-1/t$ probability.
\end{claim}
\begin{proof}
It is easy to see that the total time for generating $\ell$ primes  is $\ell \cdot  \polylog(t)\le \polylog(t)$ with probability $1-1/t^{100}$, and is negligible. In the following we focus on the inner for-loop that invokes \cref{cor:est}.

Each call of \cref{cor:est} returns $\bot$ with only $\le 1/\log t$ probability.
Since our algorithm runs the inner for-loop $\ell$ times, we need $\ell$ calls of
\cref{cor:est} to return non-$\bot$ in order for our algorithm to finish (before being aborted).
Among the first (up to) $2\ell$ calls of \cref{cor:est} made by the algorithm, the probability that only $<\ell$ of these calls returned non-$\bot$ can be bounded by
$ \binom{2 \ell}{2 \ell - \ell} \cdot (1/\log t)^{2\ell - \ell} \le  (4/\log t)^{\ell} \le t^{-2}$.  Hence, with $1-t^{-2}$ probability, the total time for the inner for-loop is at most $\ell \cdot O(t\log \log t) \le O(t\log t)$.

 In summary, our algorithm  finishes in $O(t\log t)$ time (without being aborted) with $1-t^{-100}-t^{-2} \ge 1-1/t$ probability.
\end{proof}

Now we analyze the correctness of the algorithm.
We only need to consider the $\|A\star B\|_0 \le t$ case.
\cref{cor:est} guarantees that each $\tilde g_{p_r}$ is correct with probability $\ge 1-1/\sqrt{t}$, where being correct means differing from $\|\tilde A\star_{p_r} \tilde B\|_0 =\|A\star_{p_r} B\|_0 $ by at most an additive $\sqrt{\|\tilde A\star_{p_r} \tilde B\|_0} t^{0.4} \le \sqrt{\|A\star B\|_0} t^{0.4} \le t^{0.9} $.
By a union bound over the first (up to) $2\ell$ calls of \cref{cor:est} (which we know is sufficient by the proof of \cref{claim:canfinish}), all $\tilde g_{p_r}$ are correct with probability at least $1-2\ell/\sqrt{t} \ge 1-1/t^{0.4}$. In this case, the final $p_{r^*}$ returned by our algorithm, which maximizes $\tilde g_{p_{r^*}}$, should also approximately maximize $\|A\star_{p_{r^*}} B\|_0$, i.e., 
\begin{equation}
\label{eqn:approxmax}
 \|A\star_{p_{r^*}} B\|_0 \ge \|A \star_{p_{r}} B\|_0 - 2t^{0.9} \text{ for all $1\le r \le \ell$.}
\end{equation}

On the other hand, by \cref{claim:primegood}, we know with at least $1- (1/\log t)^{\ell} \ge 1-1/t^5$ probability, there is some sampled prime $p_r$ $(1\le r\le \ell)$ such that $f_{p_r} \le \frac{c t}{\log^3 t}$. Combine this with \cref{eqn:approxmax} and \cref{eqn:relation}, and we get
\begin{align*}
 f_{p_{r^*}} & \le 2( \|A\star B\|_0 - \|A \star_{p_{r^*}} B\|_0 )\\
 & \le 2(\|A\star B\|_0 -  \| A \star_{p_r} B\|_0 + 2t^{0.9} ) \\
 & \le 2(f_{p_r} +2t^{0.9})\\
& \le 2 (\tfrac{c t}{\log^3 t} + 2t^{0.9})\\
& \le \tfrac{6c t}{\log^3 t},
\end{align*}
which means the prime $p_{r^*}$ found by our algorithm indeed satisfies the desired property.
Finally we apply a union bound over all the bad events mentioned earlier. The probability that the algorithm fails to run within time $O(t \log t)$ is $\le 1/t$ by \cref{claim:canfinish}, the probability some call of \cref{cor:est} fails is $\le 1/t^{0.4}$, and the probability that none of the sampled primes $p_r$ has $f_{p_r} \le \frac{ct}{\log^3 t}$ is $\le 1/t^5$. Hence, the overall error probability is $\le 1/t + 1/t^{0.4} + 1/t^{5} \le 1/t^{0.3}$ for sufficiently large $t$. This finishes the proof of \cref{lem:findpwhp}.
\end{proof}

Finally, we are ready to prove the main theorem, under \cref{assumption1} (which is required by \cref{lem:lasvegas-tiny-universe}).

\lasvegasmain*

\begin{proof}
By the length reduction step of \cref{lem:reducetopolyt}, it suffices to solve the case where $N\le (\|A\|_0  \|B\|_0)^{3}$.
Recall that in the nonnegative setting we have $t = \|A\star B\|_0 \in [\max\{\|A\|_0,\|B\|_0\}, \|A\|_0\cdot  \|B\|_0]$.

In the following, we describe a Las Vegas algorithm that solves the problem with the following guarantee: it runs in worst-case $O(t\log t)$ time, and either outputs the correct result or outputs $\bot$, and the probability of outputting $\bot$ is at most $1/\poly(t)$. 
\begin{enumerate}
    \item 
    We first apply the algorithm from \cref{lem:estimatesumset} to $X=\supp(A),Y=\supp(B)$ and $\delta = 1/(\|A\|_0+\|B\|_0)^{2}\le 1/t$, to obtain an approximation $t^*$ that always satisfies $t^* \le O(t)$, and satisfies $t \le t^*$  with $1-\delta\ge
1- 1/t$ probability.  
 The running time is $O(t \log(1/\delta) + t\log \log t) = O(t \log t)$. 
\item  Then, we run \cref{lem:findpwhp} (using parameter $t^*$) to find a good hash function $h(x) = x \bmod{p}$ where $p \le t\polylog(t)$, in $O(t\log t)$ time. Assuming the estimate $t^*$ from the previous step is successful (i.e., $t\le t^*$), with at least $1-1/t^{0.3}$ success probability it holds that hashing $\supp(A\star B)$ to $[p]$ only incurs $O(t/\log^3 t)$ collisions. 
 \item
Then, we run \cref{cor:sumsetwhp} (using parameter $t^*$) to compute $S = \supp(h(A)\star_p h(B))$ in $O(t\log t)$ time.  If \cref{cor:sumsetwhp} returns $\bot$, then we terminate the algorithm and return $\bot$.
By \cref{cor:sumsetwhp}, assuming the estimate $t^*$ from the first step is successful, the probability of returning $\bot$ is at most $1/t^2$.
Moreover, \cref{cor:sumsetwhp} states that the probability of computing an incorrect (and non-$\bot$) $S$ is at most $1/t^{0.5}$.
\item 
 Using the set $S$ computed by the previous step, we run \cref{lem:lasvegas-tiny-universe} to compute $h(A) \star_p h(B)$, $h(\partial A) \star_p h(B) + h(A) \star_p h(\partial B)$ and $h(\partial^2 A) \star_p h(B) + 2 h(\partial A) \star_p h(\partial B) + h(A) \star_p h(\partial^2 B)$, in $O(t \log t)$ time.  
(Since \cref{lem:lasvegas-tiny-universe} only takes non-cyclic vectors as input, we need to first transform these cyclic vectors to non-cyclic ones in the same way as we did in the proof of  \cref{lem:lasvegas-small-universe}.)

 By the guarantee of \cref{lem:lasvegas-tiny-universe}, either these vectors are computed correctly, or \cref{lem:lasvegas-tiny-universe} returns $\bot$ at least once, in which case we also terminate the algorithm and return $\bot$. Moreover, if the support $S$ computed by the previous step is correct, then the probability of returning $\bot$ is at most $O(1/t)$.
\item
Given
$h(A) \star_p h(B)$, $h(\partial A) \star_p h(B) + h(A) \star_p h(\partial B)$ and $h(\partial^2 A) \star_p h(B) + 2 h(\partial A) \star_p h(\partial B) + h(A) \star_p h(\partial^2 B)$ computed by the previous step (which must be correct, if the algorithm reaches this step), we use the same method as in the proof of \cref{lem:lasvegas-small-universe} to recover all the entries in $A\star B$ that do not have collisions using the hash function $h(x)=x\bmod p$.  Let $C$ denote the vector consisting of the recovered entries in $A\star B$. Then we have $C\le A\star B$ (always). Moreover,
if the choice of $p$ at the second step was successful, then $\|A\star B-C\|_0 \le O(t/\log^3 t)$.
\item 
Finally, apply \cref{cor:naiverecover-full-whp} to $A,B,C$ and a parameter $m =O(t^*/\log^3 t^*)$, and obtain a vector $R$ in $O(t \log t)$ time such that $R+C\le A\star B$ (always).
Moreover, if $\|A\star B-C\|_0 \le O(t/\log^3 t)$ holds, then $R+C=A\star B$ holds with $1-1/\poly(t)$ probability.
We return $R+C$ as the answer if it passes the check $\|R+C\|_1 = \|A\|_1\|B\|_1$, and return $\bot$ otherwise.
\end{enumerate}
By the discussion above, it is clear that the algorithm either returns the correct answer $A\star B$ or returns $\bot$. Moreover, by a union bound over all the steps, the probability of returning $\bot$ can be bounded by $1/\poly(t)$.

The main theorem directly follows from running a constant number of repetitions of this Las Vegas algorithm to reduce the overall failure probability from $1/\poly(t)$ to $1/t$.
\end{proof}

\subsection{Proof of \texorpdfstring{\cref{lem:modq}}{Lemma~\ref{lem:modq}}}
\label{sec:proofmodq}

In this section we prove \cref{lem:modq}  using techniques from \cite{BringmannFN21}, Prony's method with bit packing (\cref{lem:pronybitpack}), as well as some other tricks.

One of the key ingredients in the work of Bringmann, Fischer, and Nakos \cite{BringmannFN21} is an (almost) 3-wise independence result of the hash family $h(x) = ((\sigma x+\tau)\bmod p) \bmod m$, which improved an earlier result of Knudsen~\cite{Knudsen16} in some specific regime of interest. We will use this result in the same way as in \cite{BringmannFN21} to analyze the number of elements hashed to heavy buckets.

\begin{definition}[Almost linear hash family $h$]
\label{defn:almostlinearhash}
Let $p$ be a fixed prime, and $m\le p$ be the number of buckets. 
 Let $\sigma,\tau \in [p]$ be chosen uniformly and independently at random, and define the hash function 
	$h(x) := \pi(x) \bmod m$, where $\pi(x):= (\sigma x + \tau ) \bmod p$.
\end{definition}
\begin{restatable}[Overfull buckets, follows from \cite{BringmannFN21}]{lemma}{lemoverfull}
    Let $X\subseteq [N]$ be a set of keys, and let $t\ge |X|$.
 Randomly pick a hash function $h(x) = \pi(x) \bmod m$ from \cref{defn:almostlinearhash} with parameters $p > 4N^2$ and $m \le  N$.
 Let $\tilde X := \pi(X) + \{0,p\}$. Then, with at least $1-O(1/\log^2 N)$ probability,
 it holds that
 \begin{equation}
 \label{eqn:overfulltilde}
\Big \lvert \Big \{ y \in \tilde X :  \sum_{y' \in \tilde X} \mathbf{1}  \left[ y\equiv y'\pmod{m} \right] > \frac{4t}{m} \Big \} \Big \rvert \le \frac{mN\log^3 N}{t}.
 \end{equation}
 \label{lem:overfullbuckets}
\end{restatable}
\cref{lem:overfullbuckets} follows from
{\cite[Corollary 11.3, Lemma 5.5]{BringmannFN21}}
 with a slight change of parameters. We redo the calculation and prove \cref{lem:overfullbuckets} in \cref{sec:overfullbuckets}.

Now we proceed to prove \cref{lem:modq} (restated below for convenience).
\lemmodq*

\paragraph*{Preparation.} We first set up a few parameters and introduce some definitions.
In our algorithm we will
sample an almost linear hash function (\cref{defn:almostlinearhash}) $h(x) = \pi(x)\bmod m = ((\sigma x + \tau) \bmod p) \bmod m$, where the number of buckets is $m = \lceil t/\log^{12} t\rceil$. We will apply \cref{lem:overfullbuckets}  to this hash function with parameter $2N$ in place of $N$.
We set the prime $p$ in the almost linear hash family to be an arbitrary prime $p\in [20N^2,40N^2]$, so that the $p>4(2N)^2$ requirement in \cref{lem:overfullbuckets} is satisfied (such prime $p$ can be found by a Las Vegas algorithm in $\polylog(N) \le \polylog(t)$ time with $1/\poly(t)$ failure probability). We also assume $N > t$ holds in \cref{lem:modq}, so that the $m \le 2N$ requirement in \cref{lem:overfullbuckets} is satisfied (in the $N\le t$ case, \cref{lem:modq} immediately follows from the dense bit-packing FFT algorithm of \cref{thm:packfft} in $O(N\log q)\le O(t\log \log t)$ time).

\newcommand{\OneShotRecovery}{\textsc{OneShotRecovery}}
\newcommand{\RecoverBucketFast}{\textsc{RecoverBucketFast}}

Recall $q\in [0.5 \log^{25} t, \log^{25} t]$ is a prime. We find a primitive root $\omega\in \F_q^*$, by a deterministic brute force algorithm in $\poly(q) = \polylog(t)$ time.

Similar to \cite{BringmannFN21}, our algorithm contains two stages. The first stage is to invoke \OneShotRecovery (\cref{alg:recoverthreewise}) and approximately recover $A\star B$ with at most $t/\polylog(t)$ errors, and the second stage is to fix the remaining errors using \cref{lem:naiverecover-full}.

\begin{algorithm}
\DontPrintSemicolon
\caption{\textnormal{\OneShotRecovery}$(A,B,t)$}
\label{alg:recoverthreewise}
\textbf{Input:} Sparse vectors $A,B\in \F_q^N$ where $N\le t\log^5 t$ and prime $q\in [0.5 \log^{25} t, \log^{25} t]$, and a parameter $t$\;
\textbf{Output:} A sparse vector $R \in \F_q^{2N-1}$, for details see \cref{lem:recoverthreewise-step}\;
Let parameter $m =\lceil t/\log^{12} t \rceil$, $v = \lceil \log^{22} t\rceil $.\;
Sample an almost linear hash function $h(x):= \pi(x) \bmod m$ from \cref{defn:almostlinearhash}\;
Let parameter $s = \lceil 8t/m\rceil \le O(\log^{12} t)$.\label{line:oneshotbucketsize}\;
\For{$i \gets 0,1,2,\dots, 2s$\label{Line:recoverthreewise:forloop1}}{
    Define vector $ A^{(i)}\in \F_q^N$ by $ A^{(i)}[k]:= A[k]\cdot \omega^{(k\bmod v)\cdot i}\in \F_q$ \tcp{$\omega\in \F_q^*$ is a primitive root}
    Define vector $ B^{(i)}\in \F_q^N$ by $ B^{(i)}[k]:= B[k]\cdot \omega^{(k\bmod v)\cdot i}\in \F_q$\;
    Compute vector $Z^{(i)}\gets h(A^{(i)}) \star_m h(B^{(i)}) $ via FFT \label{line:recoverthreewise:FFT2}\;
}
Initialize $R \gets (0,\dots,0)\in \F_q^{2N-1}$\;
\For{$b\in [m]$}{
    $f(x)\gets \textnormal{\RecoverBucketFast}\big (\{Z^{(i)}[b]\}_{i=0}^{2s}, b\big )$ (see \cref{alg:recoverbucketthreewise}) \label{line:recoverbucketfast} \label{line:recoverthreewise:recover-bucket} \tcp{$f(x)\in \F_q[x]$ is a polynomial in sparse representation}
    $R \gets R + (f_0,\dots,f_{2N-2})$, where $f(x) = \sum_{\ell=0}^{2N-2}f_\ell x^\ell$\;
}
\Return{$R$}\;

\end{algorithm} 

We first state the properties of 
{\OneShotRecovery}  (\cref{alg:recoverthreewise}), and use it to quickly derive \cref{lem:modq}. 
\begin{lemma}
\label{lem:recoverthreewise-step}
In \cref{alg:recoverthreewise}, if $|\supp(A) + \supp(B)|\le t$, then, with probability at least $1-O(1/\log^2 t)$,
\textnormal{\OneShotRecovery}$(A,B,t)$ 
terminates in $O(t\log\log t)$ time and returns a vector $R\in \F_q^{2N-1}$ such that $\|A\star B - R\|_0 \le t/\log^{3} t$.
\end{lemma}

\begin{proof}[Proof of \cref{lem:modq} assuming \cref{lem:recoverthreewise-step}]

We first run {\OneShotRecovery}  (\cref{alg:recoverthreewise}) on $A,B\in \F_q^{N}$ with parameter $t$, and obtain vector $R\in \F_q^{2N-1}$. If it does not terminate within $O(t\log \log t)$ time, then we abort and return $\bot$. By \cref{lem:recoverthreewise-step}, with $1-O(1/\log^2 t)$ probability we did not return $\bot$ and $R$ satisfies $\|A\star B - R\|_0 \le t/\log^{3} t$.
We assume $\|R\|_0 \le 2t$, or abort and return $\bot$ otherwise.

Then,  we apply the algorithm from \cref{lem:modprimerecoverfinitefield} to $A,B,R$ with parameter $t$ and error probability $\delta = 1/\log^3 t$. The running time is $O(\|A\|_0+\|B\|_0+\|R\|_0+t)\cdot \log (1/\delta)  = O(t\log \log t)$.
Let $C\in \F_{q}^{2N-1}$ be the vector returned by \cref{lem:modprimerecoverfinitefield}.

Finally, we verify $A\star B = C$ using the randomized sparse verification procedure from \cref{lem:sparseverifysmallp} in $O(t)$ time with $\le 1/t$ error probability. If $A\star B\neq C$ then we abort and return $\bot$. This guarantees the second bullet point.
\end{proof}

Now we start to describe {\OneShotRecovery} (\cref{alg:recoverthreewise}) in more detail.
Compared to the algorithm of \cite{BringmannFN21}, a main technicality of our algorithm lies in the procedure for recovering the terms hashed into each bucket ({\RecoverBucketFast} in \cref{alg:recoverbucketthreewise}, invoked at \linereff{line:recoverbucketfast} of \cref{alg:recoverthreewise}). We need to use information computed over a small field $\F_q$ of size $q\le \polylog(t)$ to recover the exponent of a term hashed to a bucket, which could be as large as $N \le t\polylog(t)$. This recovery is possible because the bucket number $b\in [m]$ restricts the possibilities of the exponent, but a lot of technicalities arise from dealing with the almost linear hash function 
$h(x) = ((\sigma x + \tau) \bmod p) \bmod m$. In the following we first describe this recovery process in detail.

\begin{algorithm}
\DontPrintSemicolon
\caption{\RecoverBucketFast$(v_0, v_1,\dots, v_{2s}, b)$}
\label{alg:recoverbucketthreewise}
\textbf{Input:} Values $ v_0,\dots, v_{2s}\in \F_{q}$ and $b\in [m]$\;
\textbf{Output:} a polynomial $f(x) \in \F_q[x]$ represented as a list of non-zero coefficients. For details see \cref{lem:recoverbucketthreewise}\;
Use Prony's method with bit-packing (\cref{lem:pronybitpack}) with parameter $t$ and sparsity bound $s$ (defined at \linereff{line:oneshotbucketsize} of \cref{alg:recoverthreewise}) to obtain $g(y)=\sum_{j=1}^{s'}c_j y^{e_j}\in \F_q[y]$ (where $s'\le s$) such that $g(\omega^i) = v_i $ for all $0\le i\le 2s$.\;
\lIf{\cref{lem:pronybitpack} failed to find such a $g(y)\in \F_q[y]$}{\Return{$0$}}
\For{$j \in [s']$}{
    $d = \textnormal{\textsc{QueryIndex}}(e_j\bmod v,b)$ (see \cref{lem:queryindex})\;
    \lIf{$d\neq \bot $}{$d_j \gets d$ and $a_j\gets c_j$}
    \lElse{$a_j \gets 0$}
}
\Return{$f(x) =\sum_{j=1}^{s'} a_j x^{d_j}$}\;

\end{algorithm}

\paragraph*{Index recovery.}
Let parameter $v  = \lceil \log^{22} t\rceil $ (also defined in \cref{alg:recoverthreewise}). 
By inspecting the first \textbf{for} loop in \cref{alg:recoverthreewise}, we have the following observation:
\begin{observation}
In \cref{alg:recoverthreewise}, for each bucket $b\in [m]$, the values of $\{Z^{(i)}[b]\}_{i=0}^{2s}$ are evaluations of the polynomial  
\begin{equation}
\label{eqn:bucketpoly}
\sum_{j,k\in [N]: h(j)+h(k)\equiv_m b}A[j]B[k]\cdot y^{(j \bmod v) + (k\bmod v)} \in \F_q[y]
\end{equation}
 at points  $\{\omega^i\}_{i=0}^{2s}$.
\end{observation}
At a high level, our goal is to try to recover the index $j+k$ from $(j+k)\bmod v$ (which we know from the exponents in the polynomial \cref{eqn:bucketpoly}) and $(h(j)+h(k))\bmod m$ (which we know from the bucket number $b\in [m]$).

Every index $\ell\in [2N]$ can be represented as $\lfloor \ell/v\rfloor v + (\ell\bmod v)$, where $0\le \lfloor \ell/v \rfloor \le 2N/v$. We first need the following technical definition.
\begin{definition}[Bad index]
\label{defn:badindex}
 Index  $\ell\in [2N]$ is called \emph{bad} with respect to the hash function $h(x) = ((\sigma x + \tau)\bmod p)\bmod m$, if there exist $o\in \{-2p,-p,0,p,2p\}$ and integer $0\le w'\le 2N/v$, such that $w' \neq \lfloor \ell/v \rfloor $ and 
 \begin{equation}
 \big ((\sigma \lfloor \ell/v\rfloor v +\tau)\bmod p\big ) + o \equiv_m \big ((\sigma w' v  +\tau)\bmod p\big ).
 \label{eqn:badindex}
 \end{equation}
\end{definition}
\begin{lemma}[Probability of being a bad index]
\label{lem:probbadindex}
For each index  $\ell\in [2N]$, $\Pr_{\sigma,\tau}[\ell\text{ is bad}] \le O(1/\log^{5} t)$.
\end{lemma}
\begin{proof}
Since $2N<p$, the integers between $0$ and $2N$ are distinct in $\F_p$.
    For every $w'\neq \lfloor \ell/v \rfloor$, the two values $((\sigma \lfloor \ell/v\rfloor v +\tau)\bmod p ), ((\sigma w' v +\tau)\bmod p )\in \F_p$ are pairwise independent (over the randomness of $\sigma,\tau \in \F_p$). Hence, for fixed $o\in \Z$, the probability (over $\sigma,\tau$) that \cref{eqn:badindex} happens is at most $\frac{p\cdot \lceil p/m\rceil}{p^2}\le 2/m$. Hence, by a union bound over $o\in \{-2p,-p,0,p,2p\}$ and all integers $0\le w'\le 2N/v$, we know the probability that $\ell$ is bad is at most $\frac{2}{m} \cdot 5\cdot (2N/v) \le \frac{20t\log^{5} t}{\lceil t/\log^{12} t\rceil \cdot \lceil \log^{22}t\rceil } = O(1/\log^{5} t)$. %
\end{proof}

\begin{lemma}[Recovering index $(j+k)$]
   Given the hash function   $h(x) = ((\sigma x + \tau) \bmod p) \bmod m$, after $o(t)$ time preprocessing, we can answer the following query in $O(1)$ time:  
   \begin{itemize}
       \item \textnormal{\textsc{QueryIndex}}$(u,f)$: Given the values  $u = (j+k)\bmod v$ and  $f = (h(j)+h(k))\bmod m$ (where $j\in [N],k\in [N]$ are unknown),  answer either $j+k$ or $\bot$ (``failure''). 
       We allow the answer to be $\bot$ only if $j+k$ is bad with respect to the hash function $h$.
   \end{itemize}
   \label{lem:queryindex}
\end{lemma}
\begin{proof}
Note that 
\begin{equation}
\label{eqn:hij}
    f\equiv_m \big ((\sigma j + \tau)\bmod p\big ) + \big ( (\sigma k + \tau)\bmod p\big ) \in \big ((\sigma(j+k) + 2\tau)\bmod p\big ) + \{0,p\}.
\end{equation}
Rewrite $j+k = w\cdot v+u$, where $w = \lfloor (j+k)/v\rfloor \in \N\cap [0, 2N/v]$ is unknown and $u = (j+k)\bmod v$ is known.
Then, $\sigma(j+k)+2\tau = (\sigma wv+\tau) + (\sigma u+\tau)$, and thus the right hand side of \cref{eqn:hij} (viewed over $\Z$) further satisfies 
\begin{align}
 & \big ((\sigma(j+k) + 2\tau)\bmod p\big ) + \{0,p\} \nonumber \\
   \subseteq  \ &
   \big ( (\sigma wv+\tau)\bmod p\big ) + \big ((\sigma u + \tau)\bmod p\big ) + \{0,-p\} + \{0,p\}\nonumber \\
=  \ &
   \big ( (\sigma wv+\tau)\bmod p\big ) + \big ((\sigma u + \tau)\bmod p\big ) + \{-p,0,p\}.\label{eqn:sigamwv} 
\end{align}

In order to find $j+k$ (or equivalently, find $w$), we can use \cref{eqn:hij} and \cref{eqn:sigamwv} to pin down the possibilities for $w$: Define 
\begin{align}
\label{eqn:defnW}
 W:= \big \{w' \in & \N \cap [0,2N/v] : \exists o\in \{-p,0,p\} \\
 & \text{ s.t.\ } \big ( (\sigma w'v+\tau)\bmod p\big ) \equiv_m f-\big ((\sigma u + \tau)\bmod p\big )-o\big \}. \nonumber
\end{align}
Then we must have $w\in W$ by \cref{eqn:hij} and \cref{eqn:sigamwv}. We will later describe how to find $W$ time-efficiently. If $|W|=1$, then we have uniquely determined $w$ and the answer $j+k = w\cdot v + u$. If $|W|\ge 2$, then we return $\bot$. Now we show that if $j+k$ is not a bad index (\cref{defn:badindex}), then $|W|=1$ must hold. To see this, simply observe from the definition of $W$ that any two $w,w' \in W$ can be related by $\big ( (\sigma wv+\tau)\bmod p\big ) + o \equiv_m  \big ( (\sigma w'v+\tau)\bmod p\big ) + o' $ for some $o,o'\in \{-p,0,p\}$, so if $|W|\ge 2$ then we can pick $w'\neq w = \lfloor (j+k)/v\rfloor $ and get 
$\big ( (\sigma wv+\tau)\bmod p\big ) + (o-o') \equiv_m  \big ( (\sigma w'v+\tau)\bmod p\big ) $ (where $o-o'\in \{-2p,-p,0,p,2p\}$), meaning that $(j+k)$ is a bad index.

We have almost proved the lemma, and the only remaining job is to give an $O(1)$-time algorithm that returns $W$ (or reports $|W|\ge 2$). To do this we need to use precomputation: We initialize a length-$m$ array $T[0\dd m-1]$. Then we iterate over all $w'\in \N\cap [0,2N/v]$, and insert $w'$ to the cell $T\Big [\big ( (\sigma w'v+\tau)\bmod p\big )\bmod m\Big ]$. If more than one element is inserted to a cell then we simply mark that cell as overfull. This precomputation can be done in $O(m + N/v) \le O(t/\log^{12} t + \frac{t\log^5 t}{\log^{22} t})$ time.
Using this table $T$, we can easily compute $W$ in \cref{eqn:defnW} (or report $|W|\ge 2$) during query time, by looking up the cells $T\Big [\big (f-\big ((\sigma u + \tau)\bmod p\big )-o\big )\bmod m \Big ]$ for all $o\in \{-p,0,p\}$. This finishes the proof of \cref{lem:queryindex}.
\end{proof}

Now we can formally state and prove the properties of {\RecoverBucketFast} (\cref{alg:recoverbucketthreewise}). Roughly speaking, {\RecoverBucketFast} recovers all the terms corresponding to non-bad indices that are hashed into light buckets.
\begin{lemma}
\label{lem:recoverbucketthreewise}
\textnormal{\RecoverBucketFast} (\cref{alg:recoverbucketthreewise}) runs in time  $O(s \log \log t)$ (after $o(t)$-time preprocessing).

For bucket $b\in [m]$,
if the polynomial \cref{eqn:bucketpoly} is $s$-sparse, then with at least $1-2^{-(\log t)^{0.1}}$ probability, \textnormal{\RecoverBucketFast} returns a polynomial \[f(x) = \sum_{(j,k)\in I} A[j]B[k]\cdot x^{j+k}\] where the summation ranges over some set $I$ satisfying $I'_b\subseteq I\subseteq I_b$, where  $I_b := \{(j,k)\in [N]^2: h(j)+h(k)\equiv_m b\}$ and $I_b' := \{(j,k)\in [N]^2: h(j)+h(k)\equiv_m b \text{ and } (j+k) \text{ is not bad} \}$ (see \cref{defn:badindex}).
\end{lemma}
\begin{proof}
Recall the sparsity parameter (defined at \linereff{line:oneshotbucketsize} of \cref{alg:recoverthreewise}) is $ s = \lceil 8t/m\rceil = \Theta(\log^{12} t) $ (where $m = \lceil t/\log^{12} t\rceil$), and prime $q\in [0.5 \log^{25} t, \log^{25} t]$.

We first analyze the time complexity of  \cref{alg:recoverbucketthreewise}. The time  of running Prony's method with bit-packing by \cref{lem:pronybitpack} is $O(s\log s + (\frac{s}{\log t} + 1) (\log t)^{0.1} \cdot (\log q)^{O(\sqrt{\log s})} ) \le O(s\cdot \log \log t)$ (the second term is dominated). (Also verify that the $q\ge s^2$ condition required by \cref{lem:pronybitpack} is indeed satisfied by our choices of $s,q$.) The rest of the steps take $O(s')\cdot O(1) \le O(s)$ time by \cref{lem:queryindex}. Regarding the preprocessing time, \cref{lem:pronybitpack} takes $t^{0.2+o(1)}$ preprocessing time, and \cref{lem:queryindex} takes $o(t)$ preprocessing time.

Now we show the correctness in the case where
polynomial \cref{eqn:bucketpoly} is $s$-sparse. In this case, Prony's method (\cref{lem:pronybitpack}) correctly returns the polynomial $g(y) $ equal to \cref{eqn:bucketpoly} with at least $1-2^{-(\log t)^{0.1}}$ success probability. 
Note that $((j\bmod v)+(k\bmod v))\bmod v = (j+k)\bmod v$.
Then, by \cref{lem:queryindex}, in the final answer $f(x)$ we correctly include all terms $A[j]B[k] x^{j+k}$ where $(j+k)$ is not bad and $h(j)+h(k)\equiv_m b$ (and possibly more terms).
\end{proof}

Now we can prove the properties of 
{\OneShotRecovery}  (\cref{alg:recoverthreewise}).
\begin{proof}[Proof of \cref{lem:recoverthreewise-step}]
    We first analyze the time complexity of 
    {\OneShotRecovery}  (\cref{alg:recoverthreewise}).
  First we preprocess all powers $\{\omega^j\}_{j=0}^{q-1}$ in $\poly(q) = \polylog(t)$ time.
Then, for each $0\le i\le 2s$, the time complexity for computing $Z^{(i)}:= h(A^{(i)})\star_m h(B^{(i)})$ via packed FFT (\cref{thm:packfft}) is $O(m\log q ) = O(m\log \log t)$ time. Summing over all $i$, the total time is $O(sm\log \log t) = O(t\log \log t)$.
Then, the algorithm runs {\RecoverBucketFast} $m$ times, each of which takes $O(s\log \log t)$ time by \cref{lem:recoverbucketthreewise}.
Summing over all buckets $b\in [m]$, the total time for \cref{lem:recoverbucketthreewise} is $O(m\cdot s\log \log t) = O(t\log \log t)$. 

Now we analyze the correctness of
{\OneShotRecovery}  (\cref{alg:recoverthreewise}).
At the beginning of \cref{alg:recoverthreewise} we sample 
the hash function $h(x) = \pi(x) \bmod m = ((\sigma x+\tau)\bmod p)\bmod m$. 
We first need to analyze the probability that the sampled $h$ achieves a desired load-balancing property.
Recall from previous discussion that our algorithm hashes an index $j+k$ (where $j\in \supp(A),k\in \supp(B)$) to the bucket $(h(j)+h(k))\bmod m$. 
Writing $\pi_2(x)\coloneqq (\sigma x+2\tau)\bmod p$,
from  \cref{eqn:hij} we know that $h(j)+h(k)\equiv_m \pi_2(j+k) + o$ for some $o\in \{0,p\}$. 
Note that the variant hash function $h_2(x)=\pi_2(x)\bmod m$ has the same distribution as $h$ when $\sigma$ and $\tau$ are drawn uniformly randomly, so \cref{lem:overfullbuckets} applies to $h_2$ as well.
By applying \cref{lem:overfullbuckets} to $h_2$ and the set $X= \supp(A)+\supp(B)\subseteq [2N]$ where $2N \le 2t\log^5 t$, we get that with at least $1-O(1/\log^2 t)$ probability, it holds for $\tilde X=\pi_2(X) + \{0,p\}$ that  
\begin{align*}
\Big \lvert \Big \{ y \in \tilde X :  \sum_{y' \in \tilde X} \mathbf{1}  \left[ y\equiv y' \pmod{m} \right] > \frac{4t}{m} \Big \} \Big \rvert & \le \frac{m\cdot 2N\log^3 (2N)}{t}\\
& \le O(\frac{(t/\log^{12} t)\cdot (t\log^5 t) \log^3 t }{t})\\
& \le O(t/\log^4 t).
\end{align*}
If an index $j+k \in \supp(A)+\supp(B)$ is hashed into a heavy bucket (recalling the sparsity parameter is $s = \lceil 8t/m\rceil$), then for some $y\in \{\pi_2(j+k),\pi_2(j+k)+p\}$ it holds that 
$\sum_{y' \in \tilde X} \mathbf{1}  \left[ y\equiv y' \pmod{m} \right] > \frac{4t}{m}$. 
Hence, the inequality above implies that the number of indices in $\supp(A)+\supp(B)$ hashed into heavy buckets is at most $O(t/\log^4 t)$.
On the other hand, by \cref{lem:probbadindex}, and linearity of expectation with  Markov's inequality, we know with at least $1-O(1/\log^2 t)$ probability over the sampled hash function $h$,  the number of bad indices among $\supp(A) + \supp(B)$ is at most $O(|\supp(A) + \supp(B)|/\log^3 t)\le O(t/\log^3 t)$. 
Now by a union bound, we assume both the aforementioned events successfully hold for the sampled hash function.

Then, by \cref{lem:recoverbucketthreewise}, we know that we recover all but $O(t/\log^3 t)$ terms with bad indices and $O(t/\log^4 t)$ terms that are hashed to heavy buckets, and the failed terms due to the $2^{-(\log t)^{0.1}}$ failure probability per bucket from \cref{lem:recoverbucketthreewise} itself.
The expected number of failed terms due to \cref{lem:recoverbucketthreewise} itself is at most $t \cdot 2^{-(\log t)^{0.1}} \le t/\log^{10} t $, and by Markov's inequality it is at most $t/\log^{7} t$ with probability $1-1/\log^3 t$. We assume this holds by another union bound.
Hence, the final result $R$ satisfies $\|R - A\star B\|_0 \le O(t/\log^3 t)$.
\end{proof}

\section{Finding Moduli Deterministically}
\label{sec:determi}
The main goal of this section is to prove the following theorem, which deterministically finds a list $V$ of $N/Q$ integers from $\Theta(Q)$, for $Q = \frac{N}{\polylog N}$, such that the Farey fractions with denominator from $V$ are well-spaced. 

\begin{theorem}
\label{thm:deterministic_preprocess}
For any $N$, $Q = \frac{N}{\polylog N}$, and any constant $t \ge 1$, in $\tO(N^{1+1/t})$ deterministic time, we can find a list $V$ of $\frac{N}{Q}$ integers in $[Q/2^t, Q]$, such that for any length $\frac{1}{N}$ interval from $\R / \Z$, the number of elements in $\{\frac{h}{v}: v \in V, 1 \le h < v\}$ belonging to that interval is $O(\log^t N)$. 
\end{theorem}

As a sanity check, let us verify the existence of such a small list $V$ for the case $t = 1$. In this case, let $\mathcal{P}$ be the set of primes from $[Q/2, Q]$, and we can simply pick $N / Q$ random primes from $\mathcal{P}$. For each length $1/N$ interval, the number of elements from $\{\frac{h}{p}: p \in \mathcal{P}, 1 \le h < p\}$ contained in it is at most $O(Q^2 / N)$, as every two elements in the set differ by at least $1 / (Q/2)^2$. Therefore, as we sample $N / Q$ random primes from $\mathcal{P}$, the expected number of elements from $\{\frac{h}{v}: v \in V, 1 \le h < v\}$  contained in the interval is $O(\frac{Q^2}{N} \cdot \frac{N / Q}{Q / \log Q}) = O(\log Q) = O(\log N)$, and this can be shown to hold w.h.p. Taking a union bound over all intervals (in fact, we only need to consider intervals $[0, 1/N), [1/ N, 2/N), \ldots, [(N-1)/N, 1)$, as any length $1/N$ interval in $\R / \Z$ can be covered by two of them) shows that each length $1/N$ interval contains $O(\log N)$ Farey fractions w.h.p., implying the existence of such a set. \cref{thm:deterministic_preprocess} states that we can find $V$ deterministically, and has better running time for larger $t$. 

Combining  \cref{thm:deterministic_preprocess}  and a version of the Large Sieve Inequality, \cref{lem:largesieve_v2} leads to the following result, which states that mod-prime hash (technically, now we can mod a composite number as well) performs essentially the same if the modulus is chosen randomly from the small set $V$, instead of from all primes $[Q / 2, Q]$.

\begin{corollary}
\label{cor:largesieve_det}
For any $N$, $Q = \frac{N}{\polylog N}$, and any constant $t \ge 1$, in time $\tO(N^{1+1/t})$, we can find a list $V$ of $\frac{N}{Q}$ integers in $[Q/2^t, Q]$, so that for any $A \subseteq [N]$,
\[\sum_{a,b\in A} \Pr_{v \in V}[v \mid a-b] \le   |A| \cdot O\left(\frac{|A|}{Q} + \log^t N\right).
\]
\end{corollary}
\begin{proof}
Let $V$ be the set of integers given by \cref{thm:deterministic_preprocess}. 
    Apply \cref{lem:largesieve_v2} with $\mathcal{X} = \{\frac{h}{v}: v \in V, 1 \le h < v\}$, $\delta = \frac{1}{N}$, $a_n = \begin{cases}1 & n \in A \\ 0 & n\notin A\end{cases}$ and $K = O(\log^t N)$. Let $S(\alpha) = \sum_{n=0}^{N-1} a_n  e(\alpha n)=\sum_{a\in A} e(\alpha a)$. \cref{lem:largesieve_v2} implies
    \[\sum_{v\in V}\sum_{h=1}^{v-1} \left \lvert S\left (h/v\right )\right \rvert ^2 \le O(N |A| \log^t N).\]
        
    Then we have
\begin{align*}
    \sum_{a,b\in A} \Pr_{v\in V}[v \mid a-b] &= \sum_{a,b\in A}\Ex_{v\in V} \left [\frac{1}{v}\sum_{h=0}^{v-1}e\left (\frac{h(a-b)}{v}\right )\right ]\\
    & =   \Ex_{v\in V}\left [  \frac{1}{v}\sum_{h=0}^{v-1} \left \lvert S(h/v) \right  \rvert^2\right ]\\
    & \le \frac{|A|^2}{Q/2^t} + \frac{1}{|V| Q/2^t}\sum_{v\in V}\sum_{h=1}^{v-1} \left \lvert S\left (h/v\right )\right \rvert ^2\\
    &= |A| \cdot O\left(\frac{|A|}{Q} + \log^t N\right). 
\end{align*}
\end{proof}

Next, we show a few lemmas towards proving \cref{thm:deterministic_preprocess}.

\begin{lemma}
\label{lem:primes_load}
Let $N$ and $R$ be integers.
For any fixed $x \in [N]$, the number of tuples $(p^{(0)}, p^{(1)}, \ldots, p^{(t-1)}, h)$, where $p^{(0)}, p^{(1)}, \ldots, p^{(t-1)}$ are distinct primes from $[R /2, R]$ and $1 \le h < \prod_i p^{(i)}$ such that $\frac{h}{\prod_i p^{(i)}} \in [\frac{x}{N}, \frac{x+1}{N})$ is $O(R^{t-1}+\frac{R^{2t}}{N})$. 
\end{lemma}
\begin{proof}

For any fixed integer $g \in [t]$, we will show that the number of tuples $(p^{(0)}, p^{(1)}, \ldots, p^{(t-1)}, h)$ satisfying the lemma statements, plus the additional requirement that the greatest common divisor between $h$ and $\prod_i p^{(i)}$ has $g$ prime factors, is $O(R^{t-1}+\frac{R^{2t}}{N})$. As there are $t = O(1)$ choices for $g$, the total number of tuples satisfying the requirement in the statement is $O(R^{t-1}+\frac{R^{2t}}{N})$. 

Without loss of generality, we assume that the greatest common divisor between $h$ and $\prod_i p^{(i)}$ is $\prod_{i \in [g]} p^{(i)}$. We also assume that $p^{(g)} < p^{(g+1)} < \cdots < p^{(t-1)}$. These constraints only incur a constant factor loss in the number of tuples as $t$ is a constant. We can then rewrite $h = h_0 \cdot \prod_{i \in [g]} p^{(i)}$ and $\gcd\left(h_0, \prod_{g \le i < t} p^{(i)}\right)=1$. 

First, consider the following multi-set of numbers, 
\begin{align*}
    S := \left\{\frac{h_0}{\prod_{g \le i < t} p^{(i)}} : \right. & p^{(g)} < p^{(g+1)} < \cdots < p^{(t-1)} \text{ are primes in } [R/2, R], \\ 
    & \left. 1 \le h_0 < {\prod_{g \le i < t} p^{(i)}}, \gcd\left(h_0, \prod_{g \le i < t} p^{(i)}\right)=1\right\}.
\end{align*}
All numbers in the set are distinct, because all numbers either have different numerator or denominator and all numbers are represented by irreducible fractions. For any pair of numbers $\frac{h_0}{\prod_{g \le i < t} p^{(i)}}$ and $\frac{h_0'}{\prod_{g \le i < t} p^{(i)'}}$, 
\begin{align*}
\left|\frac{h_0}{\prod_{g \le i < t} p^{(i)}} - \frac{h_0'}{\prod_{g \le i < t} p^{(i)'}}\right| &= \frac{|h_0 \cdot \prod_{g \le i < t} p^{(i)'} - h_0' \cdot \prod_{g \le i < t} p^{(i)}|}{\prod_{g \le i < t} p^{(i)} \cdot \prod_{g \le i < t} p^{(i)'}} \\ &\ge \frac{1}{\prod_{g \le i < t} p^{(i)} \cdot \prod_{g \le i < t} p^{(i)'}} \\ & = \Omega\left(\frac{1}{R^{2t-2g}}\right). 
\end{align*}
Therefore, as $[\frac{x}{N}, \frac{x+1}{N})$ has size $\frac{1}{N}$, there can be at most $1+\frac{1/N}{\Omega(1/R^{2t-2g})} = 1 + O(R^{2t-2g}/N)$ numbers from $S$ in the range.  

Now consider the following multi-set of numbers, whose size is exactly the number of tuples we need to bound:
\begin{align*}
    S^* := & \left\{\frac{h}{\prod_{i \in [t]} p^{(i)}} : \right.  p^{(0)}, \ldots, p^{(g-1)}, p^{(g)} < p^{(g+1)} < \cdots < p^{(t-1)} \text{ are } \\ 
    & \left. \text{distinct primes in } [R/2, R], 1 \le h < {\prod_{i \in [t]} p^{(i)}}, \gcd\left(h, \prod_{i \in [t]} p^{(i)}\right)=\prod_{i \in [g]} p^{(i)}\right\}.
\end{align*}
Note that each number in $S^*$ appears in $S$. Also, for each number $s = \frac{h_0}{\prod_{g \le i < t} p^{(i)}} \in S$, its number of  occurrences in $S^*$ is the number of $p^{(0)}, \ldots, p^{(g-1)}$ that are distinct primes in $[R/2, R]$ different from $p^{(g)}, \ldots, p^{(t-1)}$. Thus, the number of occurrences of $s$ in $S^*$ is $O(R^g)$. As a result, the number of elements from $S^*$ that belong to $[\frac{x}{N}, \frac{x+1}{N})$ is $O(R^g) \cdot (1 + O(R^{2t-2g}/N)) = O(R^g + R^{2t-g} / N)= O(R^{t-1} + R^{2t} / N)$. 
\end{proof}

\begin{claim}
\label{cl:CRT}
    For distinct primes $p^{(0)}, p^{(1)}, \ldots, p^{(t-1)}$, and for $0 \le h < \prod_{i} p^{(i)}$, there exist unique $a^{(0)}, a^{(1)}, \ldots, a^{(t-1)}$ with $0 \le a^{(i)} < p^{(i)}$ for every   $i \in [t]$ such that $\sum_i \frac{a^{(i)}}{p^{(i)}} \equiv \frac{h}{\prod_{i} p^{(i)}} \pmod{1}$. 
\end{claim}
\begin{proof}
Let $P = \prod_{i} p^{(i)}$. By multiplying $P$ to both sides, the equivalence can be written as 
\[\sum_i a^{(i)} \cdot \frac{P}{p^{(i)}} \equiv h \pmod{P}.\] 
By the Chinese Remainder Theorem, it is equivalent to 
\[\sum_i a^{(i)} \cdot \frac{P}{p^{(i)}} \equiv h \pmod{p^{(j)}} \quad \text{for } j \in [t].\]
By noticing that $\frac{P}{p^{(i)}}$ is a multiple of $p^{(j)}$ unless $i = j$, the above can be equivalently written as 
\[a^{(j)} \cdot \frac{P}{p^{(j)}} \equiv h \pmod{p^{(j)}} \quad \text{for } j \in [t]. \]
Thus, we can uniquely pick $a^{(j)}$ as $h \cdot (\frac{P}{p^{(j)}})^{-1} \bmod {p^{(j)}}$. 
\end{proof}

For each prime $p$, we define a length-$N$ sparse array $A_p$, indexed by $x \in [N]$, where \[A_p[x] := \left|\left\{  a  \in [p]: \frac{a}{p} \in \left[\frac{x}{N}, \frac{x+1}{N}\right)\right\} \right|.\]

For distinct $p^{(0)}, p^{(1)}, \ldots, p^{(t-1)}$ and for $x\in [N]$, let $\ell^*(p^{(0)}, p^{(1)}, \ldots, p^{(t-1)}, x)$ be the number of $1 \le h < \prod_i p^{(i)}$ where $\frac{h}{\prod_i p^{(i)}} \in [\frac{x}{N}, \frac{x+1}{N})$. Also, we define    $\ell(p^{(0)}, p^{(1)}, \ldots, p^{(t-1)}, x)$ as the number of $(a^{(0)}, \ldots, a^{(t-1)}) \in \left([p^{(0)}] \times \cdots \times [p^{(t-1)}]\right) \setminus \{(0, \ldots, 0)\}$ such that \[\sum_{i \in [t]} \left\lfloor \frac{a^{(i)}N}{p^{(i)}}\right\rfloor \equiv x \pmod{N}.\] 
In the following, we show $\ell$ is a good approximation of $\ell^*$
\begin{claim}
\label{cl:approx_l}
    For any $N$, distinct primes $p^{(0)}, p^{(1)}, \ldots, p^{(t-1)}$ and $x \in [N]$, we have
    \begin{equation}
    \label{eq:approximate_l_1}
    \ell^*(p^{(0)}, p^{(1)}, \ldots, p^{(t-1)}, x) \le \sum_{-t \le \Delta \le t} \ell(p^{(0)}, p^{(1)}, \ldots, p^{(t-1)}, (x + \Delta) \bmod{N})
    \end{equation}
    and 
    \begin{equation}
    \label{eq:approximate_l_2}
    \ell(p^{(0)}, p^{(1)}, \ldots, p^{(t-1)}, x) \le \sum_{-t \le \Delta \le t} \ell^*(p^{(0)}, p^{(1)}, \ldots, p^{(t-1)}, (x + \Delta) \bmod{N}). 
    \end{equation}
\end{claim}
\begin{proof}
Let $S$ denote the set \[\left([p^{(0)}] \times \cdots \times [p^{(t-1)}]\right) \setminus \{(0, \ldots, 0)\}.\]
By definition, $\ell(p^{(0)}, p^{(1)}, \ldots, p^{(t-1)}, x)$ is the number of $(a^{(0)}, \ldots, a^{(t-1)}) \in S$ such that \[\sum_{i \in [t]} \left\lfloor \frac{a^{(i)}N}{p^{(i)}}\right\rfloor \equiv x \pmod{N}.\] By \cref{cl:CRT}, $\ell^*(p^{(0)}, p^{(1)}, \ldots, p^{(t-1)}, x)$ is the number of $(a^{(0)}, \ldots, a^{(t-1)}) \in S$ such that \[\left\lfloor \sum_{i \in [t]} \frac{a^{(i)}N}{p^{(i)}} \right\rfloor \equiv x \pmod{N}.\]
The two inequalities follow because $\sum_{i \in [t]} \left\lfloor \frac{a^{(i)}N}{p^{(i)}}\right\rfloor$ and $\left\lfloor \sum_{i \in [t]} \frac{a^{(i)}N}{p^{(i)}} \right\rfloor$ differ by at most $t$. 
\end{proof}
\cref{cl:approx_l} implies that, for fixed primes $p^{(0)}, p^{(1)}, \ldots, p^{(t-1)}$, the maximum value of $\ell^*(p^{(0)}, p^{(1)}, \ldots, p^{(t-1)}, x)$ and $\ell(p^{(0)}, p^{(1)}, \ldots, p^{(t-1)}, x)$ over all $x$ differ by a constant factor.

Finally, we are ready to prove \cref{thm:deterministic_preprocess}. 

\begin{proof}[Proof of \cref{thm:deterministic_preprocess}]
Let  $k = \frac{N}{Q}$, $R = Q^{1/t}$, and  $\caP$ be the set of primes in $[R/2, R]$. Each integer we will add to $V$ will be the product of $t$ distinct primes from $\mathcal{P}$. For the $i$-th such integer, we will use $p_i^{(0)} < \ldots  < p_i^{(t-1)}$ to denote its prime factors. Consider the following objective function
\[\sum_{x=0}^{N-1} \left(\sum_{i \in [k]}  \ell(p_i^{(0)}, \ldots, p_i^{(t-1)}, x) \right)^{\log N}.\]
The overall strategy is to pick $p_0^{(0)},  \ldots, p_0^{(t-1)}, p_1^{(0)}, \ldots, p_1^{(t-1)}, \ldots, p_{k-1}^{(0)}, \ldots, p_{k-1}^{(t-1)}$ in this order, and each time the algorithm picks the next prime, it aims to minimize the expected value of the objective function over the randomness of the primes not yet picked. 

First, we will show that if the algorithm perfectly follows this strategy, the resulting primes will satisfy the requirement of the theorem.

\paragraph{Initial expectation. } First, we consider the expected value of the objective function before we select any primes. As $R = \frac{N^{1/t}}{\polylog N}$, the bound from \cref{lem:primes_load} becomes $O(\frac{R^{2t}}{N})$. Thus, for any fixed $x$, 
\begin{align*}
    \Ex\left[\sum_{i \in [k]}  \ell^*(p_i^{(0)}, \ldots, p_i^{(t-1)}, x)\right] &= k \cdot \frac{\sum_{p^{(0)} <  \cdots < p^{(t-1)} \subseteq \mathcal{P}, h} \mathbf{1}\left[\frac{h}{\prod_i p^{(i)}} \in [\frac{x}{N}, \frac{x+1}{N})\right]}{\sum_{p^{(0)} < \cdots < p^{(t-1)} \subseteq \mathcal{P}} 1}\\
    &= O\left(k \cdot \frac{R^{2t}}{N} \cdot \left(\frac{\log R}{R} \right)^t\right) \tag{by \cref{lem:primes_load} and $|\mathcal{P}| = \Theta(R / \log R)$}\\
    &= O\left(\frac{k R^t \log^t N}{N}\right).
\end{align*}
By \cref{cl:approx_l}, we also have 
\[\Ex\left[\sum_{i \in [k]}  \ell(p_i^{(0)}, \ldots, p_i^{(t-1)}, x)\right] = O\left(\frac{k R^t \log^t N}{N}\right).\]

Next, we show the values of $\ell(p_i^{(0)}, \ldots, p_i^{(t-1)}, x)$ are $0$ or $1$. By \cref{cl:approx_l}, $\ell(p_i^{(0)}, \ldots, p_i^{(t-1)}, x)$ is upper bounded by the sum of $\ell^*(p_i^{(0)}, \ldots, p_i^{(t-1)}, (x+\Delta) \pmod{N})$ where $\Delta \in [-t, t]$. By definition of $\ell^*$, this sum is upper bounded by the number of $1 \le h < \prod_{i=0}^{t-1} p^{(i)}$ where $h / \prod_{i=0}^{t-1} p^{(i)}$  lies between $\frac{x-t}{N}$ and $\frac{x+t + 1}{N}$. The difference of $h / \prod_{i=0}^{t-1} p^{(i)}$ between two consecutive values of $h$ is $1 / \prod_{i=0}^{t-1} p^{(i)} \ge 1 / R^t =1 / Q$, while the size of the range from $\frac{x-t}{N}$ to $\frac{x+t + 1}{N}$ is $O(t/N)$. Since $Q = N / \polylog N$, and $t$ is a constant, $O(t/N)\ll 1 / Q$, so there can be at most one value of $h$ where $h / \prod_{i=0}^{t-1} p^{(i)}$ lies in the interval. This implies the values of $\ell(p_i^{(0)}, \ldots, p_i^{(t-1)}, x)$ are $0$ or $1$.

Thus, by Chernoff bound, for any $y \ge 2\Ex\left[\sum_{i \in [k]}  \ell(p_i^{(0)}, \ldots, p_i^{(t-1)}, x)\right]$, 
\[\Pr[\sum_{i \in [k]}  \ell(p_i^{(0)}, \ldots, p_i^{(t-1)}, x) \ge y] \le e^{-y/3}.\] When $y \ge C \log N $ for sufficiently large constant $C$, $e^{-y/3} \cdot y^{\log N}$ decays exponentially. Therefore, the expectation of $\left(\sum_{i \in [k]}  \ell(p_i^{(0)}, \ldots, p_i^{(t-1)}, x) \right)^{\log N}$ can be bounded as (in the following, let $X = \sum_{i \in [k]}  \ell(p_i^{(0)}, \ldots, p_i^{(t-1)}, x)$)
\begin{align*}
    & \sum_{y \ge 0} \Pr\left[X = y\right] \cdot y^{\log N}\\
    = & \left( \sum_{y < 2\Ex\left[X\right]} \Pr\left[X = y\right] \cdot y^{\log N} \right) + \left( \sum_{2\Ex\left[X\right] \le y < C \log N} \Pr\left[X = y\right] \cdot y^{\log N} \right) \\
    & + \left( \sum_{y \ge \max\{C \log N, 2\Ex\left[X\right]\}} \Pr\left[X = y\right] \cdot y^{\log N} \right)\\
    \le & (2 \Ex[X])^{\log N} + (C \log N)^{\log N} + \left( \sum_{y \ge \max\{C \log N, 2\Ex\left[X\right]\}} \Pr\left[X \ge y\right] \cdot y^{\log N} \right)\\
    \le &  O\left(\max\left\{\frac{k R^t \log^t N}{N}, \log N\right\}\right)^{\log N} + \left( \sum_{y \ge \max\{C \log N, 2\Ex\left[X\right]\}} e^{-y/3} \cdot y^{\log N} \right)\\
    = & O(\log^t N)^{\log N}.
\end{align*}
Summing over all $x$ gives the upper bound 
\[N \cdot O\left(\log^t N\right)^{\log N}.\]

\paragraph{Final objective. } If we perfectly compute the expectation of the objective function after each step and select the prime that minimizes the expectation, the final value of the objective function is upper bounded by the initial expectation of the objective function. Thus, we will obtain $\{p^{(j)}_{i}\}_{i \in [k], j \in [t]}$ such that 
\[\sum_{x=0}^{N-1} \left(\sum_{i \in [k]}  \ell(p_i^{(0)}, \ldots, p_i^{(t - 1)}, x) \right)^{\log N} = N \cdot  O\left(\log^t N\right)^{\log N}. \]
This implies that, for any fixed $x$, $\sum_{i \in [k]}  \ell(p_i^{(0)}, \ldots, p_i^{(t-1)}, x)$ can be upper bounded by 
\[\left(N \cdot  O\left(\log^t N\right)^{\log N}\right)^{\frac{1}{\log N}} = O\left(\log^t N\right).\]
By \cref{cl:approx_l}, $\sum_{i \in [k]}  \ell^*(p_i^{(0)}, \ldots, p_i^{(t-1)}, x)$ can also be  upper bounded by 
$O\left(\log^t N\right)$ for any $x$. Thus, if we set the numbers in $V$ as $\left\{\prod_{j \in [t]} p_i^{(j)}\right\}_{i \in [k]}$, we get that, for any $x$, the number of elements in $\left\{\frac{h}{v}: v \in V, 1 \le h < v\right\}$ belonging to $[\frac{x}{N}, \frac{x+1}{N})$ is $O(\log^t N)$. As any interval of length $\frac{1}{N}$ belongs to the union of two intervals of the form $[\frac{x}{N}, \frac{x+1}{N})$, we get the desired bound.

\paragraph{Preprocessing. } 
Next, we show how to perform the conditional expectation efficiently. It is not difficult to verify that 
\[\ell(p^{(0)}, p^{(1)}, \ldots, p^{(t-1)}, x) = \left\{
  \begin{array}{lr}
    \left(A_{p^{(0)}} \star_N A_{p^{(1)}} \star_N \cdots \star_N A_{p^{(t-1)}}\right)[x] & : x \ne 0,\\
    \left(A_{p^{(0)}} \star_N A_{p^{(1)}} \star_N \cdots \star_N A_{p^{(t-1)}}\right)[x] - 1 & : x  = 0.
  \end{array}
\right.
\]
Hence, for any $q \in \caP$, and any $j \in [t]$, we prepare the following length $N$ array $S_{j, q}$, defined as 
\[\sum_{\substack{p^{(j)}, \ldots ,p^{(t-1)} \in \caP: \\ q \le p^{(j)} < \cdots < p^{(t-1)}}} A_{p^{(j)}} \star_N \cdots \star_N A_{p^{(t-1)}}.\]
Suppose $q'$ is the successor of $q$ in $\caP$ (if $q$ is the largest prime in $\caP$, then $S_{j, q}$ is easy to compute), then it is not difficult to see
\[
    S_{j, q} = \left\{
    \begin{array}{lr}
        S_{j, q'} + A_{q} & j = t - 1,\\
        S_{j, q'} + A_{q} \star_N S_{j + 1, q'} & j < t - 1.
    \end{array}\right.
\]
Thus, we can compute each $S_{j, q}$ in $\tO(N)$ time via FFT. The overall preprocessing time is thus $\tO(t N |\caP|) = \tO(N^{1+1/t})$. 

\paragraph{Computing the expectation. } Finally, we show how to compute the expectation in $\tO(N)$ time given  some fixed choices of some of the primes. If we can achieve this, then for the next prime we select, we  try all $O(|\caP|)$ possibilities for it, and evaluate $O(|\caP|)$ values of the expectation, and select the prime that minimizes the expectation. Thus, the running time would be $\tO(kt N |\caP|) = \tO(N^{1+1/t})$. 

Now, suppose for $i \in [k]$, we already fixed the choices of the first $a_i$ primes $p_i^{(0)}, p_i^{(1)}, \ldots, p_i^{(a_i-1)}$. Then for $x \in [N]$, we let $r_i(x)$ denote the probability that $\ell(p_i^{(0)}, \ldots, p_i^{(t-1)}, x) = 1$ (note that as $q \ll N$, $\ell(p_i^{(0)}, \ldots, p_i^{(t-1)}, x)$ is either $0$ or $1$), $q$ be the successor of $p_i^{(a_i - 1)}$ in $\caP$ (if $a_i = 0$, then $q$ is the smallest prime in $\caP$), and let $m$ be the number of primes in $\caP$ that are greater than or equal to $q$. Then 
\begin{align*}
    r_i(x) &= \frac{1}{\binom{m}{t - a_i}} \sum_{q \le p_i^{(a_i)} < \cdots < p_i^{(t - 1)}} \ell(p_i^{(0)}, \ldots, p_i^{(t-1)}, x)\\
    &= \frac{1}{\binom{m}{t - a_i}} \sum_{q \le p_i^{(a_i)} < \cdots < p_i^{(t - 1)}} \left(A_{p_i^{(0)}} \star_N \cdots \star_N A_{p_i^{(t-1)}} \right)[x] - \mathbf{1}[x = 0]\\
    &= \frac{1}{\binom{m}{t - a_i}} \left( A_{p_i^{(0)}} \star_N \cdots \star_N A_{p_i^{(a_i-1)}} \star_N S_{a_i, q}\right)[x] - \mathbf{1}[x = 0].
\end{align*}
Therefore, we can compute $r_i(x)$ for all $i \in [k]$ and $x \in [N]$ in $\tO(ktN) = \tO(N)$ time via FFT. 

For a fixed $x$, by the fact that $\ell(p_i^{(0)}, \ldots, p_i^{(t-1)}, x)$'s are independent for  $i \in [k]$, the probability that $\sum_{i\in [k]} \ell(p_i^{(0)}, \ldots, p_i^{(t-1)}, x) = B$ for every $B = 0, \ldots, k$ is 
\[\sum_{K \in \binom{[k]}{B}} \left( \prod_{i \in K} r_i(x) \cdot \prod_{i \not \in K} (1 - r_i(x)) \right). \]
These values can be computed using dynamic programming in $O(k^2) = \polylog(N)$ arithmetic operations. Also, observe that all numerators and denominators during the computation can be represented by $\polylog N$-bit integers, so the running time for the dynamic programming is $\polylog N$. With these values, we can compute the expected value of the objective function in $\tO(N)$ time. 
\end{proof}

\section{Text-to-Pattern Hamming Distances}
\label{sec:texttopattern}

\begin{lemma}
\label{lem:sparse_conv_minus_C}
    Let $V$ be the set from \cref{thm:deterministic_preprocess} with $N$ and $Q = \frac{N}{\log^{t+3} N}$, assumed to be given.
    Given three integer sequences $A, B, C$ of length $N, N, 2N-1$ respectively, with the promise that $(A \star B - C)[i] \ge 0$ for every $i$, we can compute $A \star B$ in $O(N)$ deterministic time, as long as $\max\{\|A\|_0, \|B\|_0, \|A \star B - C\|_0\} \le \frac{N}{\log^{t+3} N}$, and we are given a set $T$ with $|T| = O(\frac{N}{\log^{t+3} N})$, and $\supp(A \star B - C) \subseteq T$. 
\end{lemma}
\begin{proof}
Without loss of generality, we assume $T\subseteq [2N - 1]$.
Our algorithm will consist of multiple iterations. At each iteration, we will compute $(A \star B - C)[i]$ for at least half of the indices $i \in T$. Thus, after each iteration, we decrease the size of $T$ by at least a half. 

    In each iteration, we apply \cref{cor:largesieve_det} with $A := T$ and $N := N$, which implies that there exists $v \in V$, such that $\sum_{x, y \in T} [x \equiv_v y] = O(|T| \cdot (\frac{|T|}{Q} + \log^t N)) = O(|T| \log^t N)$. We can find such a $v$ in $O(|T| |V|) = O(|T| \cdot \log^{t+3} N)$ time.

    Let $\caP$ be the set of primes in $[0.5 k \log^{t + 1} N, k \log^{t + 1} N]$, for some sufficiently large $k$. Note that for every $x \ne y \in T$, $\Pr_{p \in \caP}[x \equiv_p y] = O(\frac{1}{\log^{t} N})$. Therefore, 
    \begin{align*}
        \sum_{x \ne y \in T}\Pr_{p \in \mathcal{P}} [x \equiv_{vp} y] & \le \sum_{\substack{x \ne y \in T \\ x \equiv_v y}}\Pr_{p \in \mathcal{P}} [x \equiv_p y]\\
        &\le O(|T| \log^t N) \cdot O(\frac{1}{\log^t N})\\
        &\le \frac{|T|}{2} \tag{picking large enough $k$}.
    \end{align*}
    This implies that  there exists $p \in \mathcal{P}$ such that $\sum_{x \ne y \in T}[x \equiv_{vp} y] \le \frac{|T|}{2}$.  We can find such a $p$ in $O(|T| \log^{t} N + \polylog(N))$ time. 
    The previous condition further implies that, half of the values in the multi-set $\{x \bmod{vp}\}_{x \in T}$ are unique. 

    Let $q := vp$. As $v = \Theta(Q) = \Theta(\frac{N}{\log^{t+3} N})$ and $p = \Theta(\log^{t+1} N)$, $q  = \Theta(\frac{N}{\log^{2} N})$.
    We can compute $(A \star B - C)[x]$ for every $x \in T$ with unique $x \bmod{q}$ in the following way: 
    We create $A'$ in $O(\|A\|_0)$ time, where $A'[i] = \sum_{j \equiv_q i} A[j]$. We similarly create $B'$. Similarly, we use the notation $C'$, but we do not actually spend $O(N)$ time to create it. 
    
    Then we compute $A' \star_{q} B'$ via FFT in $O(q \log q)$ time. If $x \in T$ has a unique value of $x \bmod{q}$, then 
    $(A \star B - C)[x] = (A' \star_q B' - C')[x \bmod{q}] = (A' \star_q B')[x \bmod{q}] - \sum_{y \equiv_q\, x} C[y]$. Thus, computing each such $(A \star B - C)[x]$ takes $O(\frac{N}{q}) = O(\log^2 N)$ time, and hence it takes $O(|T| \log^2 N)$ time overall. After this, we can remove these $x$ from $T$ (and we have to properly update $C[x]$ so that $(A \star B - C)[x] = 0$).

    The overall running time for the whole iteration is 
    \[O\left(|T| \log^{t+3} N + \|A\|_0 + \|B\|_0 + q \log q + |T| \log^2 N\right) = O\left(\frac{N}{\log N} + |T| \log^{t+3} N\right). \]
    Summing over all $O(\log |T|) = O(\log N)$ iterations gives the $O(N)$ running time (note that $|T|$ decreases by at least a half in each iteration).  
\end{proof}

\begin{corollary}
\label{cor:x-plus-y}
For a fixed constant $t \ge 1$, suppose we are given the set $V$ from \cref{thm:deterministic_preprocess} with $N = N'$ and $Q = \frac{N'}{\log^{t+3} N'}$ for every $N' \le N$ that is a power of $2$.  Then given two nonnegative vectors $A, B \in \N^N$ with $\|A\|_1, \|B\|_1 \le s$ for some $s \ge \sqrt{2N}$,  we can compute $A \star B$ in $O(N \log(s^2 / N) + N \log \log N)$ deterministic time.
\end{corollary}
\begin{proof}
    Without loss of generality, we assume $N$ is a power of $2$. 

    If $s^2 \ge N^{1.01}$, then we can directly apply FFT, which runs in time $O(N \log N)$, and it already achieves the claimed bound.

    Let $p$ be the smallest prime greater than $\frac{s^2 \log^{t+3} N}{N}$, which can be computed in $\tilde{O}(\frac{s^2 \log^{t+3} N}{N}) = o(N)$ time. 

    Then we compute $C = A \star B \bmod{p}$ using \cref{thm:packfft} in time $O(N \log p) = O(N \log(s^2 / N) + N \log \log N)$. Note that now $A \star B - C \ge 0$. We also have $\|A \star B - C\|_0 \le s^2 / p \le N / \log^{t+3} N$, as $\|A \star B - C\|_1 \le \|A \star B\|_1 \le s^2$ and every term in $A \star B - C$ is a multiple of $p$. Additionally, $\|A\|_0, \|B\|_0 \le s \le N^{1.01/2} \le N / \log^{t+3} N$. In order to compute $A \star B$ using \cref{lem:sparse_conv_minus_C}, we also need a set $T$ that is a superset of $\supp(A \star B - C)$. 
    
    To do so, we recursively solve the following problem: Let $N' = N / 2$, $A', B' \in \N^{N'}$ with $A'[i] = A[i] + A[i + N']$ and $B'[i] = B[i] + B[i + N']$ for $i \in [N']$, and we need to compute $A' \star B'$. Clearly, $\|A'\|_1, \|B'\|_1 \le s$ and $s \ge \sqrt{2N'}$, so this smaller instance satisfies the requirement of the corollary statement. We set $T = \{i \in [2N']: (A' \star B')[i] \ge p / 2\} + \{0, N', 2N', 3N'\}$. We need to check the following two conditions in order for the set $T$ to be applicable in \cref{lem:sparse_conv_minus_C}:
    \begin{itemize}
        \item The size of $T$ is small: As $\|A'\|_1, \|B'\|_1 \le s$, $|\{i \in [2N'-1]: (A' \star B')[i] \ge p / 2\}| \le s^2 / (p / 2) = O(N / \log^{t+3} N)$. The size of $T$ is at most four times the previous set, so $|T| = O(N / \log^{t+3} N)$ as required. 
        \item $\supp(A \star B - C) \subseteq T$: If $i \in \supp(A \star B - C)$, then $(A \star B)[i] \ge p$, and thus $(A \star_{N'} B)[i \bmod {N'}] \ge p$. Notice that $(A \star_{N'} B)[i \bmod {N'}] = (A' \star B')[i \bmod {N'}] + (A' \star B')[(i \bmod {N'}) + N']$, so either $(A' \star B')[i \bmod {N'}] \ge p / 2$ or $(A' \star B')[(i \bmod {N'}) + N'] \ge p / 2$. Therefore, $i \bmod {N'} \in \{i \in [2N'-1]: (A' \star B')[i] \ge p / 2\}$ or $(i \bmod {N'}) + N' \in \{i \in [2N'-1]: (A' \star B')[i] \ge p / 2\}$. Either way, $i \in T$ as $i < 4N'$. 
    \end{itemize}
    Therefore, we can compute $A \star B$ using \cref{lem:sparse_conv_minus_C} in $O(N)$ time.

    If the time complexity for the problem with universe size $N$ is $Q(N)$, then we have shown 
    \[Q(N) \le Q(N / 2) + O(N \log(s^2 / N) + N \log \log N),\]
    with the boundary condition 
     \[Q(N)  = O(N \log(s^2 / N) + N \log \log N) \text{ \quad if \quad } N \le s^{2/1.01}. \]
     Therefore, the running time can be upper bounded by
     \begin{align*}
     & O\left(\sum_{i \ge 0} \left(\frac{N}{2^i} \log\left(s^2 / \frac{N}{2^i}\right) + \frac{N}{2^i} \log \log \frac{N}{2^i}\right)\right)\\
     = & O\left(\sum_{i \ge 0} \left(\frac{N}{2^i} \log\left(s^2 / N\right) + \frac{N}{2^i} \log(2^i) + \frac{N}{2^i} \log \log N\right)\right)\\
     = & O(N \log(s^2 / N) + N \log \log N)
     \end{align*}
as desired. 
\end{proof}

\cref{cor:x-plus-y} can be rephrased in the following way (by taking $t > 50$), leading to \cref{thm:x-plus-y-advice} (where we renamed $A,B$ to $X,Y$ to keep consistency with \cite{focs23}). 
\XPlusY* 

We can use \cref{thm:x-plus-y-advice} to solve the Text-to-Pattern Hamming Distances problem. 

\HammingDistances*
\begin{proof}

By standard techniques (see e.g., \cite{GawrychowskiU18}), we can assume $n = \Theta(m)$ and it suffices to get an $O(m \sqrt{m \log \log m})$ bound. 

First, we apply the preprocessing phase of \cref{thm:x-plus-y-advice} with $N = m$ in $O(m^{1.01})$ time. 

For each character $c$, we use $n_c$ to denote its combined frequency in the text and the pattern. 
Then, $\sum_c n_c = O(m)$.
As in previous works (e.g., \cite{fispat, Abrahamson87}), for each character, we only need to compute the convolution between two  $0/1$ arrays of length $O(m)$ with at most $n_c$ ones. 
Our strategy here is the same as \cite[Section 4.2 (arXiv version)]{focs23}.

For characters with $n_c \le \sqrt{m \log \log m}$, we compute the convolutions using brute-force in $O(n_c^2)$ time. Since $\sum_{c}n_c=O(m)$, the total running time for them is $O(m \cdot \sqrt{m \log \log m})$. 

For other characters, we use \cref{thm:x-plus-y-advice} to compute the convolution. The total running time is 
\[\sum_{c\ :\  n_c >  \sqrt{m \log \log m}} O(m \log(n_c^2 / m) + m \log \log m). \]
We first bound the contribution of the $O(m \log(n_c^2 / m))$ term in the summation, as follows: For each integer $i\ge 0$, denoting $s =2^i \sqrt{m}$, there are at most $O(m/s)$ different characters $c$ with $n_c \in [s,2s]$, which  together contribute $O(m/s)\cdot O(m\log (s^2/m)) = O(m^{1.5}i/2^i)$ to the summation.
Over all integers $i\ge 0$, this sums to $O(m^{1.5})$. Hence,
\[\sum_{c\ :\ n_c > \sqrt{m \log \log m}} O(m \log(n_c^2 / m)) \le \sum_{c\ :\ n_c > \sqrt{m}} O(m \log(n_c^2 / m)) = O(m\sqrt{m}). \]
Next, we bound the contribution of the $O(m \log \log m) $ term. Note that
\begin{align*}
\sum_{c\ :\  n_c > \sqrt{m \log \log m}} O(m \log \log m) 
& = O\left(\left|\{c: n_c > \sqrt{m \log \log m}\}\right| \cdot m \log \log m\right). \end{align*}
Recall $\sum_c n_c = O(m)$, so the above can be bounded by $O(m \sqrt{m \log \log m})$. 

Thus, the overall running time is $O(m \cdot \sqrt{m \log \log m})$. 
\end{proof}

In the Text-to-Pattern Dominance  Matching problem, we are given a pattern $P$ of length $m$ and a text $T$ of length $n$, both over a polynomial-size alphabet, and we need to compute $|\{k \in [m]: P[k] \le T[i+k]\}|$ for every shift $i$. 

\begin{theorem}
    The Text-to-Pattern Dominance Matching problem can be solved by a deterministic algorithm in $O(n\sqrt{m\log \log m})$ time.
\end{theorem}
\begin{proof}
    Similar to the Text-to-Pattern Hamming Distances problem, we can assume $n = \Theta(m)$ and it suffices to get an $O(m\sqrt{m\log \log m})$ bound. 

    \cite{focs23} showed that it suffices to solve $O(m / s)$ instances of convolution between two length-$m$ $0/1$ sequences with sparsity $\le s$ for every $s \le n$ that is a power of $2$. 
    
    If $s \le \sqrt{m \log \log m}$, we use brute-force to solve each instance in $O(s^2)$ time, and the total time over all $O(m/s)$ instances is $O(ms)$. Summing over all $s \le \sqrt{m \log \log m}$ that are powers of $2$ gives $O(m\sqrt{m\log \log m})$ total time. Otherwise, $s > \sqrt{m \log \log m}$, and we use \cref{thm:x-plus-y-advice} to solve each instance in $O(m \log(s^2 / m) + m \log \log m)$ time. The total running time over all instances can be bounded similarly as in the proof of \cref{thm:TtPHD}. 
\end{proof}

\section{The Constellation Problem}
\label{sec:conste}
In this section we prove \cref{thm:constellation-las-vegas}, restated below.
\Constellation*

Our algorithm invokes the Monte Carlo algorithm by Cardoze and Schulman \cite{CardozeS98} (see also the exposition in \cite{ColeH02}). Their algorithm actually provides the following stronger one-sided error guarantee:

\begin{theorem}[\cite{CardozeS98}]
\label{thm:constellation-monte-carlo}
      Given two integer sets $A,B\subseteq [N]$ of size $|A|,|B|\le n$, there is a Monte Carlo algorithm that outputs a set $S \subseteq [N]$ in $O(n \log n)$ time such that 
      \begin{itemize}
          \item $S$ always contains all shifts $s$ such that $A+s\subseteq B$. 
          \item With $1-1/\poly(n)$ probability, $S$ does not contain any other elements. 
      \end{itemize}
\end{theorem}

\begin{proof}[Proof of \cref{thm:constellation-las-vegas}]
    We first run the Monte Carlo algorithm in \cref{thm:constellation-monte-carlo} and obtain $S$.
     By the one-sided error guarantee of \cref{thm:constellation-monte-carlo}, $S$ is correct if and only if $A + s \subseteq B$ holds for all $s\in S$, which is equivalent to $A + S \subseteq B$. We run our Las Vegas algorithm for Sparse Nonnegative Convolution (\cref{thm:lasvegas-main}) to compute $A+S$, but we halt the algorithm and return $\bot$ if it runs for more than $O(n \log n)$ time. If the algorithm of \cref{thm:lasvegas-main} successfully terminates, we verify whether $A + S \subseteq B$. If the verification passes, we return $S$; otherwise, we return $\bot$. 

    If $S$ is correct, then $A + S \subseteq B$ and $|A+S| \le |B|\le n$, so \cref{thm:lasvegas-main} terminates in $O(n \log n)$ time with high probability, so our algorithm will accept and return the correct $S$ with high probability. Otherwise, $A+S\not \subseteq B$, and our algorithm always returns $\bot$. 
\end{proof}

\bibliographystyle{alphaurl} 
\bibliography{main}

\appendix

\section{Removing the \texorpdfstring{$\polylog(N\Delta)$}{polylog(N Delta)} dependency}
\label{sec:numerical}

The only place with the $\polylog(N\Delta)$  dependency in our sparse convolution algorithms is to pick a prime $p_b = \poly(N \Delta)$ and work with the finite field $\F_{p_b}$ and an element $\omega_b \in \F_{p_b}$ with multiplicative order at least $2N$. 

In the following, we describe how to remove the $\polylog(N\Delta)$ dependency by modifying our algorithms to work over $\C$ instead of $\F_{p_b}$.
  We choose an element $\omega_b = e(1/N')$ for some $N'$, where $e(x)$ denotes $\exp(2\pi \sqrt{-1} x)$ (we assume this can be done in constant time in word RAM; see the discussion at the end of \cref{subsec:model}). 
  We pick $N'$ to be a prime in $[4N, 8N]$. This can be done in $\polylog(N) = \polylog(t)$ Las Vegas time with high probability (recall that at this stage of the algorithm, we have $N \le t^{2-\eps}$ for our Monte Carlo algorithm in \cref{lem:recoverbucketmontecarlo} and $N \le t \polylog t$ for our Las Vegas algorithm in \cref{lem:recoverlasvegas}). 

Most parts of our algorithms involving field $\F_{p_b}$ are numerically stable when ported to $\C$, including FFT  (see e.g., \cite{higham2002accuracy}). An issue is linear system solving involving the transposed Vandermonde matrix $A$, where the nodes are $\omega_b^{i_1}, \ldots, \omega_b^{i_s}$ for indices $i_1, \ldots, i_s$ in a bucket. If $i_1, \ldots, i_s$ are close to each other, the Vandermonde matrix can be ill-conditioned. 
By \cite[Equation (4.5)]{GiesbrechtLL09} (which is implied by an inequality from \cite{gautschi1974norm}), if an $s \times s$ Vandermonde matrix $A$ has nodes $c_1, \ldots, c_s$ which are all $m$-th roots of unity for some $m \ge s$, then the matrix $2$-norm $\|A^{-1}\|_2$  can be upper bounded by 
\[
\max_{1 \le k \le s} \frac{2^{s-1} \sqrt{s}}{\prod_{j \ne k} |c_j - c_k|}. 
\]
In our case, the nodes are $\omega_b^{i_1}, \ldots, \omega_b^{i_s}$, so the above upper bound becomes 
\[
\max_{1 \le k \le s} \frac{2^{s-1} \sqrt{s}}{\prod_{j \ne k} |\omega_b^{i_j} - \omega_b^{i_k}|}. 
\]
If the smallest difference between two indices $i_j$ and $i_k$ is $\delta$ (differences measured in $\F_{N'}$), then $|\omega_b^{i_j} - \omega_b^{i_k}|$ can be lower bounded by $|e(\delta / N') - 1|$ for every $j \ne k$. Hence,  $\|A^{-1}\|_2$ can be upper bounded by 
\[
\frac{2^{s-1} \sqrt{s}}{|e(\delta / N') - 1|^{s-1}} = \frac{2^{s-1} \sqrt{s}}{(2\sin(\delta \pi / N'))^{s-1}} \le \frac{2^{s-1} \sqrt{s}}{(\delta \pi / N')^{s-1}} \le O(N' / \delta)^s. 
\]
Note that $\|A^{-1}\|_{\max}$, which is the maximum absolute value of all entries of $A^{-1}$, is upper bounded by $\|A^{-1}\|_2$. We adopt the definition of condition number used in \cite{higham2002accuracy}: $\text{cond}(A) = \big \||A| |A^{-1}|\big \|_{\infty}$, where $| \cdot |$ means we take absolute value entrywise, and $\| \cdot \|_\infty$ denotes the matrix $\infty$-norm.  
Thus, \[\text{cond}(A) = \big \||A| |A^{-1}|\big \|_{\infty} \le \poly(s) \|A\|_{\max} \|A^{-1}\|_{\max} \le O(N' / \delta)^s.\]

We use different ways to bound $\delta$ for our Monte Carlo and Las Vegas algorithms. 

\begin{itemize}
    \item Monte Carlo (\cref{lem:recoverbucketmontecarlo}): Instead of $\omega_b = e(1/N')$, we can pick $\omega_b = e(r / N')$ for a random integer $r \in  [1, N'-1]$, as in \cite[Theorem 4.3]{GiesbrechtLL09}. As shown in \cite[Theorem 4.3]{GiesbrechtLL09}, with constant probability, $\delta \ge \Omega(N' / s^2)$. This would imply $\text{cond}(A) \le O(s^2)^s$. Thus, we can verify that all values of $r i_1 \bmod{N'}, \ldots, r i_s \bmod{N'}$ are well-spaced (i.e., the minimum difference is $\Omega(N' / s^2)$) before solving the linear system, and only continue if so. This decreases the probability that an index can be recovered by a constant, but it does not affect the asymptotic behavior of the algorithm. 
    \item Las Vegas (\cref{lem:recoverlasvegas}): In our Las Vegas algorithms, we have $N = t \polylog t$ at this stage of the algorithm, and all indices in each bucket are congruent modulo some $p \ge t / \polylog t$. Therefore, we already have $\delta \ge p$, and thus $\text{cond}(A) \le O(N' / p)^s \le \polylog(t)^s$.  
\end{itemize}
In both cases, we have $\text{cond}(A) \le 2^{\polyloglog t}$.

In \cite[Equation (9.23)]{higham2002accuracy}, it was shown that the solution $\hat{x}$ to $Ax = b$ for $s \times s$ matrix $A$ given by Gaussian elimination has (setting $u = 2^{-O(f)}$ in their Equation (9.23), which is defined on \cite[Page 38]{higham2002accuracy}) $\|\hat{x}-x\|_\infty / \|x\|_\infty \le (1+\text{cond}(A)) \cdot 2^{O(s - f)}$,  where $f$ is the precision of floating points. We need $\|\hat{x}-x\|_\infty / \|x\|_\infty = 1/\poly(N\Delta)$, as we need to recover integer vector $x$ with $\|x\|_\infty \le \poly(N\Delta)$. Using $\text{cond}(A) \le 2^{\polyloglog t}$, it suffices to set $f = O(\log(N\Delta) + s  + \polyloglog t) = O(\log(N \Delta) + \polyloglog t) = O(\log(N\Delta))$, so arithmetic operations for floating point with precision $f$ can still be implemented in the word RAM model with constant cost.  

One final detail to check is that, due to numerical precision errors before solving the linear system, the vector $b$ might be slightly different from its actual value. Let $b_0$ be the accurate vector and let $A x_0 = b_0$, and we will have $x - x_0 = A^{-1} (b - b_0)$. We want $\|x - x_0\|_\infty \ll 1$, which holds when $\|b - b_0\|_\infty \ll 1 / (s \|A^{-1}\|_{\max})$. Note that $\|b_0\|_\infty \le \poly(t \Delta)$, so this requires multiplicative error $\le \frac{1}{\poly(t \Delta) (s \|A^{-1}\|_{\max})}$. As the previous steps are numerically stable, this can be achieved as long as $f \ge \log(s \poly(N \Delta)  \|A^{-1}\|_{\max})$, so $f = \Theta(\log(N\Delta))$ is again sufficient. 

\section{Missing proofs}

\subsection{Proof of \texorpdfstring{\cref{lem:pronybitpack}}{Lemma~\ref{lem:pronybitpack}}}
\label{sec:pronybitpack}

We work over $\F_p$ where $\log p$ is much smaller than the machine word length $w = \Omega(\log t)$. We usually assume the input/output arrays of $\F_p$-elements are given in compact memory representation where each machine word holds $\Theta(\frac{\log t}{ \log p })$ $\F_p$-elements.

Recall $\tilde O (f)$ denotes $f\cdot (\log f)^{O(1)}$.
We need the following algebraic primitives achieved via bit packing:
\begin{lemma}
\label{lem:bitpackingprimitive}
Let $t$ be a parameter. Assume $d,p\le t^{o(1)}$ and $p$ is prime. 
On a machine with word length $w = \Omega(\log t)$ bits, all the following tasks can be done in deterministic $\tilde O(\lceil d\cdot \frac{\log p}{\log t} \rceil)$ time (assuming compact input/output representation), after $t^{0.2+o(1)}$ preprocessing:
\begin{enumerate}
    \item Adding or subtracting two degree-$d$ polynomials $f,g\in \F_p[x]$. 
    \label{item:addsubtract}
    \item Multiplying two degree-$d$ polynomials $f,g\in \F_p[x]$. 
    \label{item:mult}
    \item Polynomial division with remainder for polynomials $f,g\in \F_p[x]$ of degree at most $d$.
    \label{item:div}
    \item Expanding the product $(x-a_1)(x-a_2)\cdots (x-a_d)\in \F_p[x]$ given $a_1,\dots,a_d \in \F_p$.
    \label{item:dcfft}
    \item Evaluating a degree-$d$ polynomial $f\in \F_p[x]$ on $d$ given points $a_1,\dots,a_d \in \F_p$.
    \label{item:multipoint}
    \item 
    Given distinct points $a_i\in \F_p$ and their evaluations $b_i\in \F_p$ for $1\le i\le d$,
    finding the unique interpolating polynomial $f\in \F_p[x]$ of degree $\le d-1$ satisfying $f(a_i)=b_i$ for all $1\le i \le d$.
    \label{item:interpo}
\end{enumerate}
\end{lemma}
\begin{proof}[Proof sketch]
We choose $\kappa = \lceil \frac{0.1\log t}{\log p}\rceil $, and note that $\kappa$ elements from $\F_p$ fit into one machine word. 

For \cref{item:addsubtract}, we need to perform additions (or subtractions) on $\kappa$ consecutive $\F_p$-elements (packed into a single machine word) in a batch in $O(1)$ time. This can be achieved by precomputing a table that contains all possible  $ O((p^{\kappa})^2)\le  O(t^{0.2} p^2) \le t^{0.2+o(1)}$ answers in $t^{0.2+o(1)}$ time, and each time performing a table look-up in $O(1)$ time. In this way we can add/subtract two degree-$d$ polynomials over $\F_p$ in $O(\lceil d/\kappa \rceil) =  O(\lceil d\cdot \frac{\log p}{\log t} \rceil)$ time.

For polynomial multiplication (\cref{item:mult}),  we work over $\F_{p^{\kappa}}$, and we can preprocess the multiplication table and inverses in $\tilde O((p^{\kappa})^2) \le t^{0.2+o(1)}$ time, so that arithmetic operations in $\F_{p^{\kappa}}$ can be performed by table look-ups in constant time. Then, by a simple Kronecker substitution, 
we can reduce the multiplication of two degree-$d$ polynomials over $\F_p$ to the multiplication of two degree-$O(\lceil d/\kappa \rceil )$ polynomials over a finite field $\F_{p^{\kappa}}$, which can be done using the Sch\"{o}nhage--Strassen FFT algorithm \cite{SchonhageS71} (see  e.g.,\ \cite[Section 8.3]{von2013modern}) in $O(\lceil d/\kappa \rceil \log \lceil d/\kappa \rceil \log \log \lceil d/\kappa \rceil) =  \tilde O(\lceil d\cdot \frac{\log p}{\log t} \rceil) $ time as desired.%

For \cref{item:div}, note that polynomial division with remainder reduces via Newton iteration (see e.g.,\ \cite[Section 9.1]{von2013modern}) to logarithmically many polynomial multiplications with exponentially increasing lengths (for length smaller than $\kappa$ we again use precomputation and table look-ups), and hence also runs in time $\tilde O(\lceil d\cdot \frac{\log p}{\log t} \rceil)$ via bit packing.

For \cref{item:dcfft}, we use a simple divide-and-conquer and polynomial multiplication via \cref{item:mult}. The running time satisfies the recurrence 
$T(d) = 2 T(\lceil d/2\rceil ) +  \tilde O(\lceil d\cdot \frac{\log p}{\log t} \rceil)$. When $d\le \kappa =\lceil \frac{0.1\log t}{\log p}\rceil$, we can again preprocess all the possible answers in $\tilde O(p^{\kappa}) \le  t^{0.1+o(1)}$ time and every time look up the answer in $T(d)=O(1)$ time. Hence, the recurrence still solves to $T(d) = \tilde O(\lceil d\cdot \frac{\log p}{\log t} \rceil)$.

For \cref{item:multipoint}, note that the classic divide-and-conquer algorithm of \cite{Fiduccia72} also has time complexity
$T(d) = 2 T(\lceil d/2\rceil ) +  \tilde O(\lceil d\cdot \frac{\log p}{\log t} \rceil)$
via bit packing, and when $d\le \kappa$ we can look up answers in $T(d)=O(1)$ time after preprocessing in $\tilde O((p^{\kappa})^2) \le t^{0.2+o(1)}$ time.

For \cref{item:interpo}, we observe that the divide-and-conquer algorithm (see e.g., \cite[Section 10.2]{von2013modern}) also admits speedup via bit packing, similarly to the previous algorithms.
\end{proof}

As a standard result in computer algebra, the minimal polynomial of a sequence $a$ of recursion order $\le n$ can be computed in $\tilde O(n)$ time given its first $2n$ entries $a_0,\dots,a_{2n-1}$, by solving a Hankel system \cite{Morf80,bitmead1980asymptotically,KaltofenY88}. 

\begin{lemma}[Solving linear recurrences, bit packed version]
\label{lem:linearrecurrbitpack}
Let $t$ be a parameter. Assume $k,p\le t^{o(1)}$ and $p$ is prime. 
On a machine with word length $w = \Omega(\log t)$ bits, the following holds after $t^{0.2+o(1)}$ preprocessing:
 Given the first $2k$ entries $a_0,a_1,\dots,a_{2k-1} \in \F_p$ of a linearly recurrent sequence $a$ of recursion order $\le k$, there is a deterministic  algorithm that computes the minimal polynomial of $a$ in $O(k\log k + k\cdot \frac{\log p}{\log t} \cdot \polylog(k))$ time.
\end{lemma}
\begin{proof}[Proof sketch]
 As described in \cite[Section 12.3]{von2013modern}, the task of computing the minimal polynomial of a linearly recurrent sequence can be solved by an application of the Extended Euclidean algorithm. So it suffices to analyze the running time of the Extended Euclidean algorithm under bit packing.

See \cite[Section 11.1]{von2013modern}
for a description of the fast Extended Euclidean algorithm. The running time is 
\begin{equation}
\label{eqn:dcmultrecur}
T(k) = 2T(\lfloor k/2\rfloor ) + O(M(k)) + O(k) ,
\end{equation}
where $M(k)$ denotes the complexity of multiplying two length-$k$ polynomials over $\F_p$. 
Without bit packing, \cref{eqn:dcmultrecur} solves to  $T(k)=  O(M(k)\log k) = O(k\log^2 k)$.
Here we improve the $M(k)$ term using bit packing from \cref{lem:bitpackingprimitive}.
In this way, the total running time \cref{eqn:dcmultrecur} becomes $T(k) = 2T(\lfloor k/2\rfloor ) + O(\lceil k/\kappa \rceil \log \lceil k/\kappa \rceil \log \log \lceil k/\kappa \rceil) + O(k) \le O(\frac{k\log^2 k\log \log k}{\kappa}) + O(k\log k) \le O(\frac{k\log^2 k\log \log k}{\log t/\log p} + k\log k)$.
\end{proof}

Another standard result in computer algebra is that transposed Vandermonde systems can be solved in near-linear time \cite{KaltofenY88,li2000arithmetic,pan2001structured,BostanLS03}. Here we again need to use bit packing to slightly improve the running time.
\begin{lemma}[Solving transposed Vandermonde systems, bit packed version]
\label{lem:vandermondebitpack}
Let $t$ be a parameter. Assume $k,p\le t^{o(1)}$ and $p$ is prime. 
On a machine with word length $w = \Omega(\log t)$ bits, the following holds after $t^{0.2+o(1)}$ preprocessing:
Given distinct elements $a_1,a_2,\dots,a_k\in \F_p$ and a vector $(b_1,b_2,\dots,b_k)$, there is a deterministic algorithm in 
$O(k\log k + k\cdot \frac{\log p}{\log t} \cdot \polylog(k))$ time that solves the transposed Vandermonde system
\begin{equation}
\left[ \begin{array}{cccc}
1 & 1 & \cdots & 1 \\
a_1 & a_2 & \cdots & a_k \\
\vdots & \vdots & \ddots & \vdots \\
a_1^{k-1} & a_2^{k-1} & \cdots & a_k^{k-1} \\
\end{array} \right]
\left[ \begin{array}{c}
x_1 \\
x_2 \\
\vdots \\
x_{k}
\end{array} \right]
=
    \left[ \begin{array}{c}
b_1 \\
b_2 \\
\vdots \\
b_{k}
\end{array} \right]
\end{equation}
\end{lemma}
\begin{proof}[Proof sketch]
Here we adapt the algorithm described by Li  
\cite{li2000arithmetic}. In \cite{li2000arithmetic}, the first step is to expand the polynomial product $(x-a_1)(x-a_2)\cdots (x-a_k)$, which can be done using \cref{lem:bitpackingprimitive}.
The second step is to use this result to compute $s_i = \sum_{j=1}^k a_j^i$ for all $0\le i\le 2k-2$, which can be done by solving a Hankel system of order $O(k)$. The third step also involves solving a Hankel system of order $O(k)$.
It is known that solving Hankel systems can be reduced to the Extended Euclidean algorithm (see \cite[Section 2]{pan2001structured}), and thus also admits the speedup via bit packing described in the proof of \cref{lem:linearrecurrbitpack}. 
The last step is to evaluate a univariate degree-$k$ polynomial on $k$ points, which can be done using \cref{lem:bitpackingprimitive}.

In summary, the total running time is 
$O(\frac{k\log^2 k\log \log k}{\log t/\log p} + k\log k)$.
\end{proof}

We adapt the modular composition algorithm of Kedlaya and Umans \cite{KedlayaU11}.
   Given three univariate polynomials $f,g,h\in \F_p[x]$ of degree smaller than $d$, Kedlaya and Umans' algorithm computes $f(g(x))\bmod h(x)$ in $d^{1+o(1)}\cdot (\log p)^{1+o(1)}$ bit complexity (their algorithm works for more general rings as well, but here we focus on prime fields). 
   We will shave off  roughly a $\log t$ factor in word RAM with $\Omega(\log t)$-bit words.

\begin{lemma}[Modular composition with bit packing]
\label{lem:modularcomposition}
Let $t$ be a parameter. Assume $d,p\le t^{o(1)}$ and $p$ is prime and $p\ge d^2$. 
On a machine with word length $w = \Omega(\log t)$ bits, the following holds after $t^{0.2+o(1)}$ preprocessing:

    Given three polynomials $f,g,h\in \F_p[x]$ of degree smaller than $d$, one can compute $f(g(x))\bmod h(x)$ by a deterministic algorithm  in
 $O ( ( \frac{d}{\log t}+ 1 ) \cdot (\log p)^{O(\sqrt{\log d})}  )$ time.
\end{lemma}
We defer the proof of \cref{lem:modularcomposition} to the \cref{sec:proofmodularcomposition}.

\begin{lemma}[Equal-degree factorization with bit packing] \label{lem:factorize-bitpack}
Let $t$ be a parameter. Assume $k,p\le \polylog t$ and $p$ is prime, and $p\ge k^2$. 
On a machine with word length $w = \Omega(\log t)$ bits, the following holds after $t^{0.2+o(1)}$ preprocessing:
    
Given a degree-$k$ polynomial $f\in \F_p[z]$ with $k$ distinct roots in $\F_p$, we can factorize $f$ into linear factors by a Las Vegas algorithm in expected 
$ O ( ( \frac{k}{\log t}+ 1 ) \cdot (\log p)^{O(\sqrt{\log k})}  )$ time.
\end{lemma}
\begin{proof}[Proof sketch]
We use the equal-degree factorization algorithm in  \cite[Section 4.3]{hoeven2022univariate} %
with expected running time (simplified in our prime field case here) 
$ O\big ((\mathsf{C}_{\F_p}(k) + \mathsf{M}_{\F_p}(k)\log(k p) + k\mathsf{D}_{\F_p} ) \log (k)\big )$ \footnote{The algorithm involves recursion on $k$, but this does not affect the stated time bound even with bit packing, since for lower levels of the recursion we can use the same precomputation idea as in the proof of \cref{lem:bitpackingprimitive}.},
where $ \mathsf{C}_{\F_p}(k)$  is the cost for modular composition, and $\mathsf{M}_{\F_p}(k)$ is the cost for polynomial multiplication, and $\mathsf{D}_{\F_p}$ is the cost for computing inverses in $\F_p$. 
By preprocessing in $\tilde O(p)$ time, we assume $\mathsf{D}_{\F_p} = O(1)$ via table look-up.
By \cref{lem:bitpackingprimitive}, the multiplication cost is $\mathsf{M}_{\F_p}(k) \le \tilde O(\lceil k\cdot \frac{\log p}{\log t} \rceil)$. By \cref{lem:modularcomposition}, the modular composition cost is
$\mathsf{C}_{\F_p}(k) = O ( ( \frac{k}{\log t}+ 1 ) \cdot (\log p)^{O(\sqrt{\log k})}  )$.
Hence, the total time is 
$ O ( ( \frac{k}{\log t}+ 1 ) \cdot (\log p)^{O(\sqrt{\log k})}  )$.
\end{proof}

Finally we prove \cref{lem:pronybitpack}.
\pronybitpack*
\begin{proof}
   Similar to the proof of  \cref{lem:pronyfq}, the algorithm involves solving a linear recurrence (a Hankel system) using \cref{lem:linearrecurrbitpack}, performing an equal-degree factorization using \cref{lem:factorize-bitpack}, computing discrete logarithms by table look-ups (after precomputing the answers in $\tilde O(p) \le \polylog(t)$ time), and finally solving a transposed Vandermonde system using \cref{lem:vandermondebitpack}. The randomization is only used in the factorization step, and similarly to the last paragraph in the proof of \cref{lem:pronyfq}, we can make the failure probability $2^{-(\log t)^{0.1}}$ by blowing up the (worst case) running time by a $(\log t)^{0.1}$ factor.
   The total time is thus worst case $O\big (s\log s + 
   ( \frac{s}{\log t}+ 1 ) \cdot (\log t)^{0.1}\cdot (\log p)^{O(\sqrt{\log s})} \big )$.
\end{proof}

\subsection{Proof of \texorpdfstring{\cref{lem:modularcomposition}}{Lemma~\ref{lem:modularcomposition}}}
\label{sec:proofmodularcomposition}

We assume the  following setup throughout this section: Let $t$ be a parameter.
Assume $d,p\le t^{o(1)}$ and $p$ is prime. 
We assume a machine with word length $w = \Omega(\log t)$ bits and allow $o(t)$ preprocessing to compute look-up tables (we actually only need $t^{0.2+o(1)}$ precomputation time).

We need a few more word RAM primitives.
\begin{lemma}[Transposing a matrix]
\label{lem:packtranspose}
Suppose an  $n\times m$ matrix with $b$-bit elements is stored in row-major order in compact representation. Then we can transpose the matrix (i.e., store it in column-major order) in $O((\frac{nmb}{\log t}+1) \cdot \log n)$ time. Alternatively it can also be done in $O((\frac{nmb}{\log t}+1) \cdot \log m)$ time.
\end{lemma}
\begin{proof}[Proof sketch]
    This is done by a simple divide-and-conquer algorithm with recursion depth $\log n$, similar to \cite[Lemma 9]{journals/jal/Thorup02}.
    The input matrix is represented as $n$ length-$m$ lists, and each of these lists corresponds to a leaf of the recursion tree.
    Then, at each internal node of the recursion tree, we interleave the lists returned by its two child nodes.
    Note that interleaving two lists can be done with time complexity linear in the number of words in their compact representations, using word operations (which can be replaced by table look-ups after preprocessing).

    To get dependency $\log m$ instead of $\log n$, simply change the roles of $n$ and $m$ and do the procedure described above in the reverse order.
\end{proof}

\begin{lemma}[Batch small table look-up]
\label{lem:packlookup}
   Given a compact array $(a_0,\dots,a_{n-1})$ where each $a_i$ has $b$ bits, and a compact list of indices $(i_0,\dots,i_{N-1})\in [n]^{N} $, one can produce the look-up results   $(a_{i_0},a_{i_1},\dots, a_{i_{N-1}})$ in time 
   \[
   O\Big (\Big \lceil \frac{(n+N)(\log n+\log N + b)}{\log t} \Big \rceil\cdot \log N \Big ).\]
\end{lemma}
\begin{proof}
We assume $b+\log n+\log N\le 0.1\log t$; otherwise the brute force algorithm runs in $O((n+N)\lceil \frac{b}{\log t}\rceil)$ time, which already satisfies the claimed bound. 

We perform the following steps:
\begin{itemize}
   \item  We augment each query index $i_j$ into a query pair $(i_j,j)$ (with description length $\log n +\log N + O(1)$ bits), and obtain the compact list of the pairs $(i_0,0),\dots,(i_{N-1},N-1)$.  Then, we sort these query pairs in increasing order of $i_j$ (the first coordinate). 
    
    The initial augmentation step can be done in linear time in the number of resulting words, $O(\lceil \frac{N(\log n+\log N)}{\log t}\rceil )$, as follows:  given $j\in [N]$ and a sequence of $k = \lfloor \frac{0.1\log t}{\log n}\rfloor$ indices $(i'_0,\dots,i'_{k-1})$ (whose total description length is at most $\lceil \log N\rceil + 0.1\log t \le 0.2\log t +O(1)$ bits), we can use table look-ups (after pre-computation in $t^{0.2+o(1)}$ time) 
    to obtain the compact list $((i'_0,j),(i'_1,j+1),(i'_2,j+2),\dots,(i'_{k-1},j+k-1))$. 
    The entire compact list $(i_0,0),\dots,(i_{N-1},N-1)$ can be obtained by repeating this look-up operation and concatenation.

 The sorting step can be done using packed merge sort (e.g., \cite{AlbersH97}): in the merging step of merge sort, we use $O(1)$ word operations (using pre-computed look-up tables) to implement $\Theta(\frac{\log t}{\log n+ \log N})$ pointer advancements (involving $\Theta(\log t)$ bits of data). 
 The time complexity is the number of words times a logarithmic factor (the number of recursive levels in merge sort), which is at most $O(\lceil \frac{N(\log n+\log N)}{\log t}\rceil \cdot \log N)$.

    \item Then, we augment each query pair $(i_j,j)$ into a triple $(i_j,j,0^b)$, where $0^b$ is a $b$-bit placeholder which is supposed to be replaced by the look-up result $a_{i_j}$ later. 

    Similarly, the time complexity is linear in the number of resulting words in the augmented query array, $O(\lceil \frac{N(\log n+\log N + b)}{\log t} \rceil)$.
    \item  Then, we run a two-pointer scan on the data array $(a_0,\dots,a_{n-1})$ and the query array, which can fill in all look-up results in linear time in the total number of words. In more detail: suppose the pointer on the data array currently points to a word $(a_s,a_{s+1},\dots,a_{s + \ell-1})$ (where $\ell \cdot b = \Theta(\log t)$), and the pointer on the query array points to a word $((i'_t,\cdot,\cdot),\dots,(i'_{t+r-1},\cdot,\cdot))$ (where $r\cdot (\log n+\log N + b) = \Theta(\log t)$). If the interval $[i'_t,i'_{t+r-1}]$ has non-empty intersection with $[s,s+\ell-1]$, then we replace all placeholders associated with query indices in $[i'_t,i'_{t+r-1}] \cap [s,s+\ell-1]$ by the corresponding $b$-bit answers, which can be implemented in $O(1)$ time given these two words by table look-up. Then, we advance the pointer on the data array if $s+\ell-1\le i'_{t+r-1}$, or advance the pointer on the query array otherwise.

    The time complexity is linear in the number of words in the data array and the query array, $O(\lceil \frac{N(\log n+\log N + b) + nb}{\log t} \rceil)$.
    \item  Finally, we sort the query array back to its original order in the input (i.e., sorted by the second coordinate in the triple), remove the augmented indices (namely the first and the second coordinates in the triple) while keeping the $b$-bit answer, and repack the answers into a compact list. 

Similarly, this step has time complexity
$O(\lceil \frac{N(\log n+\log N + b)}{\log t}\rceil \log N )$.
\end{itemize}
Thus, the overall time complexity can be loosely upper-bounded by the claimed bound. 
\end{proof}

In \cite[Theorem 3.1]{KedlayaU11}, the modular composition problem is reduced to the Multivariate Multipoint Evaluation (MME) problem.
Then, \cite[Corollary 4.3]{KedlayaU11} gave an efficient algorithm for MME. Here, we need to speed up both the reduction   and  the MME algorithm via bit packing.

\paragraph*{Reduction to MME.}
 The MME problem with parameters $d,m,N$ over $\F_p$ is defined as follows: Given $f(x_0,\dots,x_{m-1})$ in $\F_p[x_0,\dots,x_{m-1}]$ with individual degrees at most $d-1$, and evaluation points $\alpha_0,\dots,\alpha_{N-1}$ in $\F_p^m$, output $f(\alpha_i)$ for $i=0,1,\dots,N-1$.

In our word RAM setting with word length $w=\Omega(\log t)$ bits,
we assume inputs and outputs are arrays of $\F_p$ elements stored in compact representation. 
More precisely, in the MME problem,  the input polynomial $f\in \F_p[x_0,\dots,x_{m-1}]$ 
is described by $d^m$ coefficients in $\F_p$, which are flattened into a one-dimensional sequence where the $x_0^{a_0}x_1^{a_1}\dots x_{m-1}^{a_{m-1}}$-term's coefficient ($0\le a_j \le d-1$) is the $\sum_{j= 0}^{m-1}a_j d^j$-th element, and this one-dimensional sequence is packed into $\Omega(\log t)$-bit words in the usual way.
Each evaluation point in the MME instance is stored as a chunk of $m$ $\F_p$-elements, and the list of evaluation points is represented as the concatenation of these chunks, packed into $\Omega(\log t)$-bit words.

\begin{lemma}[Direct adaptation of {\cite[Theorem 3.1]{KedlayaU11}}]
\label{lem:reductiontomme}
  Given $f(X), g(X),h(X) \in \F_p[X]$ with degrees at most $N-1$ (where $h(X)$ is monic), there is, for every $2\le d_0< N$, an algorithm that outputs $f(g(X))\bmod h(X)$ in 
  \[ O\Big (  \big (1 +   N\cdot \tfrac{\log p}{\log t}\big )\cdot d_0 \polylog(N)\Big )\]
 time plus one invocation of MME (with compact input/output representation)  with parameters $d_0, m = \lceil \log_{d_0} N\rceil, N' = Nm d_0$, provided $N'\le p$.
\end{lemma}
\begin{proof}
   This lemma follows from essentially the same proof as the original \cite[Theorem 3.1]{KedlayaU11}, with every step implemented using bit packing:%
\begin{itemize}
    \item Step 1: Define $f' = \psi_{d_0,m}(f) \in \F_p[Y_0,Y_1,\dots,Y_{m-1}]$, where $\psi_{d_0,m}$ maps each monomial $X^a$ to the monomial $Y_0^{a_0}Y_1^{a_1}\cdots Y_{m-1}^{a_{m-1}}$ where $a$ is written in base-$d_0$ as $a = \sum_{j= 0}^{m-1}a_j d_0^j$.

By our convention of how multivariate polynomials are stored, the memory representation of the result $f'$ is actually the same as the memory representation of the input univariate polynomial $f$, so this step does not require actual computation. %

\item Step 2: Compute $g_{j} := g(X)^{d_0^j}\bmod h(X)$ for all $j=0,1,\dots,{m}-1$.

For each $j$ this takes $O(\log d_0^j)\cdot \tilde O(\lceil N\cdot \frac{\log p}{\log t}\rceil )$ time using \cref{lem:bitpackingprimitive} and repeated squaring. Over all $0\le j< {m}$ it takes $\tilde O(\lceil N\cdot \frac{\log p}{\log t}\rceil {m}^2 \log d_0)$ time in total.

Note that $a(X):= f'(g_0(X),g_1(X),\dots,g_{m-1}(X))$ is a polynomial of degree smaller than $N'$ which satisfies $a(X)\bmod h(X) = f(g(X))\bmod h(X)$. The following steps aim to compute $a(X)$.

\item Step 3: Let $\beta_0,\dots,\beta_{N'-1}$ be distinct points in $\F_p$ (which exist by the assumption $N'\le p$).
Compute $\alpha_{j,k}:= g_j(\beta_k) $ for all $j\in [{m}],k\in [N']$ using fast univariate multipoint evaluation.

By \cref{lem:bitpackingprimitive}, this takes $\tilde O({m}\cdot \lceil N'\cdot \frac{\log p}{\log t}\rceil)$ time in total (recalling $\deg(g_j) < N \le N'$).
The results $\alpha_{j,k}$ are stored as ${m}$ compact arrays each of $N'\cdot \frac{\log p}{\log t}+O(1)$ words.

\item Step 4: Compute $f'(\alpha_{0,k},\alpha_{1,k},\dots,\alpha_{{m}-1,k})$ for all $k\in [N']$.

This involves solving MME for $f'$ on $N'$ evaluation points. 
In order to invoke the MME algorithm, we need to first use \cref{lem:packtranspose} to transpose the representation of $\{\alpha_{j,k}\}_{j\in [{m}],k\in [N']}$ so that the coordinates  $\alpha_{0,k},\alpha_{1,k},\dots,\alpha_{{m}-1,k}$ of each evaluation point are stored consecutively in the same word. By \cref{lem:packtranspose}, the time to do the transposition is 
$O((\frac{mN'\log p}{\log t}+1)\cdot \log m) = \tilde O({m}\cdot \lceil N'\cdot \frac{\log p}{\log t}\rceil)$.

By definition of $\alpha_{j,k}$ and $a(X)$, the answers computed in this step are precisely the values $a(\beta_0),\dots,a(\beta_{N'-1})\in \F_p$.
\item Step 5: Interpolate to recover $a(X)$ (which has degree smaller than $N'$) from the evaluations $a(\beta_0),\dots,a(\beta_{N'-1})$ computed in the previous step.

We use \cref{lem:bitpackingprimitive} to do this in 
$\tilde O(\lceil N'\cdot \frac{\log p}{\log t}\rceil)$ time.

\item Step 6: Output the $a(X)\bmod h(X)$ as the desired answer.

We use \cref{lem:bitpackingprimitive} to do this in 
$\tilde O(\lceil N'\cdot \frac{\log p}{\log t}\rceil)$ time.
\end{itemize} 
Using $ m = \lceil \log_{d_0} N\rceil
$ and $N' = Nmd_0$, the total time complexity can be bounded by
  $ O\Big (  \big (1 +   N\cdot \tfrac{\log p}{\log t}\big )\cdot d_0 \polylog(N)\Big )$ as claimed.
\end{proof}

\paragraph*{The MME algorithm.}

We first describe \cite{KedlayaU11}'s MME algorithm for small prime fields:
\begin{lemma}[Direct adaptation of {\cite[Theorem 4.1]{KedlayaU11}}]
\label{lem:mmebasic}
    Given a  polynomial $f(X_0,\dots,X_{m-1})\in \F_{p_0}[X_0,\dots,X_{m-1}]$ ($p_0 \le t^{o(1)}$ prime) with degree at most $d_0-1$ in each variable, and $\alpha_0,\dots,\alpha_{N-1}\in \F_{p_0}^m$, there is a deterministic algorithm that outputs $f(\alpha_i)$ for $i=0,\dots,N-1$ in 
    \[ O\Big ( \Big \lceil \frac{N+p_0^m + d_0^m}{\log t}\Big \rceil \cdot \poly(m,d_0,p_0,\log N) \Big).\] 
    time.
\end{lemma}
\begin{proof}
We follow the original proof of \cite[Theorem 4.1]{KedlayaU11} with minor modifications. The idea is to first evaluate $f(\alpha)$ on \emph{all} points $\alpha \in \F_{p_0}^m$ via a recursive approach, and then use table look-ups to obtain the desired answers. Now we describe our bit-packed implementation in more detail.

Let $k\ge 0$ be the largest integer such that $k\le m$ and $\max\{p_0^k,d_0^k\}\lceil \log p_0\rceil \le 0.1\log t$. 
Recall that we allow $o(t)$ time pre-computation at the very beginning.  We pre-compute a look-up table
(one for every possible pair of $p_0,d_0\le 0.1 \log t$) that, given a polynomial $g\in \F_{p_0}[X_0,\dots,X_{k-1}]$ with individual degree at most $d_0-1$, returns all evaluations $\{g(\alpha)\}_{\alpha \in \F_{p_0}^k}$. The polynomial $g$ and the list of evaluations $\{g(\alpha)\}_{\alpha \in \F_{p_0}^k}$ have description lengths $d_0^k \lceil \log p_0\rceil, p_0^k \lceil \log p_0\rceil \le 0.1\log t$ bits respectively, and can be packed into a machine word. The total number of possibilities for the polynomial $g$ is $p_0^{d_0^k} \le t^{0.1}$, so the look-up table has at most $t^{0.1}$ entries and can be pre-computed in $\tilde O(t^{0.1})$ time.

Consider the following recursive algorithm which takes as input an $m'$-variate $(m'\ge k)$ polynomial $g\in \F_{p_0}[X_0,\dots,X_{m'-1}]$ with individual degree at most $d_0-1$  (in the outermost level, $g=f$ and $m'=m$), and outputs a compact list of evaluations $g(\alpha)$ for all $\alpha \in \F_{p_0}^{m'}$. For convenience, the output is organized as $p_0^{m'-k}$ machine words, each of which is associated with $\beta\in \F_{p_0}^{m'-k}$ and holds the list of evaluations $\{g(\gamma,\beta)\}_{\gamma\in \F_{p_0}^k}$. 
\begin{itemize}
    \item If $m'= k$, then directly look up the answers from the look-up table in constant time.
    \item If $m'\ge k+1$, then write $g = \sum_{j=0}^{d_0-1} g_j(X_0,\dots,X_{m'-2})\cdot X_{m'-1}^j$. 
    Recursively solve $d_0$ subproblems on the $(m'-1)$-variate polynomials $g_j$.
    Then, use the answers to the subproblems to compute the evaluation results, $g(\alpha) = \sum_{j=0}^{d_0-1} g_j(\alpha_0,\dots,\alpha_{m'-2})\cdot \alpha_{m'-1}^j$.
    More precisely, denoting the result of the $j$-th subproblem ($j\in [d_0]$) as a vector $R_j \in \F_{p_0}^{(p_0^{m'-1})}$, represented as $p_0^{m'-1-k}$ words, we need to output the vector (which is the concatenation of $p_0$ vectors) \[\big (g(\alpha)\big )_{\alpha \in \F_{p_0}^{m'}} =  \Big (\sum_{j=0}^{d_0-1} 0^j R_j,\sum_{j=0}^{d_0-1} 1^j R_j,\dots, \sum_{j=0}^{d_0-1} (p_0-1)^j R_j\Big ). \]
    This involves $p_0d_0$ many additions and scalar multiplications of vectors in $\F_{p_0}^{(p_0^{m'-1})}$, each of which can be implemented in $O(p_0^{m'-1-k})$ time using word operations (after pre-computing the needed look-up tables in $\tilde O((p_0^{p_0^k})^2) = \tilde O(t^{0.2})$ time).

    In total there are $d_0^{m-m'}$ subproblems at level $m'$ in the recursion tree, so the total time complexity is
    \[ \sum_{m'=k}^{m} O(d_0^{m-m'}\cdot p_0d_0 \cdot p_0^{m'-1-k}) = O\big (m\cdot (\max\{p_0,d_0\})^{m+1-k}\big ). \]
    Since $k\le m'-1\le m-1$, the definition of $k$ implies $\max\{p_0,d_0\}^{k+1} \lceil \log p_0\rceil \ge 0.1\log t$, so the time complexity can be bounded by
    \[ O\Big (m\cdot \frac{(\max\{p_0,d_0\})^{m+2}}{\log t}\Big ). \]
\end{itemize}
Having computed $f(\alpha)$ for all $\alpha \in \F_{p_0}^m$, we use the 
 batch table look-up algorithm from  \cref{lem:packlookup} to output the list of $f(\alpha_i), i = 0,1,\dots,N-1$. The parameters for \cref{lem:packlookup} are $b = \lceil \log p_0\rceil, N=N$, and $n = p_0^m$ (we convert $\F_{p_0}^m$ into an integer in $[p_0^m]$ in the natural way), and its time complexity is 
   \[
   O\Big (\Big \lceil \frac{(N+n)(\log n+\log N + b)}{\log t} \Big \rceil\cdot \log N \Big ) = O\Big ( \frac{(N+p_0^m)\cdot \polylog(Np_0^m)}{\log t} + \log N \Big).\]

Therefore, the total time complexity of the two parts can be bounded by 
    \[ O\Big ( \Big \lceil \frac{N+p_0^m + d_0^m}{\log t}\Big \rceil \cdot \poly(m,d_0,p_0,\log N) \Big).\] 
\end{proof}

\begin{lemma}[Direct adaptation of {\cite[Theorem 4.2, Corollary 4.3]{KedlayaU11}}]
\label{lem:solvemme}
MME over $\F_p$ with parameters $d_0,m,N$ can be solved in 
    \[ O\Big (\Big ( \frac{N + d_0^m}{\log t}+ 1\Big ) \cdot \poly\big ((32m\log p)^m d_0 \log N\big )  \Big )\]
   time. 
\end{lemma}
\begin{proof}%
We closely follow  the proof of \cite[Theorem 4.2]{KedlayaU11} and adapt it to our bit-packed setting. The algorithm consists of the following steps:%
    \begin{enumerate}
        \item Lift the polynomial $f\in \F_p[X_0,\dots,X_{m-1}]$ to $\tilde f\in \Z[X_0,\dots,X_{m-1}]$,  and lift the evaluation points $\alpha_i \in \F_p^m$ to $\tilde \alpha_i \in \Z^m$, by identifying $\F_p$ with $\{0,1,\dots,p-1\}\subset \Z$. 
        
        This type conversion does not incur any actual steps in the algorithm.

        Note that $\tilde f(\tilde \alpha_i)\in \Z$ is at most $d_0^m(p-1)^{md_0}$.
        \item  Compute the primes $p_1,\dots,p_k$ less than or equal to \[\ell := 16\log(d_0^m(p-1)^{md_0}),\] and note that $k\le \ell$.  (It is shown in \cite[Lemma 2.4]{KedlayaU11} that the product $p_1\cdots p_k$ is greater than $d_0^m(p-1)^{md_0}$.) %
        \item For $h=1,\dots,k$, compute the reduction $f_h \in \F_{p_h}[X_0,\dots,X_{m-1}]$ of $\tilde f$ modulo $p_h$.  For $h = 1,\dots, k$ and $i=0,\dots,N-1$, compute the reduction $\alpha_{h,i}\in \F_{p_h}^m$ of $\tilde \alpha_i$ modulo $p_h$. 

        Recall that $\tilde f$ and $\{\tilde \alpha_i\}_i$ are stored as compact lists of $\F_p$ elements, and we need to reduce them modulo $p_h$, resulting in lists of $\F_{p_h}$ elements (which are more compactly represented since $p_h$ is much smaller than $p$).
        We implement this step by pre-computing look-up tables (one for each possible $p_h$), so that given $(x_1,\dots,x_s)\in \F_p^s$ (where $s \cdot \log p= \Theta(\log t)$), we can obtain 
        $(x_1\bmod p_h,\dots,x_s\bmod p_h)$ in constant time.
        
        The total time complexity to perform the reductions is linear in the total number of words processed,
        $O(k\cdot \lceil\frac{(d_0^m + Nm)\cdot \log p}{\log t}\rceil )$.

        \item  For $h=1,2,\dots,k$, 
       invoke \cref{lem:mmebasic} to compute $f_h(\alpha_{h,i})$ for all $i\in [N]$. %

The total time complexity of invoking \cref{lem:mmebasic} is

\[ 
\sum_{h=1}^{k} O\Big ( \Big \lceil \frac{N+p_h^m + d_0^m}{\log t}\Big \rceil \cdot \poly(m,d_0,p_h,\log N) \Big).
\]
Using $p_h, k, d_0,m\le \ell$, we can further bound the time complexity by \[O\Big (\Big (\frac{N + \ell^m}{\log t}+1 \Big ) \cdot \poly(\ell \log N)\Big ).\]
        \item For $i\in [N]$, recover $\tilde f(\alpha_i) \in \Z$ as the unique integer in $[p_1p_2\cdots p_k]$ congruent to $f_h(\alpha_{h,i})$ modulo $p_h$ for $h=1,\dots,k$, and return $\tilde f(\alpha_i) \bmod p$ as the answer $f(\alpha_i)\in \F_p$.

        Before doing this, we first need to transpose the results obtained from the previous step, so that $f_1(\alpha_{1,i}),f_2(\alpha_{2,i}),\dots,f_k(\alpha_{k,i})$ are compactly and consecutively stored in memory, using \cref{lem:packtranspose} in time $O((\frac{Nk\log \ell }{\log t}+1)\log k)$. 
        Then, we perform Chinese remaindering for all $i\in [N]$ in a bit-packed fashion using pre-computed look-up tables, so that given a list of $s$ tuples $((x_1^{(1)},\dots,x_k^{(1)}),\dots,(x_1^{(s)},\dots,x_k^{(s)})) \in (\F_{p_1}\times \dots \times \F_{p_{k}})^s$ (where $s\cdot k\cdot \log \ell = \Theta(\log t)$), we can look up the Chinese remaindering results for all $s$ tuples in constant time. The total time for performing the Chinese remaindering is thus $O(\frac{Nk\log \ell}{\log t}+1)$.
    \end{enumerate}
    Using $ k, d_0,m\le \ell$, the total time complexity of all steps can be loosely bounded by 
    \[ O\Big (\Big  ( \frac{N + \ell^m}{\log t} + 1 \Big ) \cdot \poly(\ell \log N)  \Big ).\]
    Since $\ell =  16\log(d_0^m(p-1)^{md_0}) = 16m(\log d_0 + d_0\log(p-1)) \le 32md_0\log p$, we can further bound the time complexity by
    \[ O\Big (\Big ( \frac{N + d_0^m}{\log t}+ 1\Big ) \cdot \poly\big ((32m\log p)^m d_0 \log N\big )  \Big ).\]
        \end{proof}

        \begin{proof}[Proof of \cref{lem:modularcomposition}]
            We first invoke \cref{lem:reductiontomme} with $N=d$ by choosing $d_0 = 2^{\Theta(\sqrt{\log d})}$. Then we reduce to an MME problem over $\F_p$ with parameters $d_0$ and  $m = \lceil \log_{d_0} d\rceil = \Theta(\sqrt{\log d})$ and $N' = Nmd_0 \le O(d\, 2^{O(\sqrt{\log d})})$ (we have $N'\le p$ since $p\ge d^2$ by assumption).
            By \cref{lem:reductiontomme}, the time complexity of this reduction is \[O\Big (\Big (1+\frac{d\log p}{\log t}\Big )\cdot d_0\cdot  \polylog(d)\Big ).\]
            We solve this MME problem using \cref{lem:solvemme} in  time
            \begin{align*}
 &    O\Big (\Big ( \frac{d + d_0^m}{\log t}+ 1\Big ) \cdot \poly\big ((32m\log p)^m d_0 \log d\big )  \Big )\\
  =  \ &   O\Big (\Big ( \frac{d}{\log t}+ 1\Big ) \cdot \poly\big ((32m\log p)^m d_0 \log d\big )  \Big ) \tag{using $d_0^m\le dd_0$}\\
  \tag{\text{using $m= \Theta(\sqrt{\log d})$, $d_0 = 2^{\Theta(\sqrt{\log d})}$, $p\ge d^2$}}
 = \ & O\Big (\Big ( \frac{d}{\log t}+ 1\Big ) \cdot (\log p)^{O(\sqrt{\log d})} \Big ), 
  \end{align*}
 which dominates the total time complexity.
\end{proof}

\subsection{Proof of \texorpdfstring{\cref{lem:modprimerecoverfinitefield}}{Lemma~\ref{lem:modprimerecoverfinitefield}}}
\label{sec:modprimerecoverfinitefield}

We restate \cref{lem:modprimerecoverfinitefield} below for convenience.

\modprimerecoverfinitefield*

The rest of this section proves \cref{lem:modprimerecoverfinitefield}.

We first consider the finite field $\F_p$ case.
Precompute $z^{-1}$ for $z \in \F_p^*$ in $O(p)\le O(t)$ time using \cite[Lemma 4.3]{BringmannFN21}.
Let $\kappa$ be the smallest integer such that $p^\kappa > 2t^{2-\eps}$. Let $q = p^\kappa$, and let $\omega \in \F_q$ be an element with multiplicative order at least $2t^{2-\eps}$. Recalling the discussion in \cref{subsec:generalalgo}, such an element $\omega$ can be found in $o(t)$ time.  Compute $\omega^{-1}$ in $O(\log q)\le O(\log t)$ time and then precompute $\omega^{-i}$ for all $i\in [t]$ in $O(t)$ time.

\begin{claim}
\label{cl:modprimerecoverfinitefield}
    If $\|A \star B - C\|_0 \le 2^{-1.5^\ell} \cdot t / \log^2 t$ for some integer $\ell \ge 0$, then there is a Monte Carlo algorithm that computes $R$ such that 
    $\|A \star B - C - R\|_0 \le 2^{-1.5^{\ell+1}} \cdot t / \log^2 t$ with error probability $\delta'$ in  time 
    \[O\left((\|A\|_0+\|B\|_0+\|C\|_0 + t)\cdot (1 + \log(1/\delta')/1.5^\ell) \right).\]
\end{claim}
\begin{proof}
We repeat the following round $\lceil \log(1/\delta') / (0.5 \cdot 1.5^\ell) \rceil$ times. 

In each round, we pick a random prime $r \in [c_r t / \log t, 2c_r t / \log t]$ for some sufficiently large constant $c_r$ (this can be done in worst-case $O(t)$ time by precomputing all primes in the range and sampling from them), and define $h(x) := x \bmod{r}$. Define $\tilde{A}$  as $\tilde{A}[i] = \omega^i \cdot A[i]$, and similarly define $\tilde{B}$ and $\tilde{C}$. We use FFT to compute $Z := h(A) \star_r h(B) - h(C)$ and $\tilde{Z} := h(\tilde{A}) \star_r h(\tilde{B}) - h(\tilde{C})$ in $O(r \log r) = O(t)$ time. 

Among all rounds, we keep $r, Z, \tilde{Z}$ for which $\|Z\|_0$ is maximized. Then we do the following for this triple $(r, Z, \tilde{Z})$. 

Compute $\omega^r$ in $O(\log r)\le O(\log t)$ time, and then precompute a discrete log hash table, $\text{Log}[\omega^{ri}] = ri$ for all integers $0\le i < 2t^{2-\eps}/r$, in $O(t^{2-\eps}/r) = o(t)$ time. 
For every $i \in [r]$ where $Z[i] \ne 0$, we compute $(\tilde{Z}[i] / Z[i]) \cdot \omega^{-i}$ in $O(1)$ time using the precomputed
$\omega^{-i}$ and $z^{-1}$ for all $z \in \F_p^*$, and then check in $O(1)$ time whether  $(\tilde{Z}[i] / Z[i]) \cdot \omega^{-i}$  can be found in the discrete log table; if so, we add $Z[i]$ to $R[i + \text{Log}[(\tilde{Z}[i] / Z[i]) \cdot \omega^{-i}]]$.  At the end, we update $C \gets C + R$. 

The idea is that, if some $0\le j  <2t^{2-\eps}$ with nonzero $(A \star B - C)[j]$ is in a bucket by itself, then $Z[h(j)] = (A \star B - C)[j]$ and $\tilde{Z}[h(j)] = \omega^{j} \cdot (A \star B - C)[j]$. Therefore, we can find $(\tilde{Z}[h(j)] / Z[h(j)]) \cdot \omega^{-h(j)} = \omega^{j - (j \bmod{r})}$ from the discrete log table, and successfully add $Z[h(j)] = (A \star B - C)[j]$ to $R[(j \bmod{r}) + \text{Log}[\omega^{j - (j \bmod{r})}]] = R[j]$.

We consider the expected size of $\|A \star B - C - R\|_0$ for any particular round. 

For every $j \in \supp(A \star B - C)$, the probability that it collides with any other element with respect to the hash function $h$ is bounded by $O(\|A \star B - C\|_0 \cdot \frac{\log t}{c_r t / \log t}) \le \|A \star B - C\|_0 \cdot \frac{\log^2 t}{4t}$ by picking sufficiently large $c_r$. Therefore, in expectation, the number of non-isolated elements is $\le \|A \star B - C\|_0^2 \cdot \frac{\log^2 t}{4 t} \le 2^{-2 \cdot 1.5^\ell} \cdot \frac{t}{4 \log^2 t}$. Therefore, for any particular round, the probability that the number of non-isolated elements exceeds $2^{-1.5 \cdot 1.5^\ell} \cdot \frac{t}{4 \log^2 t}$ is $\le 2^{-0.5 \cdot 1.5^\ell}$, by Markov's inequality. As we repeated $\lceil \log(1/\delta') / (0.5 \cdot 1.5^\ell) \rceil$ rounds, with $\delta'$ error probability, there exists a round where the number of non-isolated elements is upper bounded by $2^{-1.5 \cdot 1.5^\ell} \cdot \frac{t}{4 \log^2 t}$. When this happens, $\|Z\|_0 \ge \|A \star B - C\|_0 - 2^{-1.5 \cdot 1.5^\ell} \cdot \frac{t}{4 \log^2 t}$. Since we pick the round with the biggest $\|Z\|_0$, the final $\|Z\|_0$ we pick has $\|Z\|_0 \ge \|A \star B - C\|_0 - 2^{-1.5^{\ell+1}} \cdot \frac{t}{4 \log^2 t}$. This in turn implies that the number of non-isolated elements is at most $2 \cdot (\|A \star B - C\|_0 - \|Z\|_0) \le 2^{-1.5^{\ell+1}} \cdot \frac{t}{2 \log^2 t}$ (this can be seen as follows. We first identify one element for each bucket appearing in $\|Z\|_0$, then we add all $\|A \star B - C\|_0 - \|Z\|_0$ non-identified elements one by one. Each new element can make at most one identified element become non-isolated). We will correctly recover all isolated elements; for the remaining elements, we might erroneously put some entry into $R$. Therefore, $\|A \star B - C - R\|_0 \le 2^{-1.5^{\ell+1}} \cdot \frac{t}{\log^2 t}$. 
\end{proof}

Now, it suffices to apply \cref{cl:modprimerecoverfinitefield} with $\ell = 0, 1, \ldots, O(\log \log t)$ and with $\delta' = \Theta(\delta / \log \log t)$ in order to apply a union bound. The overall running time is 
\begin{align*}
& \sum_{\ell=0}^{O(\log \log t)} O\left((\|A\|_0+\|B\|_0+\|C\|_0 + t)\cdot (1 + \log(1/\delta')/1.5^\ell) \right)\\
= & O\left((\|A\|_0+\|B\|_0+\|C\|_0 + t)\cdot (\log \log t + \log(1/\delta')) \right)\\
=& O\left((\|A\|_0+\|B\|_0+\|C\|_0 + t)\cdot (\log \log t + \log(1/\delta)) \right).
\end{align*}
Finally, $\log \log t = O(\log(1/\delta))$ by the assumption $\log t \le 1/\delta$, so we obtain the claimed running time. 

Next, we state the main differences for applying the result to $\Z$. We compute $\tilde{Z}' := h(\partial (A \star B - C))$ instead of $\tilde{Z}$, which can be computed via $h(\partial A) \star_r h(B) + h(A) \star_r h(\partial B) - h(\partial C)$. Then, we add $Z[i]$ to $R[\tilde{Z}'[i] / Z[i]]$ if $\tilde{Z}'[i] / Z[i]$ is an integer within range. The analysis follows similarly.

\subsection{Proof of \texorpdfstring{\cref{lem:overfullbuckets}}{Lemma~\ref{lem:overfullbuckets}}}
\label{sec:overfullbuckets}

We need the following lemma from \cite{BringmannFN21}.
	\begin{lemma}[{\cite[Corollary 11.3]{BringmannFN21}}]
 \label{lem:overfulloriginal}
 Let $X \subseteq [N]$ be a set of $k$ keys. 
 Randomly pick a hash function $h$ from \cref{defn:almostlinearhash} with parameters $p > 4N^2$ and $m \le  N$. Fix a key $x \notin X$ and buckets $a,b \in [m]$. Moreover, let
$ F = \sum_{y \in X} \mathbf{1}[h(y) = b]$.
Then:
\begin{align*}
\Ex[F \mid h(x) = a] &= \Ex[F] = \frac{k}{m} \pm \Theta(1),
\end{align*}
and, for any $\lambda > 0$,
\begin{align*}
\Pr\left[ |F - \Ex(F)| \geq \lambda \sqrt{\Ex(F)} \mid h(x) = a \right] \leq O\left( \frac{N \log N}{\lambda^2 k} \right).
\end{align*}
	\end{lemma}

We restate \cref{lem:overfullbuckets} below.
\lemoverfull*
\begin{proof}
   We closely follow the proof in \cite[Lemma 5.5]{BringmannFN21}.
Since by definition \(\tilde{X} = \pi(X) + \{0, p\}\), we may fix offsets \(o, o' \in \{0, p\}\) so that $y = \pi(x)+o,y'=\pi(x')+o'$, and instead bound the probability of 
\begin{equation}
\label{eqn:overfullnotilde}
\left \lvert \left \{ x \in X :  \sum_{x' \in X} \mathbf{1}  \left[ h(x) + o \equiv h(x') + o' \pmod{m} \right] > \frac{2t}{m} \right \} \right \rvert \le \frac{mN\log^3 N}{4t}.
\end{equation}
Indeed, if the event in \cref{eqn:overfullnotilde} happens for all $o,o'\in \{0,p\}$, then the event in \cref{eqn:overfulltilde} also happens.

Fix \(o, o'\in \{0,p\} \) and fix any \(x \in X\). Then,
\begin{align}
 \ &  \Pr\left[\sum_{x' \in X} \mathbf{1} \left[ h(x) + o \equiv h(x') + o' \pmod{m} \right] > \tfrac{2t}{m}\right] \nonumber \\
 =  \ &  \sum_{a \in [m]} \Pr[h(x) = a] \cdot \Pr\left[\sum_{x' \in X} \mathbf{1}\left[h(x') \equiv (a + o - o') \pmod{m}\right] > \tfrac{2t}{m} \middle| \,  h(x) = a \right]. \label{eqn:tempineq}
\end{align}
To bound the conditional probability in \cref{eqn:tempineq},  apply \cref{lem:overfulloriginal} to the set $X\setminus \{x\}$ and the  buckets \(a\) and \(b = (a+o-o') \mod m\): Let \(F' = \sum_{x' \in X \setminus \{x\}} \mathbf{1}[h(x') = b]\), and then \cref{lem:overfulloriginal} states that \(\Ex[F'] = |X|/m + O(1)\), and that  (here we are using $|X|\le t$)
\begin{align*}
  \Pr\left[ F' > \tfrac{2t}{m} \, \middle|\, h(x) = a \right]  & \le \Pr\left[ | F' - \Ex[F'] | >  \tfrac{t}{m} - O(1) \, \middle|\,  h(x) = a \right]\\
  & \le O(\tfrac{mN\log N}{t^2}),
\end{align*}
where we set $\lambda = \sqrt{t^2/(m|X|)}- O(1)$ in \cref{lem:overfulloriginal}.
Plugging back into \cref{eqn:tempineq} (and noting that the contribution of $x$ to the sum is at most $1$), we get
\[
   \Pr\left[\sum_{x' \in X} \mathbf{1} \left[ h(x) + o \equiv h(x') + o' \pmod{m} \right] > \tfrac{2t}{m}\right]  \le O(\tfrac{mN\log N}{t^2})\]
   for fixed $x,o,o'$.
Now, by summing over $x\in X$, we know the expected value of the left hand side of \cref{eqn:overfullnotilde} is $|X|\cdot O(\frac{mN\log N}{t^2}) =O(\frac{mN\log N}{t}) $. Then, by Markov's inequality, we know the event in \cref{eqn:overfullnotilde} happens with at least $1-O(1/\log^2 N)$ probability (for fixed $o,o'$). Then, by a union bound, we know the event in \cref{eqn:overfullnotilde} happens for all $o,o'\in \{0,p\}$ with at least $1-O(1/\log^2 N)$ probability.
Since the event in \cref{eqn:overfullnotilde} for all  $o,o'\in \{0,p\}$ implies the event in \cref{eqn:overfulltilde}, this finishes the proof.
\end{proof}

\end{document}